\documentclass{article} 
\pdfoutput=1
\usepackage{color}
\usepackage{amsmath}
\usepackage{theorem}
\usepackage{proof}
\usepackage{macros} 
\usepackage{graphicx}
\usepackage{xspace}
\usepackage[all]{xy}
\usepackage{qtree}
\usepackage{alltt}
\usepackage{examplep}
\usepackage{epigraph}
\usepackage{MnSymbol}
\usepackage{mathtools}
 
\begin{document}

\title{Polyregular functions}
\author{Miko{\l}aj Boja\'nczyk\footnote{
	Supported by the European Research Council under the European Union’s Horizon 2020 research and innovation programme (ERC consolidator grant LIPA, agreement no. 683080). 
}}
\maketitle   
\begin{abstract}  This paper is about  certain string-to-string functions,  called the \emph{polyregular functions}. These are like the regular string-to-string functions, except that they can have polynomial (and not just linear) growth.   The class has four equivalent definitions: 
\begin{enumerate}
	\item deterministic two-way transducers with pebbles;
	\item the smallest class of string-to-string functions that is closed under composition, contains all sequential functions as well as:
	\begin{align*}
\underbrace{1|23|456|78 \quad \mapsto \quad 1|32|654|87}_{\text{iterated reverse}} \qquad \underbrace{1234 \quad \mapsto \quad \underline 1234 1\underline 234 12\underline 34 123 \underline4}_{\text{squaring}};
\end{align*}
	\item a fragment of the $\lambda$-calculus, which has a list type constructor and limited forms of iteration such as $\mapterm$ but not {\tt fold};
	\item an imperative programming language, which has {\tt for} loops that  range over input positions.
\end{enumerate}
The first definition comes from~\cite{milo2003typechecking}, while the remaining three are new to the author's best knowledge. 
The class of polyregular functions  contains  known classes of string-to-string transducers, such as the sequential, rational, or regular ones, but goes beyond them because of super-linear growth. Polyregular functions have  good algorithmic properties, such as:
\begin{enumerate}
	\item the output can be computed in linear time (in terms of combined input and output size);
	\item the inverse image of a regular word language is (effectively) regular.
\end{enumerate}
We also identify a fragment of polyregular functions, called the \emph{first-order polyregular functions}, which has additional good properties, e.g. the output can be computed by an \aczero circuit.
	
\end{abstract}

\pagebreak
\tableofcontents
\pagebreak

\setcounter{section}{-1}

\section{Introduction}
\bigskip
\begin{quote}
    \emph{The author (along with many other people) has come recently to the conclusion that the functions computed by the various machines are more important---or at least more basic---than the sets accepted by these devices.} \hfill Dana Scott~\cite{scott1967some}\footnote{I got this quote from Wolfgang Thomas, who got it from Boris Trakhtenbrot~\cite[p.~14]{trakhtenbrot2008logic}.}    
\end{quote}
\bigskip

This paper is about string-to-string functions that are  defined by finite-state devices. There are three main classes of string-to-string functions\footnote{For more on sequential, rational and regular string-to-string functions, including additional references, see the book~\cite{sakarovitch2009elements}, the survey paper~\cite{filiot2016transducers}, or~\cite[Sections 12, 13]{toolbox}.}:
\begin{enumerate}
    \item {\bf Sequential functions.} The sequential functions  are the ones recognised by deterministic finite automata with transitions  labelled by output words (all states should be accepting if we care about total functions). Here is an example, which recognises the function that doubles every $a$, and appends $\#$ in case the input has odd length:
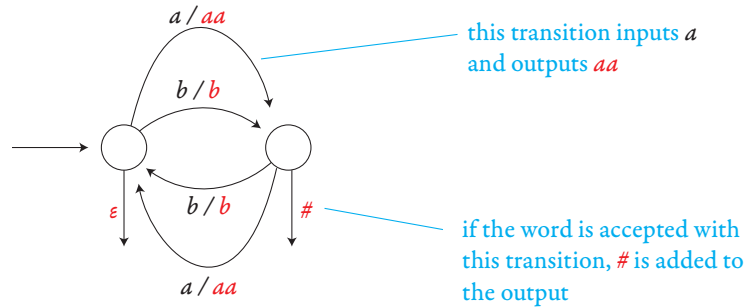
\mypic{38}
    \item {\bf Rational functions.} Rational functions are defined  like sequential functions, except that the underlying automaton is no longer required to be deterministic, but only \emph{unambiguous}, which means that for every input word it has at most one accepting run (and exactly one accepting run is we care about total functions).  Here is an example automaton (with two connected components), which recognises the function that doubles every $a$, and prepends $\#$ in case the input has odd length:
  \mypic{37}
  Apart from the above description, which originates from~\cite[Chapter IX]{eilenberg1974automata}, there are other equivalent descriptions: regular expressions which use pairs of strings~\cite[Section IV.1]{sakarovitch2009elements}, Eilenberg bimachines~\cite[Chapter XI.7]{eilenberg1974automata}, and  unary queries of \mso with associated outputs (see Definition~\ref{def:rational-function} later in the paper). 
    \item {\bf Regular functions.} A \emph{regular string-to-string function} is defined to be one that is recognised by a deterministic two-way automaton with output~\cite{aho1970characterization}. Equivalent models include string-to-string \mso transductions~\cite{engelfriet2001mso}, streaming string-transducers~\cite{alur2010expressiveness}, and various formalisms that use combinators~\cite{alur2014regular,DBLP:conf/lics/DaveGK18,DBLP:conf/lics/BojanczykDK18}. 
\end{enumerate}
\begin{figure}
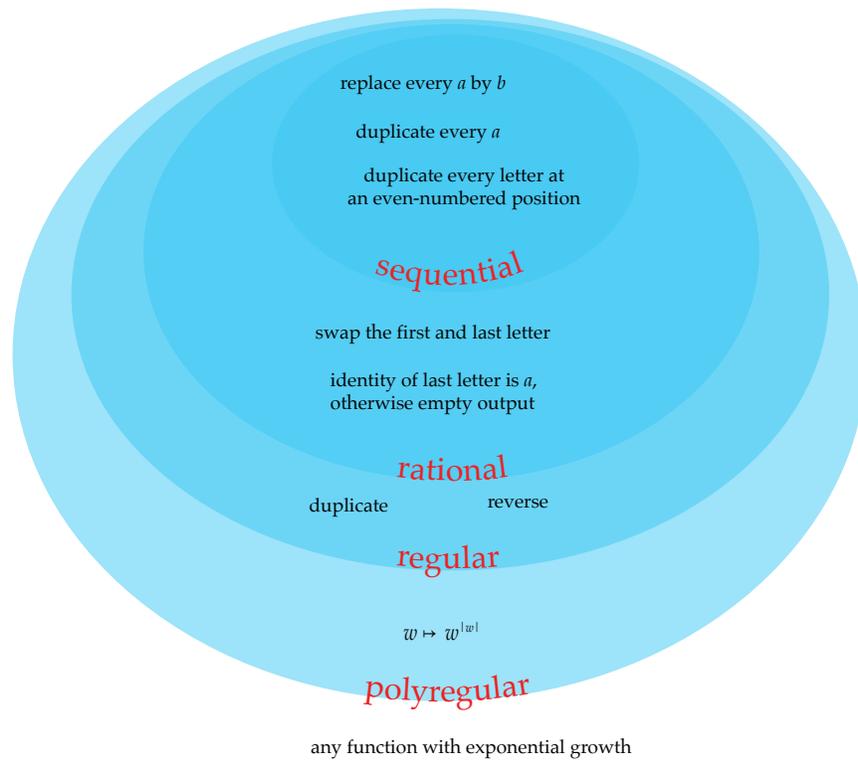
 
    \mypic{36}    
    \caption{\label{sec:three-classes} The sequential, rational and regular string-to-string functions.}
\end{figure}

Not only are the three above classes  robust (i.e.~they have multiple equivalent definitions, using machines, logic, or expressions), but the associated models  have the good decidability  properties  typical for  finite automata, as illustrated by the following results.  One can minimise automata for sequential~\cite[Main Theorem]{choffrut1979generalization}, see also~\cite[Section 3]{choffrut2003minimizing}, likewise for  and rational functions, see~\cite[Section 4]{reutenauer1991minimization} and~\cite[Section 3.3]{filiot2016first}. One can decide if a rational function is already sequential~\cite[Corollaire 3.5]{choffrut1977caracterisation} and one can decide if a regular function is already rational, see~\cite[Theorem 1]{DBLP:conf/lics/FiliotGRS13} and~\cite[Theorem 4]{baschenis2015one}. Equivalence is decidable for sequential and rational functions (which follows from minimisation), and also for regular functions see~\cite[Theorem 1]{gurari1982equivalence} and~\cite[Theorem 12]{alur2011streaming}.  

The contribution of this paper is a proposal for  fourth class in the list:
\begin{enumerate}
    \item[4.]  {\bf Polyregular functions.} These are the string-to-string functions recognised by pebble automata, which were  introduced by Globerman and Harel as acceptors~\cite{Globerman:1996he} and by Milo, Suciu and Vianu as transducers~\cite{milo2003typechecking}. The class has also three other equivalent descriptions; the models in these equivalent descriptions and their equivalence are  the contribution of this paper. 
\end{enumerate}
One of the distinguishing properties of polyregular functions, and also the reason for the ``poly'' in the name, is that the output size is polynomial in the input size, as opposed to the linear bounds that hold for the sequential, rational and regular functions. 

\subsection{An example: all prefixes in reverse order}
\label{sec:prefixes-eample}
We begin by illustrating the equivalent definitions of polyregular functions (pebble transducers as well as the three new models that are equivalent to them) on a running example. Precise definitions are given in Part I of the paper, the equivalence of the definitions is proved in Part II, while Part III discusses algorithmic questions. 

The running example  is  the function
\begin{align*}
    f : \set{a,b}^* \to \set{a,b,|}^*
\end{align*}
which maps a word to the reverses of all prefixes, separated by $|$, as in  this example
\begin{align*}
    babaaa \quad \mapsto \quad  b|ab|bab|abab|aabab|aaabab| 
\end{align*} 
The size of the output is quadratic in the size of the input, which means that $f$  is not recognised by a deterministic two-way automaton with output. In other words, $f$ is not regular (and therefore it is also neither rational nor sequential). The function $f$  is, however polyregular, as demonstrated by the following descriptions, which correspond to the four equivalent models discussed in this paper.
\begin{enumerate}
    \item {\bf Polyregular Functions, see Section~\ref{sec:polyregular}}. The first definition   is that the polyregular functions are the compositions  of certain atomic operations. For the running example $f$, the composition uses four steps, illustrated below for the input word 
    \begin{align*}
        babaaa
    \end{align*}
    \begin{enumerate}
        \item Append a separator symbol $|$  giving this result:
        \begin{align*}
            babaaa|
        \end{align*}
        \item Take the result of the first step, and for each position $x$, produce a copy of the word, with the position $x$ underlined, giving this result:
        \begin{align*}
            \underline babaaa|b\underline abaaa|ba\underline baaa|bab\underline aaa|baba\underline aa|babaa\underline a|babaaa\underline |
        \end{align*}
        \item Remove the last block between separators $|$. For the remaining blocks, keep only the positions before and including the underlined position, and finally remove the underlines, yielding this result:
        \begin{align*}
            b|ba|bab|baba|babaa|babaaa|
        \end{align*} 
        \item Reverse each word between  separators $|$,  yielding this result 
        \begin{align*}
            b|ab|bab|abab|aabab|aaabab|
        \end{align*} 
       
    \end{enumerate}
    The operations in steps (a) and (c) are rational functions, while the operations in steps (b) and (d) are not (we call these operation \emph{squaring} and \emph{iterated reverse}, respectively). The class of polyregular functions is defined to be the closure under composition of the  rational functions, the squaring function, and iterated reverse.
    \item {\bf Pebble transducers, see  Section~\ref{sec:pebble-transducers}}. The second description of the polyregular functions uses pebble transducers. The idea is to have an automaton which runs on the input word, and uses pebbles to mark positions. The pebbles are organised in a stack, of height fixed by the syntax of the automaton, and only the topmost pebble in the stack can be moved. To recognise the function $f$ from the running example, we use a  pebble automaton with two pebbles: a main pebble (first pebble on the stack), and a secondary pebble (second pebble on the stack). The main pebble runs through all input positions in left-to-right order. For each position, the secondary pebble is used to copy the input word from main pebble down to the beginning of the word. Here is a picture of the run of this automaton.
    \mypic{40}  
    
    \item {\bf For-transducers, see  Section~\ref{sec:for-programs}}. The third description uses programs which have variables that range over positions in the input word, as given in the following example.
    \begin{verbatim}
        for x in first..last 
          for y in last..first  
            if x =< y and a(x) then output a
            if x =< y and b(x) then output b
          output |
    \end{verbatim}
    The input positions can be compared for order, and their labels in the input word can be tested. The output word is produced by instructions of the form {\tt output a}. The programming language also allows Boolean variables, which are useful to simulate the control of a finite state automaton.  The programs are easily seen to be a special case of pebble automata. The opposite inclusion is also true, but harder to show, because the loops in a for-transducer can only move first-to-last or last-to-first, while the head in a pebble automaton can alternate between left and right moves an unbounded number of times.
    
    \item {\bf Polynomial List Functions, see  Section~\ref{sec:list-programs}}. The final description of the polyregular functions uses a functional programming language. Define $\splitterm$ to be the function which inputs a list and outputs all possible ways of splitting it into two parts, as illustrated on the following example
\begin{eqnarray*}
	& [1,2,3,4] \\ &  \downmapsto \\ &[([1,2,3,4],[]),([1,2,3],[4]), ([1,2],[3,4]), ([1],[2,3,4]), ([],[1,2,3,4])]	
\end{eqnarray*} 
Let {\tt reverse}  be the function which reverses a list, and finally let {\tt fst} be the function which projects a pair to its first coordinate. The function $f$ from the running example is then defined using the  following Haskell code
    \begin{verbatim}
        \x -> (map (\y -> reverse (fst y)) (split x)).
    \end{verbatim}
    Given  an input list {\tt x}, the above program first applies the $\splitterm$ operation. In the resulting list of pairs of lists, one keeps only the reverse of the first coordinate of every pair.
    
    The general  idea behind the fourth equivalent description of the polyregular functions is to use functional programs without recursion, which are equipped with certain atomic string manipulators, like $\reverseterm$, and some higher-order combinators, like $\mapterm$.

\end{enumerate}

\paragraph*{Structure of the paper.} The paper has three parts.

Part I  introduces the four equivalent models which describe the polyregular functions. 

Part II proves that the models described in Part I are equivalent. The main insights are that: (a) a pebble transducer can  implement $\beta$-reduction and therefore evaluate $\lambda$-terms; and (b) results from semigroup theory such as the Krohn-Rhodes Theorem and Simon's Factorisation Forest Theorem can be used to decompose the computation of a pebble transducer in a way that can be then simulated by very limited string-to-string transformations.

Part III discusses algorithms for evaluating polyregular functions. The first result is that polyregular functions can be evaluated in linear time, in terms of the combined input and output size. This result uses constant delay enumeration algorithms for first-order queries on strings~\cite{Kazana:2013jq}. The second result is that first-order definable polyregular functions can be computed by \aczero circuits, as long as the circuits can use an $\varepsilon$ letter which is ignored when producing the output string.

\paragraph*{Future work.} We are missing a logical characterisation of the polyregular functions, and a streaming (one way) machine model. Such characterisations are left for future work. A natural candidate for the logical characterisation is string-to-string \mso interpretations, i.e.~an extension of \mso transductions~\cite[Section 7]{courcelle2012graph} where a tuple of input positions can be used to represent a single output position. There are also several algorithmic questions left for future work, including: (a) is equivalence decidable for polyregular functions?; and (b) can one decide if a polyregular function is already regular?

\paragraph*{Acknowledgements.} I would like to thank the following people for many helpful discussions: Jacek Chrz{a}szcz, Amina Doumane, Sandra Kiefer, Bartek Klin, Anca Muscholl, Nathan Lhote, Aleksy Schubert, Helmut Seidl, Mahsa Shirmohammadi, Pawe{\l} Urzyczyn,  Igor Walukiewicz,  Daria Walukiewicz, James Worrell

\pagebreak 
\part{Description of the models}
In this part, we introduce the four models which describe polyregular functions. Their equivalence will be proved in Part II.

\section{Polyregular functions}
\label{sec:polyregular}
The first definition of the  polyregular functions is that these are  finite compositions of atomic functions that are either: a sequential function, or two string operations called squaring and iterated reverse. The design objectives for the definition of polyregular functions are:
\begin{itemize}
	\item the class is closed under composition by definition;
	\item the atomic functions are as simple as possible.
\end{itemize}
The minimality of the atomic functions will make it easy to evaluate the polyregular functions, or to prove that the preimage of a regular language is always regular. The  minimality  will also make the formalism cumbersome to use, which is why the polyregular functions can be seen as a sort of assembly language, as opposed to more user-friendly languages  defined in Sections~\ref{sec:pebble-transducers},~\ref{sec:for-programs} and~\ref{sec:list-programs}.

We begin by recalling in more detail the definition of sequential functions, and introducing the squaring and iterated reverse operations.

\paragraph*{Sequential functions.} A sequential function is a string-to-string function that arises from a deterministic automaton with outputs on  transitions. The idea is that the automaton process the input word from left to right, and the output is produced during this run based only finite state control.  The syntax  is given in the following definition.

\begin{definition}[Sequential function]\label{def:sequential-function}
	The syntax of a  sequential function consists of:
	\begin{enumerate}
		\item input and output alphabets $\Sigma$ and $\Gamma$;
		\item a deterministic finite automaton with input alphabet $\Sigma$;
		\item for each transition in the automaton, an associated output in $\Gamma^*$;
		\item for each state in the automaton, an associated  end-of-input word in  $ \Gamma^*$.
	\end{enumerate}
	The semantics of a sequential is a function $\Sigma^* \to \Gamma^*$  defined as follows. Given an input word $w \in \Sigma^*$, one runs the underlying automaton. When executing a transition, the associated label defined in item 3 is produced. At the end of the run, the output is extended by the end-of-input word associated to the last state reached by the automaton. 
\end{definition}

Apart from sequential functions, the polyregular functions use also two string-to-string operations called squaring and iterated reverse, which are described below.

\paragraph{Squaring.} The squaring operation on strings is illustrated in the following example:
\begin{eqnarray*}
1234 \qquad \mapsto \qquad   \underline 1234    1\underline 234   12\underline 34   123\underline 4  
\end{eqnarray*}
For each position $x$ in the input word, we produce a copy of the input with $x$ underlined, and then we concatenate all these copies in left-to-right  order of the underlined positions.  If the input word has size $n$, then its   square  has length $n^2$, which explains the name of the operation. Formally speaking, squaring is a family of operations, with one squaring operation for every choice of  input alphabet. The output alphabet for the squaring operation is two copies of the input alphabet: the underlined and the non-underlined letters. 

\paragraph*{Iterated reverse.} The iterated reverse operation takes a string, which contains occurrences of a separator symbol, and reverses each block between consecutive separators, but keeps the order of the blocks as it was in the input word, as illustrated in the following example:
\begin{align*}
  123|45|678|9 \qquad \mapsto \qquad 321|54|876|9
\end{align*}
More formally, iterated reverse is not a single operation, but a family of operations, with one iterated reverse operation for every choice of input alphabet and designated separator symbol.

\paragraph*{Polyregular functions.} We are now ready to give the first definition of the class of polyregular functions.
\begin{definition}[Polyregular functions] The class of \emph{polyregular functions} is the smallest class of string-to-string functions which is closed under composition of functions, and contains:
		\begin{enumerate}
		\item sequential functions; 
		\item squaring; 
		\item iterated reverse.
	\end{enumerate}
\end{definition}

The ``regular'' in the name polyregular refers to the fact that polyregular functions extend regular functions, i.e.~those that are recognised by two-way deterministic automata with output (this fact is most apparent with the definition from Section~\ref{sec:pebble-automata} that uses pebble automata).  The ``poly'' in the name polyregular stands for polynomial, because the output of a polyregular is polynomial (possibly super-linear, unlike  sequential  functions which are linear) in the size of the input.

\paragraph*{The first-order case.}
	 We pay particular attention to   the subclass of first-order definable languages, and  the corresponding functions. For readers unfamiliar with logic as a means of defining regular languages, a good place to start is~\cite{thomas1997languages}. A \emph{first-order definable} language $L \subseteq \Sigma^*$, see~\cite[Section 2.2]{thomas1997languages} is one that can be defined by a formula of first-order logic, which quantifies over positions in the input word, has a binary predicate $x<y$ for order and unary predicates $a(x),b(x),\ldots$ for testing labels of positions. 
	 For example, the formula
\begin{align*}
\forall x \ b(x) \ \Rightarrow \ \exists y\ x < y \land a(y)
\end{align*}
says that every position with label $b$ is followed (not necessarily in the successor position) by a position with label $a$. If the alphabet is $\set{a,b}$, then the formula defines the regular language $(a+b)^*a$.  Every first-order sentence defines a regular language of words, but not all regular languages can defined this way, because in general, set quantification of \mso is needed, see~\cite[Section 4.1]{thomas1997languages}. 

A theorem of McNaughton, Papert and Sch\"utzenberger\footnote{See~\cite[Section 6]{straubing-siglog}  for a more  in-depth  discussion of this result and its history.}, says that the  first-order languages can be characterized in terms of the automata that recognize them. Call a deterministic finite automaton \emph{counter-free} if its transition monoid is aperiodic, which means that for every input word $w \in \Sigma^*$ the following sequence of states is ultimately constant (i.e.~from some point on it has only one state appearing in the sequence):
	\begin{align*}
  qw, qw^2, qw^3,\ldots
\end{align*}
In other words, an automaton is counter-free if it does not have a pattern like this
\mypic{14}
The McNaughton, Papert and Sch\"utzenberger Theorem says  that 	
 an automaton is counter-free if and only if for every state $q$, the set of words which reach state $q$ is definable in first-order logic. 

\begin{definition}[First-order sequential and polyregular functions]
	A sequential function is called \emph{first-order sequential} if it is recognized by a sequential transducer where  the underlying automaton, i.e.~the automaton in item 2 of Definition~\ref{def:sequential-function}, is counter-free. The \emph{first-order polyregular functions} are  the special case of polyregular functions  where the sequential functions are required to be first-order sequential.
\end{definition}

\subsection{Equivalent definitions}
One of the building blocks in the class of polyregular functions is the class of sequential functions. This building block could be replaced by other types of functions, without affecting the expressive power of the class, because closure under function composition ensures  robustness of the model. We give below two alternatives for sequential functions as building blocks in the polyregular functions, one less expressive and one more expressive, and discuss why -- in the presence of closure under composition -- the alternatives lead  the class  polyregular functions.

\paragraph*{Krohn-Rhodes.} Although we claimed that the atomic functions in the definition of polyregular functions are minimal, this is not really the case, because sequential functions can be further decomposed. 
The Krohn-Rhodes Theorem, see~\cite[Corollary 4.1]{krohn1965algebraic} or~\cite[Appendix A]{straubing2012finite},  says that every first-order sequential function is equal to a composition of finitely many first-order sequential functions where the underlying automaton has two states. For general (not necessarily first-order) sequential functions, one also needs sequential transducers where the underlying automaton is a group, in the sense that each input letter induces a permutation on its states.  Therefore, one can replace sequential functions by the more basic building blocks from the Krohn-Rhodes Theorem, as stated in the following theorem.

\begin{theorem}\label{thm:krohn-rhodes}
	The class of polyregular functions is equal to the smallest class of string-functions which is closed under composition, and contains 
	\begin{enumerate}
		\item sequential functions recognized by:
\begin{enumerate}
	\item two state counter-free automata; or
	\item automata where every input letter acts as a permutation on the states. 
\end{enumerate}
		\item squaring; 
		\item iterated reverse.
	\end{enumerate}
The first-order polyregular functions are the special case where 1(b) is not used.
\end{theorem}

\paragraph*{Rational functions.} Sequential functions are a left-to-right model. A more expressive, and symmetric, model is the rational functions. There are several ways to define rational functions, including deterministic one way automata with lookahead   or nondeterministic unambiguous one way automata with output, as discussed in the introduction. We use a definition of a logical character,  which is convenient for describing the first-order fragment of rational functions.  In the following, by  a \emph{unary query} we mean an \mso formula (which includes  the special case of a first-order formula)  with one free first-order variable. We apply unary queries to words, and therefore a unary query can be viewed as defining a property of pairs (input word, distinguished position in the input word). For example, a unary query can say ``position $x$ has label $a$, and all later positions have label $b$''.

\begin{definition}[Rational function]\label{def:rational-function}
	A rational function is given by:
	\begin{enumerate}
		\item input and output alphabets $\Sigma$ and $\Gamma$;
		\item a finite set $\Ff$ of unary \mso queries over the input alphabet $\Sigma$, such that for every word in $\Sigma^*$ and distinguished position, exactly one query from $\Ff$ is true;
		\item \label{it:rational-output-words} for each query from $\Ff$, an associated output in $\Gamma^*$;
		\item an \emph{output word for empty  input} in $\Gamma^*$.
	\end{enumerate}
		The semantics is  is a function $f : \Sigma^* \to \Gamma^*$
defined as follows. Suppose that $a_1 \cdots a_n \in \Sigma^*$ is an input word. If $n=0$, then the output is the word from item 4. Otherwise, the output is the word $w_1 \cdots w_n$, where $w_i$ is the word associated in item 3 to the unique query from $\Ff$ which selects position $i$ in the input word. 

A first-order rational function is the special case when all queries in $\Ff$ are defined in  first-order logic, i.e.~set quantification is disallowed.
\end{definition}

The definition above is not in the same spirit as the definition of sequential functions from Definition~\ref{def:sequential-function}. An equivalent definition of  the rational functions would be to use unambiguous automata with output, or Eilenberg bimachines, and the first-order subclass would be recovered by considering an aperiodic restriction on the machines, see~\cite[Theorem 3.1]{lautemann2001descriptive}.

\begin{myexample}
Consider the function 
\begin{align*}
f : \set{a,b}^* \to \set{a,b}^*	 \qquad f(w)= \begin{cases}
	w & \text{if $w$ ends with $a$}\\
	\varepsilon & \text{otherwise.}
\end{cases}
\end{align*}
Consider the \mso (in fact, first-order) sentence
\begin{align*}
\varphi = \forall x \exists y \ y \ge x \land a(y)	
\end{align*}
which says that the last position has label $a$. The set of $\Ff$ of queries in item 2 of Definition~\ref{def:rational-function} has three queries
\begin{align*}
\underbrace{a(x) \land \varphi}_{\alpha(x)} \qquad 	 \underbrace{b(x) \land \varphi}_{\beta(x)} \qquad  	\underbrace{\neg  \varphi}_{\gamma(x)}
\end{align*}
where the third query $\gamma(x)$ does not depend on the position $x$. The outputs from item 3 of the definition are defined by
\begin{align*}
\alpha(x) \mapsto a \qquad \beta(x) \mapsto b 	 \qquad \gamma(x) \mapsto \varepsilon,
\end{align*}
while the output for empty input from item 4 is defined to be $\varepsilon$.
\end{myexample}

As shown by Elgot and Mezei, see~\cite[Theorem 7.8]{elgot1965relations}, a function is rational if and only if it can  be decomposed as a sequential function followed by a reverse sequential function (i.e.~reverse, then a sequential function, then reverse again). Since polyregular functions have reverse built in, we get the following result.

\begin{theorem}\label{thm:rational-polyregular} 
If in the definition of polyregular functions, one replaces sequential functions  by rational functions, the resulting class is the same. Likewise for the first-order case.
\end{theorem}

In this paper we  use rational functions more often than sequential ones, and hence the definition of polyregular functions that uses rational functions would be more in the spirit of the technical development below.

\subsection{Regularity preservation.} One advantage of the definition of polyregular functions in terms of composing atomic functions is that if we want to prove that a property is true for all polyregular functions, and that property is preserved under function composition, then it is enough to prove the property for the atomic functions. Here is an example. 
\begin{theorem}[Regular preimages]\label{thm:regular-continuous}
	If $f : \Sigma^* \to \Gamma^*$ is polyregular and $L \subseteq \Gamma^*$ is regular, then the preimage $f^{-1}(L)$ is regular. Furthermore, if $f$ is first-order polyregular and $L$ is first-order definable, then $f^{-1}(L)$ is first-order definable.	
\end{theorem}
\begin{proof}
It is enough to prove the theorem when $f$ is an atomic function, since the property in the statement of the theorem is preserved under function composition. The case for sequential functions is easy to see, see e.g.~\cite[Corollary 4.2]{berstel2013transductions} for a stronger result, which also covers the rational functions. 

It remains to deal with iterated reverse and squaring. We only do the case of squaring, the iterated reverse is handled in a similar way. Consider an alphabet $\Sigma$ and the squaring operation
\begin{align*}
\text{square} : \Sigma^* \to (\Sigma + \underline \Sigma)^*	
\end{align*}
We want to show that for every regular language $L$ over the output alphabet of $f$, the inverse image $\text{square}^{-1}(L)$ is also regular. Suppose that $L$ is such a language, which is  recognized by a monoid homomorphism 
\begin{align*}
h : (\Sigma + \underline \Sigma)^*	 \to M.
\end{align*}
where the monoid $M$ is finite\footnote{For monoids and homomorphisms as an alternative to recognising regular languages, see~\cite[Chapter V]{straubing2012finite}}.
Define \begin{align*}
g : \Sigma^* \to M^*	
\end{align*}
to be the function which inputs a word $w$, and replaces each position $x$ in that letter by the value of $h$ on the word obtained from $w$ by underlining position $x$. Here is a picture:
\mypic{15}
It is not hard to see that the following diagram commutes
\begin{align*}
\xymatrix@C=4cm{ \Sigma^* \ar[r]^{\text{square}} \ar[d]_g & 	(\Sigma + \underline \Sigma)^* \ar[d]^h \\ M^* \ar[r]_{\text{product in $M$}} & M}
\end{align*}
The language $\text{square}^{-1}(L)$ is the inverse image, under the function $h \circ \text{square}$, of some accepting set of elements $F \subseteq M$. Because the diagram commutes, $\text{square}^{-1}(L)$ is also equal to  the inverse image under $g$ of the language
\begin{align*}
K = \set{v \in M^* : \text{the product of $v$ is in $F$}}.
\end{align*}
The function $g$ is  rational and the language $K$ is regular, and therefore the inverse image $g^{-1}(K)$ is regular, as we have discussed at the beginning of this proof. Furthermore, if $L$ is first-order definable, then the function $g$ is first-order rational and the language $K$ is first-order definable, and therefore $g^{-1}(K)$ is also first-order definable. 
\end{proof}

The construction in the above theorem is effective. A corollary  is that one can effectively check, given a polyregular function, if some  output is nonempty  (because this is the same as checking if the inverse image of the language $\set{\varepsilon}$ does not contain  all input words), or if the function outputs only words of even length. 
Another corollary  is that  a language $L \subseteq \Sigma^*$   is regular if and only if its characteristic function 
\begin{align*}
w \in \Sigma^* \quad \mapsto \quad \begin{cases}
	1 & \text{if }w \in L\\
	0 & \text{if }w \not \in L
\end{cases}
\end{align*}
is a polyregular function. In other words, for functions with Boolean outputs, the polyregular functions are the same as the regular languages. 

\begin{myexample}
	If a function is  polyregular then it has  (a) polynomial size increase; and (b) preimages of regular languages are regular. One could ask if these properties characterise the polyregular functions, i.e.~if the polyregular functions are exactly the string-to-string functions that have properties (a) and (b). Here is a counterexample. For a function  $f  : \Nat \to \Nat$,  define
	\begin{align*}
		\hat f : a^* \to a^* \qquad a^n \mapsto a^{f(n)!}.
	\end{align*}
	If $f$ grows slow enough, e.g.~if it is $\log \log n$, then  $\hat f$ has polynomial size increase, i.e.~it satisfies condition (a). We claim that if  $f$ tends to infinity, then $\hat f$ also satisfies (b). This is  because it produces only words of factorial length, i.e.~words of the form
	\begin{align*}
		a^{1} \quad a^{2!} \quad a^{3!} \quad \cdots,
	\end{align*} Indeed, it is well known that every regular language $L \subseteq a^*$  contains either finitely many, of co-finitely many words of factorial length. In particular, if $f$ tends to infinity, then the  preimage under $\hat f$ of every regular language is going to be either finite or co-finite, and therefore also regular. Summing up, if $f$ tends to infinity slowly enough, then $\hat f$ is going to satisfy conditions (a) and (b). If $f$ is hard enough, e.g.~it is not computable, then $\hat f$ is not polyregular.
\end{myexample}

\section{Pebble transducers}
\label{sec:pebble-transducers}
In this section, we describe a second model for string-to-string functions, which will end up being equivalent to the polyregular functions from the previous section. The model is defined in terms of  automata and transducers which have a two-way head and  use pebbles  to mark positions in an input word. For languages, pebble automata where introduced in Globerman and Harel~\cite{Globerman:1996he}, while the transducer version -- in the more general setting of trees -- comes from  Milo, Suciu and Vianu as transducers~\cite{milo2003typechecking}. Without any restriction on the way that pebbles are placed, one can use pebbles  to  simulate  logarithmic space Turing machine computation~\cite[Corollary 3.5]{Ibarra:1971ij}, and the model has many undecidable properties, e.g.~the following problem is undecidable ``decide if a pebble automaton accepts at least one input word''. That is why one considers pebbles with   a stack discipline~\cite[Definition 4.1]{Globerman:1996he}:  pebbles are totally ordered, and a pebble can  be moved only if all  pebbles smaller in the order have been lifted.

Although  we are mainly interested in pebble automata not as acceptors, but as string-to-string transducers, we begin our presentation with automata (i.e.~acceptors) in Section~\ref{sec:pebble-automata},  and only then extend the definition to transducers in Section~\ref{sec:pebble-transducers-definition}.

\subsection{Pebble automata}
\label{sec:pebble-automata}
The idea behind an $k$-pebble automaton is that at any given moment of its computation, it stores a stack of at most $k$ positions (called pebbles) in the input word. The stack discipline condition says that only the topmost position in the stack -- called the \emph{head} -- can be modified, by moving it left or right\footnote{Sometimes, the head is counted separately from the pebbles, e.g.~one talks about an automaton with one pebble and one head. In this paper, the head is counted as one of the pebbles, namely the topmost one.}. Also, the  automaton can pop the topmost stack position, or push a new pebble (which initially is equal to the head) Here is a picture  of a configuration which has four pebbles:
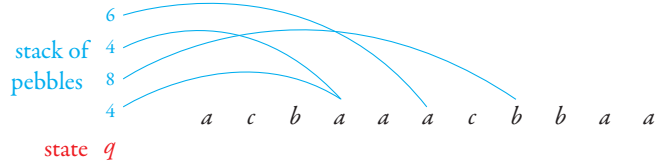
\mypic{16} 
A 1-pebble automaton is the same thing as a two-way automaton, since it can only move its head and push/pop are disallowed.  We do not decorate the input word with end-markers, and therefore a pebble automaton can be run only on a nonempty word.  The exact syntax of a pebble automaton  is described in the following definition.

\newcommand{\wrapends}[1]{\vdash \! #1 \! \dashv}

\begin{definition}[Pebble automaton]\label{def:pebble-transducer} 
	A $k$-pebble automaton consists of:
	\begin{enumerate}
	\item a finite \emph{input alphabet} $\Sigma$;
		\item a finite   set $Q$ of \emph{states};
		\item two designated states: an initial and final one;
		\item a transition function of type
\begin{align*}
	 \underbrace{Q}_{\substack{\text{current} \\ \text{state}}} \times \underbrace{\Sigma}_{\substack{\text{label} \\ \text{under} \\ \text{head}}} \times
	\underbrace{(\set{\text{first, last, $1,\ldots, k$}}^2\to \set{\le , \not \le})}_{\substack{\text{the order comparison of the pebbles} \\ \text{and the first and last positions}}}   \to   \underbrace{Q}_{\substack{\text{new} \\ \text{state}}} \times \underbrace{\set{-1,0,1,\text{push, pop}}}_{\substack{\text{pebble action}}}.
  \end{align*}
\end{enumerate}
\end{definition}
A pebble automaton is  a $k$-pebble automaton for some $k$.
A configuration of the automaton consists of: (a) a nonempty \emph{input word} over the input alphabet; (b) a \emph{control state} from the set of states; and (c) a \emph{pebbling} which is a nonempty stack of at most $k$  positions in the input word. The $i$-th position in the pebbling is called the $i$-th pebble, with pebble 1 being the  bottom of the stack. The topmost position in the stack  is called the  \emph{head}. For a configuration $c$, its \emph{successor configuration}, which might be undefined, is the configuration with the same input word obtained as follows. Apply the  transition function of the automaton  in the natural way to $c$, yielding a new state $q$ and pebble action $a$.  The  control state in the successor configuration is $q$ and the pebbling is updated as follows:
\begin{itemize}
	\item If the pebble action $a$ is in $\set{-1,0,1}$ then  the topmost position on the stack is offset by $a$, i.e.~the head moves by $a$, as in the following picture
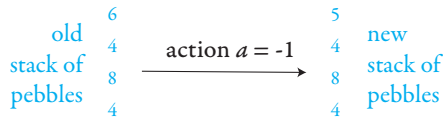
\mypic{17}
	 The successor configuration is undefined if adding $a$ yields something that is not a position, i.e.~the automaton moves left from the first position or right from the last position.
	\item If the pebble action is ``push'', then the topmost position on the stack is duplicated, as in the following picture
	\mypic{18}
 If ``push'' is executed when all $k$ pebbles are already present, then the successor configuration is undefined.
	\item If the pebble action is ``pop'',  then the topmost pebble on the stack is removed, and therefore the head is moved to the second-to-last pebble, as in the following picture
	\mypic{19} 
	If the ``pop'' action is executed when there is only one pebble, then the successor configuration is undefined. 
\end{itemize}

A \emph{run} of the pebble automaton is a sequence of configurations where consecutive configurations are connected by the successor relation on configurations described above. For a given input word, the \emph{initial configuration}  of the pebble automaton  has the initial state, and the pebbling has one pebble which points to the first position. The automaton accepts a word if there is an \emph{accepting run}, i.e.~a run where the first configuration is initial, the last one has an accepting state, and no other configurations have an accepting state. The accepting run, if it exists, is unique, by determinism of the transition function.

\paragraph*{Defining the reachability relation in logic.}
We are interested in pebble automata where the reachability (not successor) relation on configurations can be defined in first-order logic, according to the following definition.
\begin{definition}[First-order definable pebble automaton]\label{def:fo-definable-pebble-automaton}
A $k$-pebble automaton is called \emph{first-order definable} if for every states $p$ and $q$ and numbers $i,j \in \set{1,\ldots,k}$  there is a first-order formula 
\begin{align*}
\varphi^{ij}_{pq}(\underbrace{x_1,\ldots,x_i}_{\bar x}, \underbrace{y_1,\ldots,y_j}_{\bar y})	
\end{align*}
such that for every nonempty word $w$ over the input alphabet,
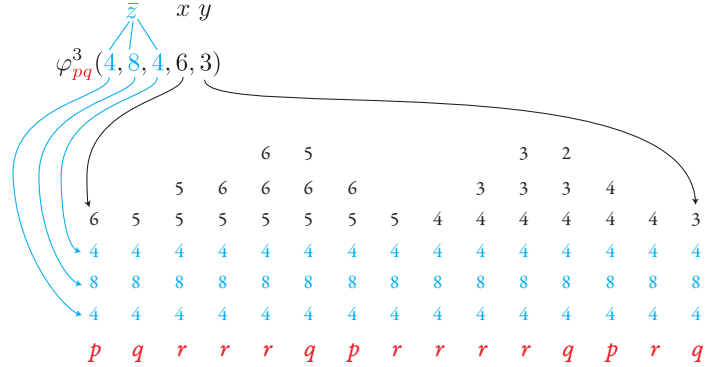
\begin{align*}
w \models \varphi^{ij}_{pq}(\bar x, \bar y)	
\end{align*}
holds if and only  the automaton has a run from the configuration with state $p$ and pebbling $\bar x$ to the configuration with state $q$ and  pebbling  $\bar y$. 
\end{definition}

One could also consider a variant of the above definition -- call it \mso definability --  where \mso is used instead of first-order logic. It turns out that \mso definable reachability is automatically true.

\begin{lemma}\label{lem:all-pebble-mso}
Every pebble automaton is \mso definable.
\end{lemma}
\begin{proof}
This result is implicit in~\cite[Theorem 4.2]{Globerman:1996he}. Call a run \emph{$i$-superficial} if the source and target configurations have a stack of  height $i+1$ and pebbles $1,\ldots,i$ are not moved through the run. Here is a picture of a 3-superficial run:
\mypic{22}
  The general idea is that an $i$-superficial run can be decomposed as a concatenation of several $(i+1)$-superficial runs, and this concatenation can be simulated in \mso, since it corresponds to taking a transitive closure of a binary relation on positions (and not tuples of positions). A more formal proof is given below.
  
Suppose that the automaton has $k$ pebbles.
 
\begin{claim}
For every states $p$ and $q$ and $0 \le i < k$  there is an \mso formula 
\begin{align*}
\psi^{i}_{pq}(\underbrace{z_1,\ldots,z_i}_{\bar z},x,y)	
\end{align*}
which is true in a nonempty word $w$ over the input alphabet if and only the automaton admits an $i$-superficial run  state $p$ and pebbling $\bar zx$ to state $q$ and pebbling $\bar zy$.
\end{claim} 
 
\begin{proof}
Here is a picture of the run in the statement of the claim (which shows only the pebblings and states, since the input word remains the same):
\mypic{20}
The proof is by induction on $i$ in the opposite order, i.e.~the induction base is when $i=n-1$.   The induction base is proved the same way as the induction step, so prove only the induction step. Suppose that we haver already proved the claim for $i+1$, and we want to prove it for $i$. Using the induction assumption, we can write for every states $p$ and $q$ and  \mso formula 
\begin{align*}
\theta_{pq}(z_1,\ldots,z_{i},x,y)	
\end{align*}
which is describes runs as in the statement of the claim, but with the added requirement that for all configurations used in the run except for the last one, the first $i+1$ pebbles are the same, namely $\bar zx$. In particular, the last transition moves the $(i+1)$-st pebble by some offset in $\set{-1,0,1}$, as shown in the following picture:
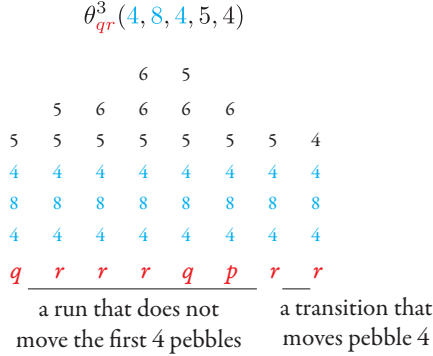
\mypic{21}
  The formula $\theta$ can be written using the induction assumption because the runs it describes consist of: a run that can be described using the induction assumption, and then a single transition. 

To get the  formula $\psi^i_{pq}$ from the statement of the claim, we use \mso to do the following fixpoint computation. Fix an input word and a  valuation of the variables $\bar z$. Consider the least set 
\begin{align*}
X \subseteq \text{states} \times \text{positions}
\end{align*}
which contains the pair $(q,x)$ and which has the following closure property:
\begin{align*}
\bigwedge_{r,s \in \text{states}} \qquad \forall u \forall v \qquad  	 (r,u) \in X \ \land \  \theta_{rs}(\bar z,u,v) \ \Rightarrow \ (s,v) \in X
\end{align*}
The formula $\psi^i_{pq}$ from the statement of the claim then simply says that $(q,y)$ belongs to $X$. This reasoning can be formalised in \mso, by encoding $X$ a tuple of sets of positions, one for each control state of the automaton.
\end{proof}

The above claim shows that \mso can define reachability relation for $i$-superficial runs. To define the reachability relation in general, we observe that an arbitrary run can be decomposed into a concatenation of at most $2k$ superficial runs separated by push/pop transitions, as in the following picture:
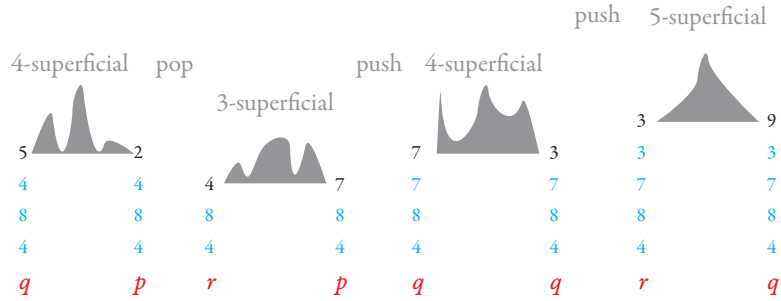
\mypic{23}
\end{proof}

Motivated by Lemma~\ref{lem:all-pebble-mso}, we do not discuss \mso definable pebble automata.

\subsection{Pebble transducers}
\label{sec:pebble-transducers-definition}
We are mainly interested in pebble automata as string-to-string functions (i.e.~transducers), rather than as acceptors. To get a string-to-string function, a pebble automaton is extended with output words, as described in the following definition.
\begin{definition}[Pebble transducer]
A  \emph{$k$-pebble transducer} is defined to be a $k$-pebble automaton, plus:
\begin{enumerate}
	\item a finite set $\Gamma$ called the \emph{output alphabet};
	\item an \emph{output function} $Q \to \Gamma^*$.
	\item an \emph{output word for empty input} in $\Gamma^*$.\end{enumerate}
A pebble transducers is a $k$-pebble transducer for some $k$. 
 	A pebble transducer is called first-order definable if its underlying pebble automaton is such.
\end{definition}
If a pebble transducer has input alphabet $\Sigma$ and output alphabet $\Gamma$, then its semantics is a partial function $\Sigma^* \to \Gamma^*$ defined as follows. If the input word is empty, then the output is the empty output word from item 3 in the definition.
If the input word is nonempty, then consider the accepting  run of the underlying pebble automaton on the input. If this accepting run does not exist (i.e.~the underlying pebble automaton rejects), then the output is undefined. If  the accepting run exists, then apply the output function from item 2 in the definition to (the state in) each configuration, and concatenate the resulting words. This concatenation is the output of the pebble transducer.  We are interested in pebble transducers that define total functions, i.e.~the underlying pebble automaton accepts all nonempty inputs.

Note that a 1-pebble transducer is the same thing as a deterministic two-way automaton with output; the class of functions recognised by such devices coincides with string-to-string \mso transductions~\cite[Theorem 31]{engelfriet2001mso} and with the class of streaming string transducers~\cite[Theorems 1, 2, 3]{alur2010expressiveness}. In particular, 1-pebble transducers are closed under composition, because this is true for \mso transductions. Closure under composition also extends to pebble transducers with more pebbles, as shown  by Engelfriet and Maneth in \cite[Theorem 2]{engelfriet2002two}, see also~\cite[Theorem 11]{engelfriet2015two}:
\begin{theorem}\label{thm:pebble-closed-under-composition}
	The class of functions recognized by pebble transducers is closed under composition. Likewise for first-order definable pebble transducers.
\end{theorem}
Since our syntax for pebble transducers is a little different than in~\cite{engelfriet2002two,engelfriet2015two}, and since we also need the closure under composition of first-order definable pebble transducers, we present a self-contained proof.

\begin{proof}[Of Theorem~\ref{thm:pebble-closed-under-composition}]
	Consider two pebble transducers
	\begin{align*}
\xymatrix@C=3cm{\Sigma^* \ar[r]^f_{\text{$k$ pebbles}} & \Gamma^* \ar[r]^g_{\text{$m$ pebbles}} & \Delta^*}	
\end{align*}
 We assume without loss of generality that the transducer for $f$ satisfies the following conditions
 \begin{itemize}
 \item[(a)] each transition produces at most one output letter; and
 	\item[(b)]  when a pop action is executed, then the head does not move, i.e.~pop is only allowed when the topmost two pebbles are in the same place.
 \end{itemize}
Every pebble transducer can be easily converted into one that satisfies (a) and (b) and produces the same outputs.

We  prove below that the composition $g \circ f$ can be computed by a pebble transducer.
Consider an input word $w \in \Sigma^*$. In the proof we  talk about configurations of $f$ in $w$ and about configurations of $g$ in $f(w)$; we write $f$-configurations for the former and $g$-configurations for the latter. 

 The idea is as follows.  
For an input word $w \in \Sigma^*$, define a  \emph{composite configuration in $w$} to be  a state of $g$ plus a stack of at most $m$ $f$-configurations in the word $w$. Since all of the $f$-configurations have the same input word, the stack stores  only the states and the pebble positions. A composite configuration can be represented using at most $km$ pebbles. 
 For an input word $w \in \Sigma^*$, and a $g$-configuration
 $c$ over input $f(w)$, define its \emph{composite} to be the result of replacing in $c$ each pebble by the $f$-configuration which produced that pebble's position, as in the following picture \mypic{24}

The automaton for the composition $g \circ f$  simulates the run of $g$ on $f(w)$, by computing the composites of the $g$-configurations on that run. It remains to show that the composite configurations can be updated, i.e.~if we know the composite of a $g$-configuration $c$, then we can compute the composite of its successor $g$-configuration.  To perform these updates, we use the toolkit of operations described in the following sublemma.

\begin{sublemma}
	The following operations on  composite configurations  can be performed by a pebble automaton:
 \begin{enumerate}
  \item pop the topmost $f$-configuration;
 \item duplicate the topmost $f$-configuration;
 	\item replace the topmost $f$-configuration   by its successor $f$-configuration;
 	\item replace the topmost $f$-configuration   by its predecessor $f$-configuration;
 \end{enumerate}
\end{sublemma}
\begin{proof}
\begin{enumerate}
	\item 	Pop the pebbles at the top of the composite stack. 
	\item For duplication, suppose that the topmost $f$-configuration has $i \in \set{1,\ldots,k}$ pebbles.  The simulating automaton does $i$ left-to-right passes through the input word. In the $j$-th pass it waits until it sees the $j$-th pebble from the topmost configuration and then pushes a new copy of that pebble.
	\item For the successor  of the topmost configuration, use the automaton for $f$.
	\item The predecessor configuration is the hard part. There is a smart solution to this problem that avoids using extra pebbles, and which is based on an idea  of Hopcroft and Ullman~\cite{ullman1967approach}, see also~\cite[Lemma 13.4]{toolbox}. In the 
	interest of simplicity, we present a less smart idea which uses $k$ extra pebbles. Suppose that the composite configuration is $C$, and let $c$ be the $f$-configuration that is at the top of the stack in $C$. Our goal is to produce the predecessor of $c$. Recall the assumption (a) described at the beginning of this proof, which says that a pop transition is only allowed when the two topmost pebbles are in the same position. A corollary of this assumption is that if we know  
	$f$-configuration $c$ and a transition $t$ of the automaton in $f$, then there is a unique $f$-configuration that goes to $c$ via $t$. Therefore,  it is enough to find the transition that was executed by $f$ when entering configuration $c$. To find this transition,  one simply restarts the  automaton $f$ from the initial configuration, using at most $k$ extra pebbles, until the $f$-configuration $c$ is reached. 
	This subcomputation allows us to determine the transition that was executed by $f$ just before entering configuration $c$.
\end{enumerate}
\end{proof}

Using the above toolkit, we can compute the successor on composite configurations, with item 4 used whenever the automaton $g$ wants to move its head to the previous position. It is not hard to see that if both $f$ and $g$ were first-order definable, then the same is true for the composite automaton described in the above construction.
\end{proof}

From Theorem~\ref{thm:pebble-closed-under-composition}, we get the following corollary, which says that the polyregular functions, as defined in Section~\ref{sec:polyregular}, are contained in the functions recognised by pebble transducers (we will also see later on that the opposite inclusion is true as well).

\begin{corollary}
	Every polyregular function is recognised by a pebble transducer.
\end{corollary}
\begin{proof}
	Since pebble transducers are closed under composition, it is enough to show that the basic building blocks of polyregular functions, namely the sequential functions, squaring, and reverse are all recognised by pebble transducers. This is easy to check. A one-way automaton with output, i.e.~a sequential function, is a special case of a two-way automaton with output (which is the same as a 1-pebble transducer), and iterated reverse can easily be implemented using a two-way automaton with output.  The only place where more than one pebble is needed is to implement squaring; this requires two pebbles. Note that although each building block of polyregular functions requires at most two pebbles, the compositions of these building blocks require more pebbles, because the composition of pebble transducers with $k$ and $m$ pebbles is implemented by a pebble transducer with at least $km$ pebbles. 
	
	In general, an unbounded number of pebbles is necessary to capture all polyregular functions. This is because for a $k$-pebble automaton, if the input has length  $n$, then the output has length at most $O(n^k)$, since the number of configurations is at most $|Q| \cdot n^k$.  On the other hand, there are polyregular functions (squaring iterated multiple times) where the output size is given by a polynomial of arbitrarily high degree.	
\end{proof}

%

\section{For-transducers}
\label{sec:for-programs}
The third description of polyregular functions uses a type of imperative programming language.
A \emph{for-transducers} is a type of program where loops to range over positions of the input word. The power of the programming language is  constrained to the point that it can be simulated by a pebble automaton. 
We begin with  an example for-transducer that computes the squaring function 
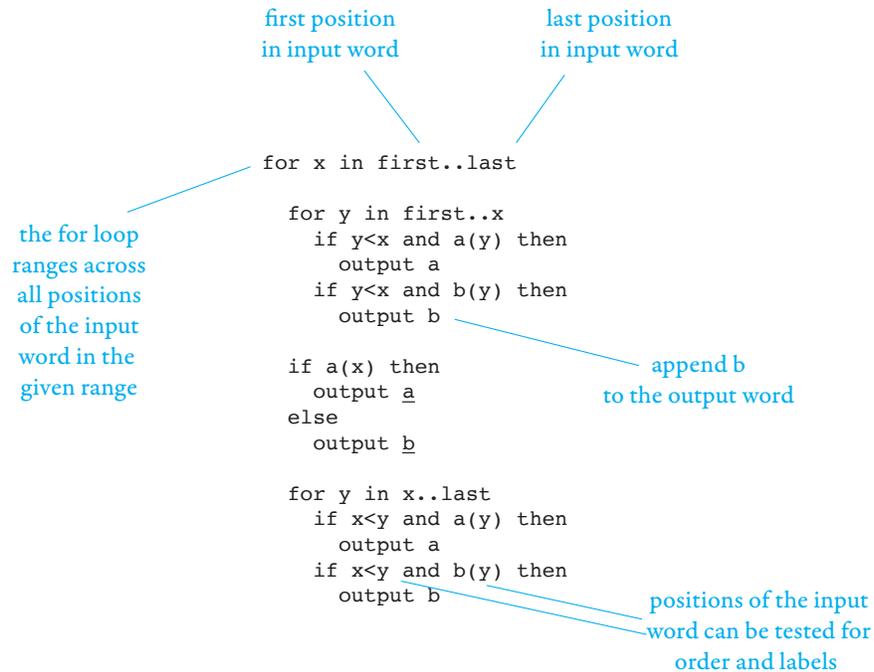
\begin{align*}
\set{\mathtt a,\mathtt b}^* \to \set{\mathtt a,\mathtt b, \underline {\mathtt a}, \underline {\mathtt b} }^*
\end{align*}
\mypic{26}   
The above program illustrates almost all programming constructs  allowed in for-transducers: a {\tt for} loop ranging over positions in the input word, an {\tt if} conditional, {\tt a(x)} for checking if position {\tt x}  in the input word has label {\tt a}, and an instruction {\tt output a} which appends {\tt a} to the output word. (There is one more feature, namely Boolean variables, which will be explained below.) 

The for loop is of the form 
\begin{verbatim}
    for x in y..z
\end{verbatim} 
where {\tt x} is the iterating  variable which  is introduced and  bound by the loop, and each of {\tt y} and {\tt z} is either some variable that has been bound  in a containing loop, or one of the keywords {\tt first} or {\tt last}\xspace  representing the first and last positions in the input word. The body of the  loop is  executed for all positions  {\tt x} in the interval from {\tt y} to {\tt z} including both endpoints, in order that is either  increasing or decreasing order, depending on whether {\tt y} is smaller or bigger than  {\tt z}. Here is an example illustrating the order of positions
\mypic{27}

The final feature of the language is Boolean variables, which are illustrated in the following program, which  computes the parity of length of the input word:
\mypic{28}

In general, a for-transducer uses  two types of variables: \emph{position variables}, which range over positions of the input word, and \emph{Boolean variables} such as {\tt odd} above, which can have value {\tt true} or {\tt false}. Position variables are introduced in {\tt for} loops and one cannot use  assignment {\tt :=} to change the value of a position variable.  Boolean variables can be declared in any block, e.g.~in the scope of a {\tt for} loop, and assignments for   Boolean variables are allowed, i.e.,~one can write {\tt b := true} or {\tt b := false}.

This completes the description of for-transducers. It is straightforward to see that every string-to-string function  recognized by for-transducer is also recognized by a pebble transducer; we will discuss this more detail in Section~\ref{sec:for-to-pebble}. The converse  is also true and will follow from the results in Part II. The difficulty in transforming a pebble transducer into a for-transducer is that the loops in a for-transducer have a fixed direction (they sweep the input either from left to right or from right to left), while the head of a pebble transducer can move both left and right. 

\begin{myexample}
        \label{ex:for-sequential}
        Every sequential function is recognised by a for-transducer. The for-transducer does only one for loop of type
        \begin{center}
            {\tt for x in first..last}
        \end{center}
        to scan through all positions in the input word. 
        Boolean variables are used to maintain the state of the automaton underlying the sequential transducer.  To simulate rational functions, one would use two nested loops.
\end{myexample}
     
\paragraph*{First-order definable for-transducers.}
There is also  a fragment of for-transducers which corresponds to the first-order restriction. Recall that all Boolean variables are initialized to {\tt false}.   A for-transducer is called \emph{first-order} if Boolean variables can only go from {\tt false} to {\tt true}, but not back. In other words,  the only allowed update for Boolean variables is {\tt q := true}. For the first-order restriction, it is important that Boolean variables can be  declared inside for loops, and that they are reinitialized to {\tt false} at each iteration of the loop that they are declared in. This is illustrated in the following example. 
\mypic{29}

\begin{myexample}\label{ex:for-transducers-do-formulas} In Example~\ref{ex:for-sequential} we showed how to simulate a sequential function using a for-transducer. If we want to simulate a first-order sequential function using a first-order definable for-transducer, we use a different approach, where the number of nested for loops will correspond to the quantifier depth of the formula. This approach is illustrated below for languages, i.e.~functions with yes/no outputs.

    Consider a formula of first-order logic defining a language of words. Here is a corresponding for-transducer, which outputs 1 on words in the language and 0 for words outside the language. The for-transducer is obtained using a straightforward implementation of quantifiers using for loops; and the construction is linear in the size of the formula. For example, if the formula is
    \begin{align*}
        \forall x\ \exists y\ a(x) \Rightarrow x<y \land b(x)
    \end{align*}
    then the corresponding for-transducer looks like this: 
    \mypic{39}
For the above program, it is important that the Boolean variable {\tt atleastoney} gets reinitialized at each iteration of the {\tt for} loop that binds variable {\tt x}. 
Note that the transducer needs only first-to-last loops; such loops would no longer be sufficient for transducers (as opposed to yes/no formulas). 
\end{myexample}
\section{Polynomial list functions}
\label{sec:list-programs}
We now present the fourth definition of the polyregular functions, which uses a functional programming language. 
Roughly speaking, a polynomial list function is functional program that uses a list data type,  some (higher-order) atomic functions for list manipulation like
\begin{align*}
  \mapterm : (\tau \to \sigma) \to \tau^* \to \sigma^*
\end{align*}
and which has no recursion. To model functional programs, we use the $\lambda$-calculus. As  programming constructs we allow $\lambda$-abstraction and application, but no general recursion mechanisms (apart from those implicit in the atomic functions like $\mapterm$). The language is defined so that when restricted to string-to-string functions, the polynomial list functions have  the same expressive power as the polyregular functions. In particular, the polynomial list functions have outputs of at most polynomial size.

 Before giving the formal definition of  polynomial list functions, we begin with some examples that illustrate the programming constructs and atomic operations that are allowed.

\begin{myexample}[List duplication]\label{ex:flatterm}
	Suppose that $\tau$ is some type, e.g.
\begin{align*}
	\tau= \set{1,2,3,4,5,6}	
\end{align*}	
One of the atomic functions is 
\begin{align*}
	\flatterm^\tau : (\tau^*)^* \to \tau^* 
\end{align*}
which flattens the input list as illustrated in the following example
\begin{align*}
[[1,2,3],[4,5],[],[6]] \quad \mapsto \quad [1,2,3,4,5,6].	
\end{align*} 
Using $\flatterm$, we can write a function
\begin{align*}
\lambda x : \tau^* \ . \ \flatterm_\tau \ [x,x]	
\end{align*}
which duplicates the input list as in the following example
\begin{align*}
[1,2,3,4] \quad \mapsto \quad [1,2,3,4,1,2,3,4]	
\end{align*}
\end{myexample}

\begin{myexample}[Squaring]\label{ex:mapterm} 
Another atomic function is 
\begin{align*}
	\mapterm^{\tau \sigma} : (\tau \to \sigma) \to \tau^* \to \sigma^* \qquad \text{for every types }\tau,\ \sigma
\end{align*}
which applies the function in the first argument to every element of the list in the second argument.  
Using the  function described above, we can write the following program 
\begin{align*}
\lambda x : \tau^*\ .\  \mapterm_{\tau \tau^*}\  (\lambda y : \tau\ .\  x)\ x,
\end{align*}
which has type $\tau^* \to (\tau^*)^*$, and which replaces each element of the input list by  the list itself, as illustrated in the following example
\begin{align*}
[1,2,3,4] \quad \mapsto \quad [[1,2,3,4],[1,2,3,4],[1,2,3,4],[1,2,3,4]].
\end{align*}
\end{myexample}

\begin{myexample}[Squaring, continued]\label{ex:splitterm}
	The operation in Example~\ref{ex:mapterm} is similar to the squaring operation described in Section~\ref{sec:polyregular}; but it is weaker in the sense that it does not have the underlines which distinguish consecutive copies of the input list. To recover the squaring function, we use a function 
	\begin{align*}
	\splitterm^\tau : \tau^* \to (\tau^* \times  \tau^*)^* \qquad \text{for every type }\tau
\end{align*}
which splits the input list in all possible ways, as illustrated below:
\begin{eqnarray*}
	& [1,2,3,4] \\ &  \downmapsto \\ &[([1,2,3,4],[]),([1,2,3],[4]), ([1,2],[3,4]), ([1],[2,3,4]), ([],[1,2,3,4])]	
\end{eqnarray*} 
If the input list has length $n$, then the output contains $n+1$ pairs, with the $i$-th pair for $i \in \set{0,1,\ldots,n}$ having the first $n-i$ elements on the first coordinate, and the remaining $i$ elements on the second coordinate. 
\end{myexample}

\paragraph*{Syntax and semantics}
We now give a more formal description of the syntax and semantics of polynomial list functions. We begin by describing the types in the language and their associated semantic domains.
\label{sec:syntax-and-semantics-of-list-programs}
\begin{definition}[Types and their domains] Every finite set is a type\footnote{An alternative, more minimalistic, approach would be to have only one atomic type, namely the empty set $\emptyset$.  Then $\emptyset^*$ would be a type which has only one element $[]$, and larger finite sets could  be defined using disjoint union, e.g.~a three element type could be encoded as 
\begin{align*}
	(\emptyset^* + \emptyset^*) + \emptyset^*.
\end{align*}
Since such a representation of finite sets would be cumbersome to use, we choose to define the type system so that it has finite sets as atomic types.}, and if  $\tau, \sigma$ are types, then so are:
		\begin{align*}
\tau^*\qquad \tau+\sigma \qquad \tau \times \sigma \qquad \tau \to \sigma
\end{align*}
For a type $\tau$, its associated domain $\sem \tau$ is defined as follows: 
	\begin{itemize}
		\item if $\tau$ is a finite set, then $\sem \tau = \tau$;
		\item $\sem{\tau \to \sigma}$ is the set of all total functions from $\sem \tau$ to $\sem \sigma$;
		\item $\sem{ \tau \times \sigma}$ is the product $\sem \tau \times \sem \sigma$, likewise for disjoint union $+$;
		\item $\sem {\tau^*}$ is the set of all finite lists of elements in $\sem \tau$.	\end{itemize}
\end{definition}

There is not much difference between $\tau$ and $\sem \tau$, except that the former can be viewed as a syntactic expression of finite size, while the latter is a typically infinite set. An important issue is that the functions in $\sem {\tau \to \sigma}$ are total, because we only consider always terminating programs. 

\begin{definition}\label{def:set-of-terms}
	Assume some set\footnote{We gloss over the fact that there is no ``set of all finite sets''. A more formal approach would require distinguishing only some representative family of finite sets that forms a set.
} of variables, each variable having an associated type, such that for every type there are infinitely many variables of that type.
The  terms and their associated types are generated by the following rules
\begin{enumerate}
	\item\label{it:terms-var} Every variable is a term, of same type as the variable.
	\item\label{it:terms-app} If $M : \tau \to \sigma$ and $N : \tau$ are terms\footnote{For brevity, we write ``$M : \tau$ is a term'' instead of ``$M$ is a term of type $\tau$''.}, then so is  $MN : \sigma$.
	\item\label{it:terms-abs} For every variable $x : \tau$, if  $M : \sigma$ is a term then so is $(\lambda x : \sigma \ M) : \tau \to \sigma$.
	\item \label{it:terms-finite} If $\tau$ is a finite set and  $a \in \tau$,  then $a : \tau$ is a term;
	\item \label{it:terms-pair} If $M : \tau$ and $N : \sigma$ are terms, then so is  $(M,N) : \tau \times \sigma$.
		\item \label{it:terms-coproj} If $\tau_0,\tau_1$ are types and $M : \tau_i$ is a term, then $\coprojterm i M : \tau_0 + \tau_1$ is a term.
 \item \label{it:terms-list} If $M_1 : \tau,\ldots,M_k :\tau$ are terms, then $[M_1,\ldots,M_k] :\tau^*$ is a term.
\item \label{it:terms-atomic} All atomic programs in Figure~\ref{fig:atomic-combinators-1} are terms.
\end{enumerate}
\end{definition}
The pair, list and coprojection constructors in items~\ref{it:terms-pair},~\ref{it:terms-list} and~\ref{it:terms-coproj} could be replaced by extending the operations from Figure~\ref{fig:atomic-combinators-1} with a pairing function, an empty list constant, an append  function, and a coprojection function.

\begin{figure}
\begin{eqnarray*}
\atomicfunctionbis{\isterm^\tau_a}{\tau \to \set{\mathtt{true}}+ \set{\mathtt{false}}}{{\tt true} if the input is $a \in \tau$ and {\tt false} otherwise} {(defined only when $\tau$ is a finite set)}
\atomicfunction{\projterm i^{\tau_0 \tau_1}}{(\tau_0 \times \tau_1) \to \tau_i}{projection to coordinate $i \in \set{0,1}$}
\atomicfunction{\caseterm^{\tau_0 \tau_1 \sigma}}{(\tau_0 \to \sigma) \to (\tau_1 \to \sigma) \to (\tau_0 + \tau_1) \to \sigma}{apply first or second argument, according to  case of third argument}
\atomicfunction{\mapterm^{\tau \sigma}}{(\tau \to \sigma) \to \tau^* \to \sigma^*}{apply  function to all elements in the list}
\atomicfunction{\headterm^{\tau}}{\tau^* \to (\tau + \set \bot)}{return first element, or $\bot$ when empty}
\atomicfunction{\tailterm^{\tau}}{\tau^* \to (\tau^* + \set \bot)}{return all but first element, or $\bot$ when empty}
\atomicfunctionbis{\flatterm^{\tau }}{(\tau^*)^* \to \tau^*}{concatenate list of lists:}{[[1,2],[],[3],[4,5]] $\mapsto$ [1,2,3,4,5]}
\atomicfunctionbis{\splitterm^{\tau }}{\tau^* \to (\tau^* \times \tau^*)^*}{return all possible ways of splitting the list in two parts}{[1,2,3,4] $\mapsto$ [([],[1,2,3,4]), ([1],[2,3,4]), ([1,2],[3,4]), ([1,2,3],[4]), ([1,2,3,4],[]) }
\atomicfunction{\groupterm^G}{G^* \to G}{return the product (in the group) of the input list}
\end{eqnarray*}
  \caption{\label{fig:atomic-combinators-1}Atomic polynomial list programs. For every types $\tau, \tau_0, \tau_1, \sigma$ and every finite group $G$, the above functions are polynomial list programs. The semantics of the functions is explained using Haskell code in Section~\ref{sec:haskell-atomic}.  }
\end{figure}

\paragraph*{Semantics.}
Since there is no recursion,  there is no difficulty in providing compositional denotational semantics, i.e.~assigning to each term $M : \tau$ an element of the domain $\sem \tau$ by induction on the size of term. For a term $M : \tau$ with free variables $x_1 : \tau_1,\ldots, x_n : \tau_n$,  its semantics 
	\begin{align*}
\sem M \in \sem {\underbrace{\tau_1 \to \cdots \to \tau_n}_{\text{environment}} \to \tau}
\end{align*}
	 is defined by induction on the size of $M$ in the natural way. The semantics of the terms defined in items~\ref{it:terms-var},~\ref{it:terms-app},~\ref{it:terms-abs},~\ref{it:terms-finite},~\ref{it:terms-pair}, and~\ref{it:terms-list} should be self-explanatory.  
	 The term $\coprojterm i$ in item~\ref{it:terms-coproj} represents the injection of type $\tau_i$ into the coproduct $\tau_0 + \tau_1$. Regarding item~\ref{it:terms-atomic}, the semantics of the atomic programs from Figure~\ref{fig:atomic-combinators-1} is explained in the Haskell code in Section~\ref{sec:haskell-atomic}.

\begin{definition}[Polynomial list functions]
	A \emph{polynomial list function} is the semantics $\sem M$ of some term without free variables. A  \emph{first-order polynomial list function}  is one that does not use the group product operations.
\end{definition}

The above definition covers functions with higher-order types, like
\begin{align*}
(\set{a,b}^* \to \set{c}^*) \to \set c^* \to \set{a,b,c}^*
\end{align*}
but  in the end, we will be interested in  programs of type $\tau^* \to \sigma^*$, where $\tau$ and $\sigma$ are finite sets. Nevertheless, these programs will typically involve subterms of other types, including higher-order types.

\begin{myexample}\label{ex:headtail}
	We write a program
\begin{align*}
\mathtt{headtwo}_\tau : \tau^* \to (\tau \times \tau) + \bot
\end{align*}
which returns the first two elements of the input list (or the error value if input list has length at most one).  The first idea that comes to mind is 
\begin{align*}
\lambda x : \tau^*\ . \ (\headterm_\tau \ x,\  \underbrace{\headterm_\tau \ (\tailterm_\tau\  x)}_{\text{does not type}}) 
\end{align*}
which is does not type, because the underlined part applies $\headterm_\tau$ to an argument that has type $\tau^* + \bot$. To write a properly typed program, we need error handling. For the error handling, which is admittedly cumbersome, we use $\caseterm$ and $\coprojterm\ $ to implement an error handler, which lifts a function without errors to a function with errors (which is similar to using the error monad in Haskell, but with our syntax being more verbose)
\begin{align*}
\mathtt{err}_{\tau \sigma} : (\tau \to \sigma) \to (\tau + \bot) \to (\sigma + \bot)	
\end{align*}
which is implemented by the following code
\begin{align*}
\lambda f : \tau \to \sigma\  \ . \ \overbrace{\caseterm^{\tau \bot (\sigma+ \bot)}\  \underbrace{(\lambda y : \tau\ . \  \leftterm^{\sigma \bot}\ 	(f\ y))}_{\tau \to (\sigma + \bot)} \ \underbrace{(\lambda y : \bot\ .\ \rightterm^{\sigma \bot} \bot)}_{\bot \to (\sigma + \bot)}}^{(\tau+\bot) \to (\sigma + \bot)}
\end{align*}
Using the above error handler, we can write the correct version of the program which outputs the first two elements of a list, namely
\begin{align*}
\lambda x : \tau^*\ . \ (\headterm_\tau \ x,\ (\mathtt{err}_{\tau^* \tau}\ \headterm_\tau) \ (\tailterm_\tau\  x))
\end{align*}
The last remaining issue is that the output type of the above program is
\begin{align*}
(\tau + \bot) \times (\tau + \bot) \qquad \text{instead of} 	\qquad  (\tau \times \tau) + \bot.
\end{align*}
To convert the type on the left into the type on the right, we use $\distrterm$.
\end{myexample}

We end this section with two non-examples of polynomial list functions, namely {\tt fold} and equality checks.

\begin{myexample}
		The language lacks a fold operation, which can be viewed as an evaluator of automata:
		\begin{align*}
			\mathtt{fold} : \overbrace{(\tau \times \sigma \to \tau)}^{\text{transition function}} \to \overbrace{\tau}^{\text{first state}} \to \overbrace{\sigma^*}^{\text{input word}} \to \overbrace{\tau}^{\text{last state}}
		\end{align*}
Such an operation would change the expressive power, because it would go beyond polynomial growth, and polynomial list functions have polynomial growth. For example, if we take 
\begin{align*}
	\mathtt{delta} = \lambda q\ \lambda a \ \mathtt{concat}[q,q]
\end{align*}
to be the function that doubles the list in the first argument regardless of the second argument, then
\begin{align*}
	\mathtt{fold} \ \mathtt{delta}\ [1]
\end{align*}
will be a program that inputs a list of length $n$ and outputs a list of ones of length $2^n$. In fact, using fold one can get functions of primitive recursive growth, see~\cite[4.1]{hutton1999tutorial}. Also, for the same reasons as described in~\cite{hutton1999tutorial}, {\tt fold} would break preservation of regularity under preimages, see Theorem~\ref{thm:regular-continuous}. 
\end{myexample}

\begin{myexample}
		Checking equality on lists over a binary alphabet, i.e.~a function
		\begin{align*}
			\mathtt{eq} : \set{0,1}^* \to \set{0,1}^* \to \set{0,1} 
		\end{align*}
		would also change the expressive power of the language, because it would break preservation of regularity under preimages. 
\end{myexample}

\pagebreak
\part{Equivalence of the models}
This part shows that all of the models described in Part I are equivalent. There are two variants of the equivalence result: the first-order one, and the general one.

\begin{theorem}\label{thm:main}
	The models $\blue{1,2,3,4}$ listed below define the same classes of string-to-string transductions. Same for $\red{1,2,3,4}$.
	\begin{center}
		\begin{tabular}{ll}
1. \blue{first-order pebble transducers} & \red{1. pebble transducers}\\		
2. \blue{first-order polyregular functions} & \red{2. polyregular functions}\\
3. \blue{first-order polynomial list functions} & \red{3. polynomial list  function}\\
4. \blue{first-order for-transducers} & \red{4. for-transducers}\\
\end{tabular}
	\end{center}
\end{theorem}

Polynomial list programs can define functions other than string-to-string functions. The theorem, in items \blue 3 and \red 3, talks only about string-to-string functions defined by polynomial list functions, as opposed to the more general class of functions with higher order types that can be defined. 

 The  equivalences  are shown according to the following plan
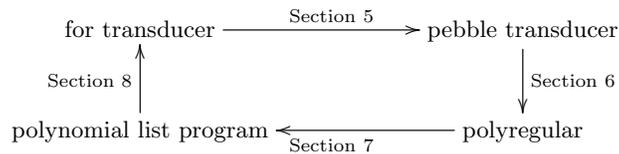
\begin{align*}
\xymatrix{
   {\small \text{for transducer}} \ar[rr]^{\text{Section~\ref{sec:for-to-pebble}}} & & {\small \text{pebble transducer}} \ar[d]^{\text{Section~\ref{sec:pebble-to-polyregular}}}  \\
{\small \text{polynomial list program}} \ar[u]^{\text{Section~\ref{sec:program-to-for}}} & &   {\small\text{polyregular}} \ar[ll]^{\text{Section~\ref{sec:polyregular-to-programs}}}
} 	
\end{align*}

Before showing all of the implications, we give some high level comments on the difficulties involved and the techniques used to solve them. 

The transformation from for-transducers to pebble transducers is straightforward, and the only nontrivial part is showing that the first-order restriction in a for-transducer (the Boolean variables can only change value once) translates to the first-order restriction in a pebble transducer (there is a first-order formula defining reachability on configurations). This is proved using the same ideas as in Lemma~\ref{lem:all-pebble-mso} about \mso definability of reachability in pebble automata.

The transformation from a pebble transducer to a polyregular function is the most technical part in the proof of Theorem~\ref{thm:main}. The general idea is that Simon's Factorization Forest Theorem is applied to find  repeating patterns in the movement of the head of a pebble transducer; and transducers where the head moves using such repeating patterns can be simulated using the atomic polyregular functions.

The transformation from polyregular functions to polynomial list programs is conceptually quite straightforward:  the polyregular functions are designed to have  minimal syntax, while polynomial list programs are designed to be a usable programming language. The hard part is  simulating a finite automaton  (which is the model underlying rational functions, which are one of the atomic types of polyregular functions) can be simulated by a polynomial list program. The difficulty is that polynomial list programs  do not have any explicit iteration mechanisms. To prove this, we use a result from~\cite{DBLP:conf/lics/BojanczykDK18}, which solved  the same kind of problem: simulating a rational function using a restricted functional programming language. 

Finally, to transform a polynomial list program into a for-transducer, we use two main ideas. The first main idea is that for-transducers are closed under composition, which is proved using the same kind of stack-of-stacks idea as was used in Theorem~\ref{thm:pebble-closed-under-composition} about closure of pebble transducers under composition. The second main idea is to use a term rewriting semantics of polynomial list programs (namely, normal forms under $\beta$-reduction), and  then to  show this semantics can be implemented using a composition of finitely many for-transducers. 
\section{For-programs to pebble-transducers}
\label{sec:for-to-pebble}

In this section we show that  for-transducers can be transformed into pebble transducers, as stated in the following lemma. 

\begin{lemma}\label{lem:for-to-pebble}
    Every string-to-string function recognized by a for-transducer is  recognized by a  pebble-transducer.   Furthermore, if the for-transducer is first-order definable, then so is the pebble transducer. 
\end{lemma}

Without the  first-order restriction, the simulation  is completely straightforward: the simulating  pebble transducer  stores the memory state of the simulated for-program, i.e., the instruction that is about to be executed, and the valuation of the  position and Boolean variables. 
The  values of the Boolean variables are stored in the state of the pebble transducer, while the values of the position variables are stored in the pebbles. Stack discipline for the pebbles follows from the nesting of the loops in the for-program. This construction is described in more detail below, in the first-order case, where additional attention is needed.

In the first-order case,  we have to show that if the for-transducer is first-order definable (which means that Boolean variables can only go from {\tt false} to {\tt true}) then the simulating pebble transducer is also first-order (which means that the reachability relation on configurations is definable in first-order logic).    To prove this, it is convenient to assume that the for-transducer is in  \emph{prenex normal form} as defined in the  following picture:
\mypic{30}

A straightforward structural induction shows the following result. 
\begin{lemma}\label{lem:for-normal-form}
    For every for-program there is a for-program in prenex normal form which recognizes the same function.  The construction preserves first-order definability.  
\end{lemma}

Using prenex normal form, we complete the proof of Lemma~\ref{lem:for-to-pebble} in the first-order case. Consider a first-order for-transducer $f$ in prenex normal form.  Suppose that $f$ has $n$ position variables. Define a \emph{configuration of $f$}  in an input word $w$ to be a tuple 
\begin{align*}
(q,x_1,\ldots,x_n) 
\end{align*}
where $q$ is a valuation of the Boolean variables and $x_1,\ldots,x_i$ are positions in $w$. The idea is that the configuration represents the program state just before executing the kernel of $f$, assuming that the variables have values as given in the configuration. The for-transducer begins in the configuration where $q$ maps all Boolean variables to {\tt false}, and where $x_i$ is set to either the first or last position, depending on the type of for loop that binds variable $x_i$. 

The reachability relation on configurations is defined in the natural way: one configuration is reachable from another one (in a fixed input word) if it appears later in the computation. We write
\begin{align*}
    w \models (q,x_1,\ldots,x_n)  \to^* (p,y_1,\ldots,y_n) 
\end{align*}
to say that in the input word $w$, the for-transducer can go from configuration $(q,x_1,\ldots,x_n)$ to $(p,y_1,\ldots,y_n)$. For a configuration, define the \emph{associated output} to be the output produced by the  kernel of the for-transducer assuming that the values of the  variables (position and Boolean) are as in the configuration just before executing the kernel. Because the kernel does not change the position variables, and can only test their relative order and labels in the input word,   the associated output depends only on $q$ and the quantifier-free type of the positions $x_1,\ldots,x_n$ with respect to the order and labelling predicates on positions in the input word. 

The simulating pebble automaton stores $q$ in its state and uses at most $n$ pebbles to store the positions $x_1,\ldots,x_n$. Therefore, to show that the simulating pebble automaton is first-order definable (and thus complete the proof of Lemma~\ref{lem:for-to-pebble}), it is enough to show that the reachability relation on configurations is definable in first-order logic. This is proved  the following lemma.

\begin{lemma}\label{lem:pebble-to-for-formula}
    Let $p,q$ be valuations of the Boolean variables and let  $i \in \set{0,\ldots,n}$. There is a first-order formula 
\begin{align*}
    \varphi_{p,q}(\overbrace{x_1,\ldots,x_n}^{\bar x},\overbrace{y_1,\ldots,y_n}^{\bar y})
    \end{align*}
which holds in an input word $w$ with position tuples $\bar x, \bar y$ if and only if
\begin{align*}
    w \models (p, \bar x) \to^* (p, \bar y)
    \end{align*}
\end{lemma}
\begin{proof}    This proof follows the same structure as Lemma~\ref{lem:all-pebble-mso}. The only difference is that instead of using transitive closure (as was the case in Lemma~\ref{lem:all-pebble-mso}), we use the assumption that the variables in a first-order transducer can only go from {\tt false} to {\tt true}. 

    Let $i \in \set{1,\ldots,n}$. Define $X_{i}$ to be the set of Boolean variables in the for-transducer that are declared in the $i$-th for loop or before.      Define an $i$-configuration in an input word $w$ to be a tuple of the form 
    \begin{align*}
        (q,x_1,\ldots,x_i)
    \end{align*}
    where $x_1,\ldots,x_i$ are positions in $w$ and $q$ is a valuation of  $X_{i}$. This is like a configuration, except that it does not contain the valuations of all variables. Intuitively, an $i$-configuration describes the first configuration in an iteration of the $i$-th loop.   A configuration, in the sense used in the statement of the  lemma,  is the same thing as an $i$-configuration for $i=n$. Conversely, an $i$-configuration can be viewed as a configuration (i.e., an $n$-configuration) by setting all remaining Boolean variables to {\tt false} and setting all positions $x_{i+1},\ldots,x_n$ to their initial values in their corresponding loops (which are first or last positions in $w$, depending on whether the corresponding loops are {\tt first..last} or  {\tt last..first}). We write 
    \begin{align*}
        w \models \underbrace{(p, \bar x)}_{\text{$i$-configuration}} \to_i^* \underbrace{(q, \bar y)}_{\text{$i$-configuration}}
        \end{align*}
if $\to^*$ holds for the corresponding $n$-configurations.

\begin{sublemma}
    Let $i \in \set{1,\ldots,n}$ and let  $p,q$ be valuations of $X_i$.  There is a first-order formula 
    \begin{align*}
        \varphi_{p,q}(\overbrace{z_1,\ldots,z_{i-1}}^{\bar z},x,y)
        \end{align*}
    which holds in an input word $w$ and a valuation of $\bar z$ and $y,z$ if and only if
    \begin{align*}
        w \models (p, \bar zx) \to_i^* (p, \bar zy)
        \end{align*}
    \end{sublemma}
\begin{proof}
Note that in the statement of the sublemma, we do not have to say that the run from $(p,\bar zx)$ to $(p, \bar zy)$ does not change the values of the variables in $\bar z$; because this is automatically true in for-transducers, where a position variable can never return to a previously used value once that value has changed. (The same is true for Boolean variables in a first-order for-transducer.)
Induction on $i$, in the opposite order, i.e.~the induction base is when $i=n$. Suppose that we want to prove the statement for some $i \in \set{1,\ldots,n}$ and we have already proved it for $j>i$.  Let us write 
\begin{align}\label{eq:one-step-for}
    w \models (p, \bar zx) \to_i (q, \bar zy)
    \end{align}
when there is a run from  $(p,\bar zx)$ to $(p, \bar zy)$ which contains no  $i$-configurations except for the source and target.  Note that in the one step case, the variable $y$ is uniquely determined by $x$: it is either the next or previous position, depending on the type of the $i$-th for loop.    Using the induction assumption, we can write a first-order formula
\begin{align*}
    \psi_{pq}(\bar zx)
\end{align*}
which holds whenever~\eqref{eq:one-step-for} is satisfied, with $y$ being the next/previous position corresponding to $x$.  Next,  consider the case when multiple steps are allowed, but  the valuations of the Boolean variables are the same all the time, i.e.
\begin{align}\label{eq:many-step-same-boolean}
    w \models (p, \bar zx) \to^*_i (p, \bar zy).
    \end{align}
Note that the assumption that $f$ is a first-order transducer implies that in all intermediate $i$-configurations between $(p,\bar zx)$ and $(p, \bar zy)$, the valuation of the Boolean variables $X_i$ is also $p$, because a Boolean variable can never go from {\tt false} to {\tt true} and then back again.  Thanks to this observation, we can define~\eqref{eq:many-step-same-boolean}  in first-order logic, by saying that  $\psi_{pq}(\bar zu)$ holds for all $u$ between $x$ and $y$. Since the Boolean variables can only grow in a run, the formula in the statement of the lemma is obtained by combining a bounded number of steps of the form~\eqref{eq:one-step-for} and~\eqref{eq:many-step-same-boolean}.
\end{proof}

Using the above sublemma, we get Lemma~\ref{lem:pebble-to-for-formula} easily, by the same reasoning as in the end of the proof of Lemma~\ref{lem:all-pebble-mso}. The observation is that  an arbitrary run of the for-transducer can be decomposed into a bounded number of steps, which can be described in first-order logic using the sublemma. 
\end{proof}

\section{Pebble transducers to polyregular functions}
\label{sec:pebble-to-polyregular}
 In this section we show that every pebble transducer recognises a polyregular function. The converse inclusion was already discussed in Section~\ref{sec:pebble-automata}. Furthermore, since each of the atomic polyregular functions is recognised by a 2-pebble transducer, it follows that every function recognised by a $k$-pebble transducer can be decomposed into  finitely many functions recognised by 2-pebble transducers.  One pebble is not enough, since 1-pebble transducers, also known as two-way automata with output, are closed under composition, which was proved in~\cite[Theorem 1]{chytil1977serial}, see also~\cite[Theorem 12.3]{toolbox}.

The main idea is this. A run of a pebble transducer can be viewed as a sequence of words representing the configurations, as in the following picture:
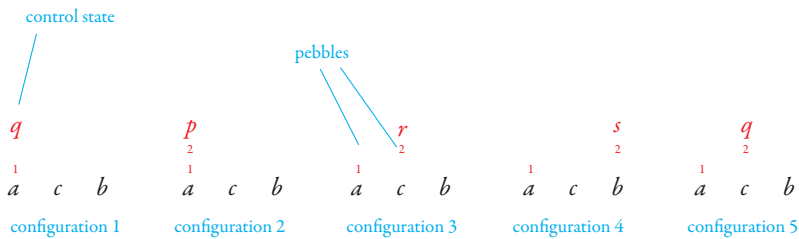
\mypic{11}
The length is  at most 
\begin{align*}
\text{(number of states)} \cdot  \text{(length of input word)}^{\text{1 + number of pebbles}}
\end{align*}
The 1 + in the exponent is because the input word is copied in each configuration. We show that for every pebble transducer, the word representing the  run can be computed using a polyregular function. Once the run has been computed, the output can easily be recovered using a sequential function, which simply replaces each configuration by the output produced in it.

   The proof is split into two parts. In Section~\ref{sec:one-pebble}, we consider the case of  1-pebble automata (also known as two-way automata), and in Section~\ref{sec:many-pebbles}, we iterate the construction for 1-pebble automata to get the result for $k$-pebble automata. The 1-pebble case in Section~\ref{sec:one-pebble} is the core of the proof, and uses the Factorisation Forest Theorem of Imre  Simon~\cite{simon1990factorization}. Even in the case of a 1-pebble automaton,  runs can have quadratic length, hence  the squaring function.


%
%

\subsection{One pebble}
\label{sec:one-pebble}
In this section, we show that a polyregular function can produce the run of a   1-pebble automaton. Since a 1-pebble automaton is the  same thing as a {two-way automaton}, we use the name two-way automaton for the rest of this section. For the rest of Section~\ref{sec:one-pebble} fix a two-way automaton $\Aa$ with input alphabet $\Sigma$ and states~$Q$. 

\paragraph*{Some notation.}  We begin by fixing some notation on runs and configurations. In the proof we analyse the behaviour of the two-way automaton on infixes of the input word, which do not necessarily begin with the first or end with the last position.  We use the name \emph{partial input} for such infixes. A  partial input is formally defined  to be a word over the alphabet
\begin{align*}
\Sigma \times \powerset {\set{\text{first, last}}} 
\end{align*}
where   ``first'' is only allowed (but not necessary) in the first position, and  ``last'' is only allowed (but not necessary) in the last position.  We write $\inputs$ for the set of partial inputs; this is a regular language.  
Define a \emph{partial configuration} to be a partial input with exactly one distinguished position labelled additionally by a state. A partial configuration is formalised as a word over the  alphabet
\begin{align*}
(\Sigma \times \powerset {\set{\text{first, last}}}) \times  (Q + \varepsilon) 	
\end{align*}
where $Q$ is used exactly once, and erasing the last component leads to a partial input. We write $\confs$ for the set of partial configurations; again this is a regular language. Here is a picture:
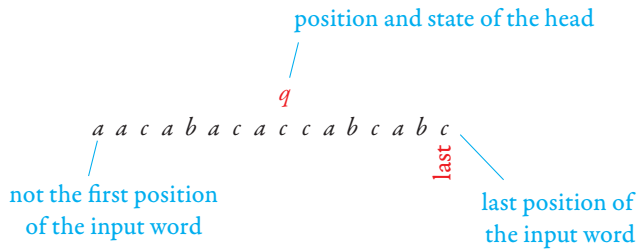
\mypic{35}
Since we represent partial configurations as strings, it makes sense to talk about string-to-string transducers that transform them. One example is the successor function 
\begin{align*}
  \mathsf{suc} : \confs \to \confs
\end{align*}
which applies a single transition of the automaton. This successor function is undefined if the source partial configuration has the accepting state, or the head is moved outside the partial configuration (which should not viewed as an error, since it corresponds to the automaton leaving the infix covered by the partial configuration).  The successor function is easily seen to be a rational function; but this is not going to be useful, since the difficulty is in iterating the successor function. For example, the successor function for configurations of Turing machines is also rational. 

Define a \emph{partial run} to be a word 
\begin{align*}
  c_1 |  \cdots | c_n \qquad c_1,\ldots,c_n \in \confs
\end{align*}
which consists of partial configurations, separated by a fresh separator symbol, where consecutive partial configurations are connected by the successor  function, and where the last partial configuration has undefined successor (typically, because the head leaves the part of the input covered by the partial configuration). We write $\runs$ for the set of partial runs. Finally, define 
\begin{align*}
\run : \confs \to \runs  
\end{align*}
to be the (total) function which maps a partial configuration to the unique partial run which begins in that partial configuration. The main result of this section is the following lemma.

\begin{lemma}\label{lem:run-is-polyregular}
	The function $\run$ is polyregular. If the fixed two-way automaton is first-order definable, then $\run$ is first-order polyregular. 
\end{lemma}

For a partial input $w$ and a state $q$, we write $qw$ for the partial configuration which has input $w$ and the head over the first position in state $q$. Likewise, we define $wq$, but with the last position used. The following slightly technical definition will be used in the induction proof that comprises the rest of this section.

\begin{definition}[Good]
We say that $  L \subseteq  \inputs$  is  \emph{good} if (a) it is a regular word language; and (b)  there is a polyregular function that agrees with $\run$ on inputs from the set
\begin{align*}
  \set{qw, wq  : w \in L, q \in Q}
\end{align*}
Furthermore, if the fixed two-way automaton is first-order definable, the we additionally require that (i) $L$ is first-order definable, and (ii) the polyregular function in (b) is first-order polyregular.
\end{definition}

The main content of  Section~\ref{sec:one-pebble} is to prove Lemma~\ref{lem:all-good} below, which says that all partial inputs are good. In other words, partial runs can be computed using a polyregular function, at least assuming that the first partial  configuration has the  head over the first or last position. The implication from Lemma~\ref{lem:all-good} to Lemma~\ref{lem:run-is-polyregular} is a straightforward argument using crossing sequences, and is given in Sublemma~\ref{sublemma:middle-q}. 
\begin{lemma}\label{lem:all-good}
The set $\inputs$ of all partial inputs is good.
\end{lemma}

To  prove the above proposition, we  apply the Factorisation Forest Theorem (see below) to the semigroup homomorphism that maps a partial input to the behaviour of the two-way automaton. We begin by describing this semigroup in more detail.

\paragraph*{Crossing types.} Crossing types are the natural information associated to a two-way automaton; since this notion is so classical and natural we only give the notation and intuition, for a more precise description see~\cite{Carton:2015bl}. Define the \emph{crossing type} of $w \in \inputs$ to be the function
\begin{align*}
  Q \times \set{\text{leftmost, rightmost}} \quad \to  \quad\set{\text{accept, reject}} + Q \times \set{\text{leftmost, rightmost}} 
\end{align*}
which describes the behaviour of the automaton on a partial input, as described in the following picture
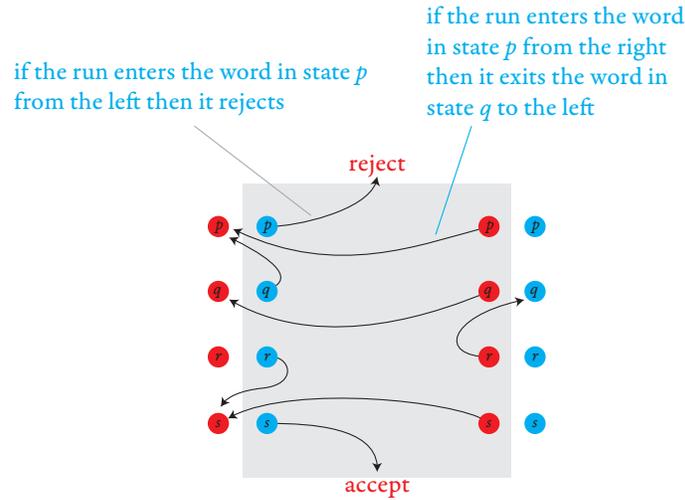
\mypic{12} 
Crossing types can be equipped with a semigroup structure such that the function which maps an input to its crossing type becomes a semigroup homomorphism, here is a picture:
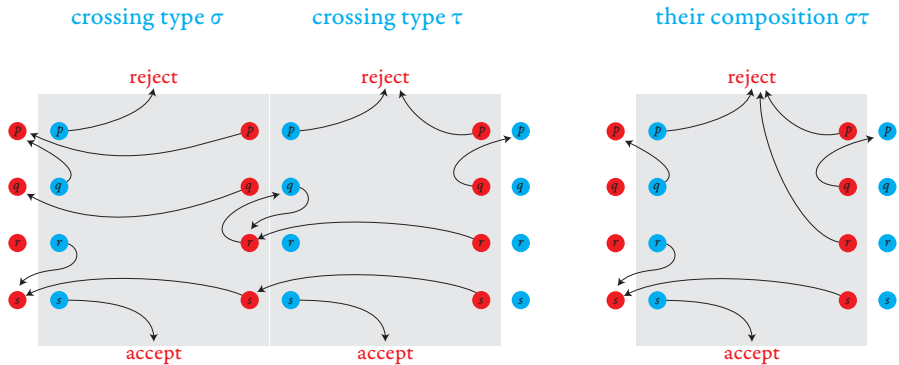
\mypic{13}

Furthermore, by~\cite[Theorem 9]{Carton:2015bl}, if $\Aa$ is  a first-order definable two-way automaton, then the semigroup of crossing types  is an aperiodic semigroup.

%

\paragraph*{Proof plan.} We now present the proof strategy for  Lemma~\ref{lem:all-good}, which says that the entire set $\inputs$ is good. The plan is to  show that good languages can be combined using a form of concatenation and Kleene star with separator symbols to get new good languages.  More formally,  we show that if  $L,K \subseteq \inputs$ are good and $|$ is a fresh separator symbol, then:

\begin{enumerate}
	\item {\bf Lemma~\ref{lem:good-concatenation}}. 	For every state $q$  there is a polyregular function $f$
 with
\begin{align*}
  f(w_1|w_2) = \run(q w_1 w_2) \qquad \text{for every } w_1 \in L, w_2 \in K.
\end{align*}

	\item {\bf Lemma~\ref{lem:unranked-sweeper}}. Assume that all words in $L$ have the same crossing type. Then for every state $q$  there is a polyregular function $f$
with 
\begin{align*}
  f(w_1|\cdots | w_n) = \run(q w_1 \cdots w_n) \qquad \text{for every }n \in \set{1,2,\ldots} \text { and }w_1,\ldots,w_n \in L.
\end{align*}

\end{enumerate}
Furthermore, if the underlying automaton is first-order definable, then the polyregular functions in the conclusions of the above lemmas are first-order polyregular.

Before proving the above two lemmas, we show how they imply that all partial inputs are good, as required by Lemma~\ref{lem:all-good}. The idea is to use the Factorisation Forest Theorem, which says that every word can be factorised into several words, and those words can also be factorised, and so on recursively, so that the depth of the recursion is bounded and the factorisations are compatible with Lemmas~\ref{lem:good-concatenation} and~\ref{lem:unranked-sweeper}. The appropriate definition is given below (think of $h$ being the homomorphism which maps a word to its crossing sequence).

\begin{definition}\label{def:h-height} Let $h : \Sigma^+ \to S$ be a semigroup homomorphism, with $S$ finite.  Define the $h$-height of a word $w \in \Sigma^+$ to be the smallest natural number that can be assigned using the following rules.
\begin{enumerate}
	\item every one letter word has $h$-height $1$;
	\item if $w,v$ have $h$-height $< n$, then $wv$ has $h$-height $\le n$.
	\item \label{it:unranked-h-height} if $w_1,\ldots,w_n$ have $h$-height $<n$ and the same value\footnote{Often one  assumes that this same value is an idempotent. We do not make this assumption, which plays a role when computing factorisations in first-order logic, see Footnote~\ref{foonote:non-idempotent}.} under $h$, then $w_1 \cdots w_n$ has height $\le n$.
\end{enumerate}
\end{definition}
Clearly every word has some $h$-height, e.g.~at most  its length (or even the logarithm of its length). 
The Factorisation Forest Theorem  says that the upper bound on $h$-height is actually $3|S|$, i.e.~there is a finite upper bound that works for all words and depends only on the semigroup $S$.  The theorem was originally proved by Imre Simon~\cite[Theorem 9.1]{simon1990factorization} with a bound of $9|S|$, while the optimal bound of $3|S|$ is from~\cite[Theorem 1]{kufleitner2008height}. What is more,  a suitable decomposition can be computed using a rational function, in the following sense. 

\begin{lemma}
[\cite{colcombet2007combinatorial,DBLP:conf/lics/BojanczykDK18}]
\label{lem:compute-fact-for}  Let $h : \Sigma^+ \to S$ be a semigroup homomorphism, and let $|$ be a fresh separator symbol. 
	For every $k \in \set{2,3,\ldots}$ there is a rational function which inputs a word $w \in \Sigma^+$ and outputs a decomposition 
	\begin{align*}
  w_1 | \cdots | w_n \qquad \text{with } w= w_1 \cdots w_n
\end{align*}
such that 
\begin{enumerate}
	\item $w_1,\ldots,w_n$ have $h$-height strictly smaller than $w$; and
	\item either $n=2$ or all $w_1,\ldots,w_n$ have the same value under $h$.
\end{enumerate}
Furthermore, if  $S$ is aperiodic, then a first-order rational function is enough\footnote{\label{foonote:non-idempotent}To get a first-order rational function, it is important that we do not require values to be idempotent in item~\ref{it:unranked-h-height} of Definition~\ref{def:h-height}. To see how this is important, consider the semigroup $\set{1,2}$ where all products have value $2$, i.e.~this is the syntactic semigroup of the language ``words of length exactly $1$''. This semigroup is clearly aperiodic, and yet a rational function cannot produce decompositions with idempotent values, since this would require grouping letters into groups (say, of size two). }. 
\end{lemma}

\begin{proof}[of Lemma~\ref{lem:all-good}]
Let $h$ be the homomorphism which maps an input word to its crossing type with respect to our fixed two-way automaton. Since the semigroup of crossing types is finite, the Factorisation Forest Theorem implies  that every partial input has $h$-height bounded by a constant that depends only on the fixed two-way automaton, and not on the length of the word.
	By induction on $k$, we show that the set of all words in $\inputs$ with  $h$-height at most $k$ is good. In the induction step, we use the rational function from Lemma~\ref{lem:compute-fact-for}, and then apply either Lemma~\ref{lem:good-concatenation} or~\ref{lem:unranked-sweeper} to the result, depending on the number of factors produced.
\end{proof}

It remains to prove Lemmas~\ref{lem:good-concatenation} and~\ref{lem:unranked-sweeper}; the rest of Section~\ref{sec:one-pebble} is devoted to proving these lemmas.  Since the proofs  use only first-order polyregular functions,  we neglect to always add that ``if the underlying automaton is first-order definable, then ...''. There is one  exception, Sublemma~\ref{sublemma:windows}, where special care is needed in the first-order case.

\paragraph*{If-then-else.}
We begin by showing  that   polyregular functions can be combined   using an ``if then else'' construction.

	\begin{lemma}\label{lem:ifthenelse}
	If $L \subseteq \Sigma^*$ is a regular language, and $f,g : \Sigma^* \to \Gamma^*$ are polyregular functions, then  the following function is also polyregular 
	\begin{align*}
  w \in \Sigma^* \quad  \mapsto  \quad \begin{cases}
  	f(w) & \text{for }w \in L\\
  	  	g(w) & \text{for }w \not \in L\\
  \end{cases}
\end{align*}
\end{lemma}
\begin{proof}
The basic 	idea is to decompose the natural parallel implementation of ``if then else'' into a sequential one. The  main step is given in the following sublemma.

\begin{sublemma}\label{lem:if-identity}
	For every polyregular function $h : \Sigma^* \to \Gamma^*$ and every alphabet $\Delta$ disjoint from $\Sigma$ and $\Gamma$, there is a polyregular function 
	\begin{align*}
  h^\Delta : (\Sigma + \Delta)^* \to (\Gamma + \Delta)^*
\end{align*}
	 which agrees with $h$ on inputs from $\Sigma^*$ and is the identity on inputs from $\Delta^*$. We have no requirements for inputs that use both $\Sigma$ and $\Delta$.
\end{sublemma}
\begin{proof}
If the conclusion of the sublemma is true for two polyregular functions, then it is also true for their composition, by setting
\begin{align*}
  (h_1 \circ h_2)^\Delta =  h_1^\Delta \circ h_2^\Delta
\end{align*}
Therefore, it remains to show the sublemma for the  atomic functions, namely the sequential functions, squaring, and iterated reverse. The case of rational functions is straightforward, by taking a union of automata. The case of squaring is also easy: first apply squaring, and  if the result contains at least one letter from $\Delta$, then use a  rational function to recover the original input. The most  interesting case is when   when $h$ is iterated reverse operation, which is implemented as follows. 
\begin{enumerate}
	\item Apply the following rational function:
\begin{enumerate}
\item if the input is in $\Sigma^*$, e.g. it is
\begin{align*}
123 | 45 | 67 | 8	
\end{align*}
then leave the input unchanged.
\item if the input is in $\Delta^*$, e.g.~it is
\begin{align*}
{abcdef}
\end{align*}
then add the separator $|$ between every two positions, like this:
\begin{align*}
{a|b|c|d|e|f}
\end{align*}
\item otherwise, output the empty word
\end{enumerate}
	\item apply iterated reverse, with the separator being 
	 $|$
	\item if the output contains letters from $\Delta$ remove all separators $|$.
\end{enumerate}

\end{proof}

Having proved Sublemma~\ref{lem:if-identity}, we return to the proof of Lemma~\ref{lem:ifthenelse} about an if-then-else construction for polyregular functions. 
 Let us write $\red \Sigma$ for a disjoint copy of the alphabet $\Sigma$, and for a word $w \in \Sigma^*$ let  us write $\red w \in \red \Sigma^*$ for the corresponding word over the copied alphabet.  The function in the statement of the lemma is implemented by the following sequence of operations (we assume without loss of generality that $\Sigma$ and $\Gamma$ are disjoint): 
\begin{enumerate}
	\item if the input is not in $L$, replace $w$ by $\red w$;
	\item apply $f^{\red \Sigma}$ as defined in Sublemma~\ref{lem:if-identity};
	\item swap $\Sigma$ with $\red \Sigma$;
	\item apply $g^{\Gamma}$ as defined in Sublemma~\ref{lem:if-identity};
\end{enumerate}
The  functions   items 1 and 3 are clearly rational, while the functions in items 2 and 4 are polyregular thanks to Sublemma~\ref{lem:if-identity}.
\end{proof}

\paragraph*{Concatenation.} We now show the first of the two lemmas needed in Lemma~\ref{lem:all-good}, namely that  two good languages can be combined via concatenation.
\begin{lemma}\label{lem:good-concatenation}
	If $L,K \subseteq \inputs$ are good, then for every state $q$  there is a polyregular function $f$
 with
\begin{align*}
  f(w_1|w_2) = \run(q w_1 w_2) \qquad \text{for every } w_1 \in L, w_2 \in K.
\end{align*}
\end{lemma}
\begin{proof}
Using the ``if then else'' construction from Lemma~\ref{lem:ifthenelse}, we can assume without loss of generality that  all words in $L$ have the same crossing type, call it $\sigma$, and all words in $K$ have the same crossing type, call it $\tau$.  The following sublemma, which follows almost immediately from the definitions, shows how the run of a two-way automaton  on input $vw$ can be reconstructed in a compositional way from its runs on the factors $v$ and $w$. In the lemma, we use the following notation: for $w \in \inputs$ and $\rho \in \runs$, we write $w \odot \rho \in \runs$ for the partial run obtained prepending $w$ to the left of every partial configuration appearing in $\rho$. Similarly we define $\rho \odot w$. Here is a picture:
\mypic{25}
We will show in Sublemma~\ref{sublemma:iterated-concat} below that $\odot$ can be implemented by a polyregular function. We begin though by showing how a run over $w_1w_2$ decomposes into a bounded number of runs over $w_1$ or $w_2$; this result is a simple application of crossing types. (The sublemma talks about runs that begin in the leftmost position; a symmetric construction deals with runs that begin in the rightmost position.)

%
%

\begin{sublemma}\label{claim:sweeper-binary-concatenation} For every crossing types $\sigma, \gamma$ and every $q \in Q$  there exist $k \in \set{0,1,\ldots}$ and states 
\begin{align*}
  q_1,\ldots,q_k \in Q 
\end{align*}
such that  partial inputs  $w_1, w_2$ with crossing types $\sigma,\tau$ respectively satisfy
	\begin{align*}
  \run (q w_1w_2) = \rho_0 \cdots \rho_k 
 \quad \text{where }  \rho_i = \begin{cases}
  \run(q w_1) \odot w_2 & \text{when $i=0$}\\
  	w_1  \odot \run(q_i w_2) & \text{when $i \in \set{1,3,5,\ldots}$}\\
  	\run(w_1 q_i)  \odot w_2 & \text{when $i \in \set{2,4,6,\ldots}$}
  \end{cases}
\end{align*}

\end{sublemma}
\begin{proof}
	The proof is best seen in a picture:
	\mypic{3}
\end{proof}

To finish the proof of the lemma, we need to show that the construction 
described in Sublemma~\ref{claim:sweeper-binary-concatenation} can be implemented by a polyregular function. The first challenge is implementing the concatenation $\rho_1 \cdots \rho_k$. Note that the number of concatenated blocks $k$ is constant in the sense that it depends only on the fixed crossing types $\sigma, \tau$ and not on the input words $w_1,w_2$. The following sublemma shows that polyregular functions are closed under this type of  concatenation.

\begin{sublemma}\label{sublemma:poly-cat}
	If $f,g : \Sigma^* \to \Gamma^*$ are polyregular, then so is $w \mapsto f(w)g(w)$.
\end{sublemma}
\begin{proof}
The sublemma follows from the two observations below.
\begin{enumerate}
	\item As in Sublemma~\ref{lem:if-identity}, let us write $\red \Sigma$ for a disjoint copy of the alphabet $\Sigma$, and for $w \in \Sigma^*$ let us write $\red w$ for the corresponding word in $\red \Sigma^*$. First observe that the function
\begin{align*}
  w \mapsto w \red w
\end{align*}
is polyregular. This is done as follows. Suppose that the input is 
\begin{align*}
1234	
\end{align*}
Add a fresh separator  at the end of $w$, yielding
\begin{align*}
1234|	
\end{align*}
and then apply squaring, yielding
\begin{align*}
\underline 1234| 1\underline 234|	12\underline 34|	123\underline 4|	1234\underline |	
\end{align*}
Remove the underlines, colour the letters before the first separator black, the letters between the first and second separator red, and remove the remaining letters, giving the desired output
\begin{align*}
1234 \red{1234}.
\end{align*}
 
\item The second observation is that if $h : \Sigma^* \to \Gamma^*$ is polyregular, then the same is true for the function $\bar h$ defined by 
\begin{align*}
  w  \in (\Sigma + \red \Sigma)^* \quad \mapsto \quad \begin{cases}
  	h(w_1) \red{w_2} & \text{if $w=w_1 \red{w_2}$ for   $w_1 \in \Sigma^+$ and $\red{w_2} \in \red \Sigma^+$}\\
  	\varepsilon & \text{otherwise}
  \end{cases}
\end{align*}
As was the case in the proof of Sublemma~\ref{lem:ifthenelse}, it is enough to show that $\bar h$ is polyregular whenever $h$ is one of the atomic functions, and the straightforward proof of that is left to the reader. 
\end{enumerate}
The sublemma follows immediately: we first duplicate the input word as in item 1 above, then apply item 2 with $h=f$, and finally apply a symmetric version of item 2  with $h=g$.
\end{proof}

We now finish the proof of the lemma. Suppose that the input is $w_1|w_2$. 
Thanks to Sublemmas~\ref{claim:sweeper-binary-concatenation} and~\ref{sublemma:poly-cat}, to  finish the proof of the lemma, it is enough to show that for every state $q$, the functions
\begin{align*}
w_1 | w_2  \mapsto \run(q w_1) \odot w_2 \qquad   w_1| w_2  \mapsto w_1 \odot \run(qw_2) 
\end{align*}
are polyregular (and likewise for runs where $q$ is at the end of the input word). By symmetry, we only do the first function. By assumption that $L$ is good and closure of polyregular functions  under concatenation (Sublemma~\ref{sublemma:poly-cat}), the function
\begin{align*}
w_1 | w_2  \mapsto \run(q w_1)  |  w_2
\end{align*}
is polyregular. To finish the proof, we use the following construction, which allows us to implement $\rho|w \mapsto \rho \odot w$ using a polyregular function.
\begin{sublemma}\label{sublemma:iterated-concat}
	For every alphabet $\Sigma$ and  separator $|$ not in  $\Sigma$, the function
	\begin{align*}
  v_1 | \cdots | v_n | v \qquad \mapsto \qquad v_1v | \cdots | v_n v
\end{align*}
is polyregular (in the above, we assume that $v_1,\ldots,v_n,v$ do not use $|$).
\end{sublemma}
\begin{proof}Suppose that the input looks like this:
\begin{align*}
  12|345|6|78
\end{align*}
The number of blocks, as separated by $|$, is unbounded.
Using a sequential function, append a fresh endmarker, say a comma
\begin{align*}
    12|345|6|78,
\end{align*}
and then apply squaring. In the result, keep only the maximal blocks in $\Sigma$ where either the first position is  underlined, or the block is directly followed by a comma and the closest underlined position to the left is the first in its block, giving a result like this:
\begin{align*}
    \underline 12|78 , \underline 345| 78 , \underline 6|78,
\end{align*}
Finally, swap $|$ and the comma, and then remove the commas.
\end{proof}
\end{proof}

\paragraph*{Homogeneous Kleene star.}
To finish the proof of Lemma~\ref{lem:all-good}, we  show that good languages are closed under a variant of the Kleene star, where the starred language contains only words of the same crossing type and the different copies of the starred language are separated by $|$.
\begin{lemma}\label{lem:unranked-sweeper}
Suppose that $L$ is good, and all words in $L$ have the same crossing type. Let $|$ be a fresh separator symbol. Then for every state $q$ of the fixed two-way automaton  there is a polyregular function $f$ with
\begin{align*}
  f(w_1|\cdots | w_n) = \run(q w_1 \cdots w_n) \quad \text{for every }n \in \set{1,2,\ldots} \text{ and }w_1,\ldots,w_n \in L.
\end{align*}
\end{lemma}
\begin{proof}
Suppose that the input is 
\begin{align*}
 w_1 | \ldots | w_n \qquad \text{with }w_1,\ldots,w_n \in L
\end{align*}
Let $q$ be a state as in the assumption of the lemma, and let us write
\begin{align*}
  \rho \eqdef \run (qw_1 \cdots w_n)
\end{align*}
For $i \in \set{1,\ldots,n}$, define
\begin{itemize}
	\item[$\rho_i$] is sequence of configurations from  $\rho$ that begins in the configuration with the  first visit in $w_i$ and ends just before the first visit in $w_{i+1}$; and
	\item[$q_i$]  is the state used in the first visit in $w_i$.
\end{itemize}
 If $\rho$ never visits $w_i$, then $\rho_i$ is empty and $q_i$ is undefined. If $\rho$ never visits $w_{i+1}$, or $i=n$, then $\rho_i$ is a suffix of $\rho$.  Note that the first visit in $w_i$ necessarily is in the leftmost position, since $\rho$ begins in the leftmost position of $w_1| \cdots | w_n$. By the assumption that all of the words $w_1,\ldots,w_n$ have the same crossing type, it follows that the sequence $q_1,\ldots,q_n$ behaves in a certain cyclic way, this will be explained in the proof of Sublemma~\ref{sublemma:windows} below, but not used.
By definition we have
\begin{align*}
  \rho = \rho_1 \cdots \rho_n
\end{align*}
The goal is to show that the function $w_1| \ldots | w_n \mapsto \rho$ is polyregular, where the number $n$ of words is not fixed, but their crossing type is. 

The following lemma is a version of Lemma~\ref{lem:good-concatenation} which allows us to compute runs on concatenations of two inputs from good languages, but with the head positioned between the inputs and not over the first position. If $w,v$ are partial inputs and $p$ is a state, then we write $w(pv)$ to denote the partial configuration where the head is in state $p$ over the first position of the input word $v$; we need the parenthesis since $(wp)v$ means something else, namely that the head is in the last position of $w$. 
\begin{sublemma}\label{sublemma:middle-q} 
If $L,K$ are good then there is a polyregular function which agrees with $\run$ on partial configurations of the form 
\begin{align*}
w(pv) \qquad \text{where }w \in K, v \in L, p \in Q.
\end{align*}
\end{sublemma}
\begin{proof}
The same reasoning, using crossing sequences, as in the proof of Lemma~\ref{lem:good-concatenation}. 
\end{proof}
The above sublemma also gives the implication from Lemma~\ref{lem:all-good} to Lemma~\ref{lem:run-is-polyregular}. 
By unfolding the definition of $\rho_i$, we see that 
\begin{align*}
  \rho_i = \run(w_1 \cdots w_{i-1} (q w_i))  \odot w_{i+1} \cdots w_n
\end{align*}
Using Sublemma~\ref{sublemma:middle-q}, we can compute $\rho_i$ for any fixed value of $i$, since the language
$L^{i-1}$ is good by  Lemma~\ref{lem:good-concatenation} applied several times. However, this approach yields a polyregular function which depends on $i$, while we need a uniform construction that does not depend on $i$, since the number runs $\rho_1,\ldots,\rho_n$ is unbounded.  To get the uniform construction, the  following observation is crucial.

\begin{sublemma}\label{claim:head-moves-little}
For every $w_1,\ldots,w_n \in L$  and every  $i \in \set{1,\ldots,n}$,  the run $\rho_i$ visits only words $w_j$ with $  i - j  \le |Q|$.
\end{sublemma}
\begin{proof}  A pumping argument, which uses the assumption that all words in $L$ have the same crossing type. Suppose that $\rho_i$ visits $w_{i-k}$ for some $k \in \set{1, 2,\ldots}$ such that $i-k>1$. By the assumption that all of the words $w_1,\ldots,w_n$ have the same crossing type, also $\rho_{i-1}$ visits $w_{i-k-1}$. A corollary is that if $\rho_i$ visits $w_1$, then the same is true for $\rho_{i-1}$. Since every input position can be visited at most once in each state, it follows that $\rho_{i}$ cannot visit the first position when $i$ exceeds the number of states, and the result follows\footnote{If we  assume that the fixed crossing type of all words in $L$ is idempotent, then we could get a stronger conclusion, namely that $\rho_i$ visits only $w_i$ and $w_{i-1}$. However, as explained in Footnote~\ref{foonote:non-idempotent}, the assumption on idempotency cannot be made in the first-order definable case.}.  \end{proof}

For words $w_1,\ldots,w_n$ and $i \in \set{1,\ldots,n}$, consider the decomposition of the word $w_1 | \cdots | w_n$ into four parts shown below
\begin{align*}
  \underbrace{w_1 |  \cdots |w_{j-1}|}_{x_i} \underbrace{w_{j}| \cdots |w_{i-1}|}_{y_i} \underbrace{w_i|}_{w_i} \underbrace{w_{i+1}| \cdots |w_n}_{z_i} \qquad j = \max(1, i - |Q|).
\end{align*}
In other words, $y_i w_i$ describes a window of the $\le |Q|$ blocks that contain the head positions of the  run $\rho_i$. 
By Sublemma~\ref{claim:head-moves-little}, 
\begin{align}\label{eq:four-parts}
  \rho_i =  x_i \odot  \run(y_i q_i w_i) \odot z_i
\end{align}
The partial input $y_i$  comes from a good language (after removing the separators), thanks to Lemma~\ref{lem:good-concatenation} iterated at most $|Q|$ times. (We need to use closure of good languages under unions, because the number of iterations could be $\le |Q|$, but closure under unions is an easy corollary of the if-then-else construction in Sublemma~\ref{lem:ifthenelse}.)  Therefore, we can use Sublemma~\ref{sublemma:middle-q} to compute the run $\rho_i$, assuming that the decomposition into $x_i, y_i,  w_i, z_i$ is given and the state $q_i$ is known. This decomposition can indeed be computed, thanks to the following result (as usual, we use red to denote a disjoint copy of the alphabet).

\begin{sublemma}\label{sublemma:windows}
	The following function is polyregular
	\begin{align*}
	w_1 | \cdots | w_n \quad \mapsto \quad x_1 \red{y_1 q_1 w_1} z_1  | \cdots | x_n \red{y_n q_n w_n} z_n
\end{align*}
where $w_1,\ldots,w_n \in L$ and $|$ is a fresh separator symbol. 
\end{sublemma}
\begin{proof}
	Similar to Sublemma~\ref{sublemma:iterated-concat}: first take a square, and then apply rational post-processing. The states $q_1,\ldots,q_n$ can be computed using a rational function (which is a first-order rational function in the case when the two-way automaton is first-order definable). Actually, the assumption that all words in $L$ have the same crossing type can be used to obtain a stronger result, namely that the sequence $q_1,q_2,\ldots$ is a lasso, in the following sense. There exist $k,k_0 \in \Nat$, which depend only on the crossing type of $L$, such that
	\begin{align*}
  q_i = q_{i+k} \qquad \text{for all }i \ge k_0.
\end{align*}
Furthermore, if the two-way automaton is first-order definable, then $k=1$, since otherwise there would be a counter.
\end{proof}

To complete the proof of the lemma, the final piece is the following result, which shows that polyregular  functions can be iterated over blocks in an input word with separators.
\begin{sublemma}\label{sublem:iterated-polyregular}
	If $f : \Sigma^* \to \Gamma^*$ is polyregular then the same is true for
	\begin{align*}
  w_1 | \cdots | w_n  \quad \mapsto \quad f(w_1)| \cdots | f(w_n)
\end{align*}
where $|$ is a fresh separator symbol not appearing in $w_1,\ldots,w_n$. 
\end{sublemma} 
\begin{proof}
It is enough to prove the lemma for the atomic polyregular functions. For sequential functions, the construction is natural, likewise for iterated reverse. For squaring, we illustrate the construction on an example. Suppose that the input is like this:
\begin{align*}
  12|3|45 
\end{align*}
We first use a sequential function to add a marker at the end (a comma), and then  apply squaring, yielding a result like this:
\begin{align*}
  \underline 12|3|45,   1 \underline  2|3|45,  12\underline |3|45,  12|\underline 3|45,  12|3\underline |45,  12|3|\underline 45,  12|3|4\underline 5,  12|3|45\underline , 
\end{align*}
Use a rational function to colour red every block (defined as maximal infix which has neither $|$ nor $,$) which contains an underlined position:
\begin{align*}
 \red{ \underline 12}|3|45,   \red{1 \underline  2}|3|45,  12\underline |3|45,  12|\red{\underline 3}|45,  12|3\underline |45,  12|3| \red{\underline 45},  12|3|\red{4\underline 5},  12|3|45\underline , 
\end{align*}
Keep only the red letters, and for every two consecutive red blocks where the  underlined position goes from first to last, separate the blocks using $|$.
\end{proof}

Let us complete the proof of the lemma. Suppose that the input is 
\begin{align*}
  w_1 | \cdots | w_n \qquad \text{with }w_1,\ldots,w_n \in L.
\end{align*}
Using Sublemma~\ref{sublemma:windows},  compute the word
\begin{align*}
x_1 \red{y_1 q_1 w_1} z_1  | \cdots | x_n \red{y_n q_n w_n} z_n
\end{align*}
By Sublemma~\ref{sublemma:middle-q}, the function
\begin{align*}
x_i \red{y_i q_i w_i} z_i \quad \mapsto \quad \rho_i  
\end{align*}
is polyregular. Applying Sublemma~\ref{sublem:iterated-polyregular} to this function,  compute 
\begin{align*}
  \rho_1 | \cdots | \rho_n.
\end{align*}
Finally, we remove the separators. 
\end{proof}

\subsection{Many pebbles}
\label{sec:many-pebbles}
In Section~\ref{sec:one-pebble}, we showed that a polyregular function can compute the run of a 1-pebble automaton. In this section, we lift the result to $k$-pebble automata, by reducing to the 1-pebble case. The reduction is given in  following lemma.

\begin{lemma}\label{lem:reduce-pebbles}
Let $k>1$. For every
\begin{align*}
  f : \Sigma^* \to \Gamma^*
\end{align*}
recognised by a $k$-pebble transducer
there exist:
\begin{enumerate}
\item a rational\footnote{Actually, the rational function in item 1 could be avoided  using Lemma~\cite[Lemma 3]{ullman1967approach}, which says that a two-way automaton can simulate regular lookaround.} function $g : \Sigma^* \to \Delta^*$; 
	\item a $1$-pebble automaton $\Aa$ with input alphabet $\Delta$;  
	\item a $(k-1)$-pebble transducer $h$ that inputs runs of $\Aa$;
\end{enumerate}
such that the following diagram commutes:
\begin{align*} 
	\xymatrix{ \Sigma^* \ar@[red][d]_{\red g}  \ar[rr]^f&  &\Gamma^* \\  \red{\Delta^*} \ar@[red][rr]_{\red {\run}} & &\red{\runs} \ar@[red][u]_{\red h}}
\end{align*}
where $\run$ is the function from Lemma~\ref{lem:run-is-polyregular} that maps an input word of $\Aa$ to the corresponding run of  $\Aa$, as described in Section~\ref{sec:one-pebble}.
Furthermore, if the automaton for $f$ is first-order definable, then the same is true for $g$, $\Aa$ and $h$.
\end{lemma}

Before proving the lemma, let us use it to show that  every pebble transducer recognises a polyregular function. The  proof is by induction on the number of pebbles. For the induction base, we use Lemma~\ref{lem:run-is-polyregular}, which says that a polyregular function can transform an input word into the run of a given 1-pebble automaton; once the run is given, the output of the run can be recovered using a rational function. For the induction step, we apply Lemma~\ref{lem:reduce-pebbles}, and observe that all three red arrows in the diagram from the lemma are polyregular: $g$  is polyregular because it is rational, $\run$  is polyregular by Lemma~\ref{lem:run-is-polyregular}, and $h$ is polyregular by induction assumption. It remains to prove the lemma.

\begin{proof}[of Lemma~\ref{lem:reduce-pebbles}] The basic idea is that a run of a $k$-pebble automaton can be decomposed into configurations that have only pebble 1, and intermediate subcomputations that do not move pebble 1. The intermediate subcomputations  can be simulated using $k-1$ pebbles. The rational function $g$ is used to determine how pebble 1 is moved, without entering into the subcomputations that use more pebbles. A more formal proof is given below.

Define a \emph{main configuration} in a $k$-pebble automaton to be a configuration  where only pebble 1 is present.  

\begin{sublemma}\label{lem:guiding-theory}
	 Let $\Bb$ be a $k$-pebble automaton with input alphabet $\Sigma$. 
There is a letter-to-letter rational function
\begin{align*}
g : \Sigma^* \to \Delta^*	
\end{align*}
such that the following is true for every $w \in \Sigma^*$.  Let $(q_i, x_i)$ be the state and head position in the $i$-th main configuration of the automaton in the run of $\Bb$ on input $w$. Then  
\begin{align*}
  \underbrace{x_{i+1} - x_i}_{\text{and offset in $\set{-1,0,1}$}} \qquad \text{and} \qquad q_{i+1}
\end{align*}
depend only on $q_i$ and the label of $g(w)$ in position $x_i$.  Furthermore, if $\Bb$ is first-order definable, then $g$ is first-order rational.
\end{sublemma}
\begin{proof}
Thanks to  Lemma~\ref{lem:all-pebble-mso}, for every state $q,p$ and offset $\delta \in \set{-1,0,1}$ there is an \mso sentence $\varphi_{q,\delta,p}(x)$ with one free first-order variable such that 
\begin{align*}
w \models \varphi_{q,\delta,p}(x)	
\end{align*}
if and only if the automaton $\Bb$, when in main configuration $(q,x)$, does a subcomputation such that the next main configuration is $(p,x+\delta)$. The  function $g$ simply labels each position $x$ with the set
\begin{align*}
\set{(p,\delta,q) : w \models \varphi_{p,\delta,q}(x)}	
\end{align*}
This function is rational, thanks to compositionality of \mso.  In the case when $\Bb$ is first-order definable, we do not need to use Lemma~\ref{lem:all-pebble-mso}, but we just appeal to the definition of a first-order definable pebble automaton.
\end{proof}
 
Using the function $g$ from the above sublemma, it is not hard to design a 1-pebble automaton $\Aa$, which has the same states as $\Bb$, with the following property. For every input word $w$, the $i$-th configuration of $\Aa$ (which is also the $i$-th main configuration, since all configurations are main configurations in a 1-pebble automaton) in the word $g(w)$  has the same state and head position as the $i$-th main configuration of $\Bb$ on input $w$.  Finally, to recover the output of $f$, we apply to each configuration of $\Aa$ the transducer from the following sublemma. 
\begin{sublemma}\label{sublemma:intermediate-subcomp}
	There is a $(k-1)$-pebble transducer which inputs a  configuration $c$ of $\Bb$, and returns the output of $\Bb$ in the subcomputation that starts in $c$  and ends in the next main configuration (not including the output produced in the next main configuration).
\end{sublemma}
\begin{proof}
By stack discipline, the subcomputation does not move pebble 1.
\end{proof}
\end{proof}

\section{Polyregular functions to list functions}
\label{sec:polyregular-to-programs}
In this section we prove that every  polyregular function (i.e.~any composition of sequential functions, squaring and iterated reverse, as described in Section~\ref{sec:polyregular}) is a polynomial list function (i.e.~can be defined in the functional programming language from Section~\ref{sec:list-programs}), and the construction preserves first-order definability. Polynomial list functions (and their first-order fragment) are closed under function composition: if $M$ and $N$ are polynomial list function, then their composition is the program
\begin{align*}
\lambda x. M (N x).
\end{align*}
Therefore, it remains to show that the basic building blocks of polyregular functions are polynomial list functions, which we do in the following sections:
\begin{itemize}
	\item Iterated reverse in Section~\ref{sec:reverse-to-haskell}; 
	\item Squaring in Section~\ref{sec:squaring-to-haskell};
	\item Sequential functions in Section~\ref{sec:rational-to-haskell}.
\end{itemize}
The main challenge is the sequential functions, since this requires showing that polynomial list functions can simulate the state updates in a finite state automaton. For this, we use a result from~\cite{DBLP:conf/lics/BojanczykDK18}, which in turn requires implementing a function called ${\blockterm}$, which separates a list which has two types of elements into a list of lists which have only one type of elements, as shown in the following example:
\begin{align*}
	[1,2,a,b,c,3,d,e,4,5,6] \quad \mapsto \quad [[1,2],[a,b,c],[3],[d,e],[4,5,6]]
\end{align*}

The code for the polynomial list functions described in this section, written in Haskell notation, is given in Appendix~\ref{sec:haskell}, and therefore this section only illustrates the programs on examples.  The interested reader is invited to read (or run) the code from the appendix.

\subsection{Iterated reverse} 
\label{sec:reverse-to-haskell}
We begin by showing that iterated reverse is a first-order polynomial list function. We first implement(not iterated) reverse (Lemma~\ref{lem:not-iterated-reverse}), and then we implement  $\blockterm$    (Lemma~\ref{lem:block}).
\begin{lemma}\label{lem:not-iterated-reverse}
	Reverse is a first-order polynomial list function.
\end{lemma}
\begin{proof} The corresponding Haskell code is in Appendix~\ref{sec:haskell-iterated-reverse}, so we just show the construction on an example.
Suppose that the input list is 
\begin{align*} 
[1,2,3,4]	
\end{align*}
Apply $\splitterm$, yielding a list like this:
\begin{align*}
[([1,2,3,4],[]),([1,2,3],[4]),([1,2],[3,4]),([1],[2,3,4]),([],[1,2,3,4])]	
\end{align*}
Project every pair in the above list to its second coordinate, yielding
\begin{align*}
[[],[4],[3,4],[2,3,4],[1,2,3,4]]	
\end{align*}
To each element of the above list,  apply the function 
\begin{align*}
[a_1,\ldots,a_n] \mapsto \begin{cases}
	[]  & \text{if $n=0$}\\
	[a_1] & \text{otherwise}
\end{cases}	
\end{align*}
yielding a list like this:
\begin{align*}
[[],[4],[3],[2],[1]]	.
\end{align*} 
Finally, apply $\flatterm$, yielding the desired list
\begin{align*}
	[4,3,2,1].
\end{align*}
\end{proof}

We  now show that polynomial list functions can implement the $\blockterm$ function mentioned at the beginning of Section~\ref{sec:polyregular-to-programs}. More formally, for types $\tau$ and $\sigma$, define  $\blockterm_{\tau \sigma}$ to be the function which maps a list $l \in (\tau + \sigma)^*$  to the unique list in $(\tau^* + \sigma^*)^*$  which: (a) yields $l$ after applying $\flatterm$; and (b) alternates between lists in $\tau^+$ and $\sigma^+$, in particular contains only nonempty lists. The function $\blockterm$  is one of the basic building blocks in the programming language from~\cite{DBLP:conf/lics/BojanczykDK18}, but it turns out to be implementable in the programming language from this paper. 

\begin{lemma}\label{lem:block}
	For every types $\tau$ and $\sigma$, the function  $\blockterm_{\tau \sigma}$ is a first-order polynomial list function.
\end{lemma}
\begin{proof}
The corresponding Haskell code is in Appendix~\ref{sec:haskell-iterated-reverse}, so we just show the construction on an example.
Suppose that the input list is 
\begin{align*}
[1,2,3,a,b,4,5,6,7,c,8,d,e,f]
\end{align*}
Apply $\splitterm$, and in the resulting list keep only those pairs $(x,y)$ such that $y$ is nonempty, and  either $x$ is empty, or the last element of $x$ has a different type ($\tau$ vs $\sigma$) than the first element of $y$, yielding a result like this:
\begin{align*}
([1,2,3,a,b,4,5,6,7,c,8],[d,e,f])],\\
([1,2,3,a,b,4,5,6,7,c],[8,d,e,f]),\\ 	
([1,2,3,a,b,4,5,6,7],[c,8,d,e,f]),\\
([1,2,3,a,b,4],[5,6,7,c,8,d,e,f]),\\
([1,2,3],[a,b,4,5,6,7,c,8,d,e,f]),\\
[ ([],[1,2,3,a,b,4,5,6,7,c,8,d,e,f])    \end{align*}
Reverse the list, see Lemma~\ref{lem:not-iterated-reverse}, and keep only the second coordinates:
\begin{align*}
[ [1,2,3,a,b,4,5,6,7,c,8,d,e,f],\\ [a,b,4,5,6,7,c,8,d,e,f],\\ [5,6,7,c,8,d,e,f],\\ [c,8,d,e,f],\\ [8,d,e,f],\\ [d,e,f]]	
\end{align*}
Finally, for each element in the result,  keep only the prefix that agrees on type ($\tau$ vs $\sigma$) with the first element of the list, yielding the desired result
\begin{align*}
[ [1,2,3],\\ [a,b],\\ [5,6,7],\\ [c],\\ [8],\\ [d,e,f]]	
\end{align*}
\end{proof}

Iterated reverse  is obtained by  applying the block function from the above lemma (with $\tau$ being the separator and $\sigma$ being the remaining letters), then using map and  reverse from Lemmas~\ref{lem:not-iterated-reverse}, and finally applying $\flatterm$.   

\subsection{Squaring}
\label{sec:squaring-to-haskell}  We now show that squaring is a first-order polynomial list function.  
\begin{lemma}\label{lem:squaring-is-program}
Squaring is a  first-order polynomial list function.
\end{lemma}
\begin{proof}The corresponding Haskell code is in Appendix~\ref{sec:haskell-squaring}, so we only illustrate the code using an example.  Suppose that the input is a list like this
\begin{align*}
  [1,2,3,4]
\end{align*}
Apply $\splitterm$,  yielding a list like this
\begin{align*}
[([1,2,3,4],[]),([1,2,3],[4]),([1,2],[3,4]),([1],[2,3,4]),([],[1,2,3,4])]\end{align*}
To every pair in the above list,  apply the function
\begin{align*}
(a,b) \quad \mapsto \quad \begin{cases}
 [] & \text{if $b$ is empty}\\
 [(a, \text{head of $b$}, \text{tail of $b$})] & \text{otherwise}
 \end{cases}	
\end{align*}
yielding a list like this
\begin{align*}
[[],[([],1,[2,3,4])], [([1],2,[3,4])], [([1,2],3,[4])], [([1,2,3],[4],[])]]\end{align*}
Flatten the above list using $\flatterm$, yielding
\begin{align*}
[([],1,[2,3,4]), ([1],2,[3,4]), ([1,2],3,[4]), ([1,2,3],[4],[])]\end{align*}
Next,  apply a reformatting function to yield a list
\begin{align*}
[[\underline 1,2,3,4], [1, \underline  2,3,4], [1,2,\underline  3,4], [1,2,3,\underline 4]]	
\end{align*}
where $\underline n$ stands for copy of letter $n$ in a disjoint copy of the alphabet. 
\end{proof}

\subsection{Sequential functions} 
\label{sec:rational-to-haskell}
We are left with showing that every sequential function can be implemented using a polynomial list function, and if the sequential function is first-order rational then a first-order polynomial list function is enough.

There are two ways to solve this problem.

The first way is to use Theorem~\ref{thm:krohn-rhodes}, which says that we only need to do the construction for  sequential  functions where the underlying automaton either: (a) has two states and is counter-free; or (b) has a transformation monoid which is a group. Both kinds of functions can be easily implemented using polynomial list functions, and the construction for (a) needs only a first-order polynomial list function. The construction for (a) is implemented using the $\blockterm$ function discussed in Lemma~\ref{lem:block}.

The second way is to use regular list functions defined in~\cite{DBLP:conf/lics/BojanczykDK18}. We discuss the second way in more detail, because it  highlights the relationship between the polynomial list functions from this paper, and the class of regular list functions, which can be seen as the linear growth fragment of polynomial list functions.
 We begin by defining the class of functions considered in~\cite{DBLP:conf/lics/BojanczykDK18}. 

\begin{definition}[Regular and first-order list functions]  \label{def:regular-list-functions} The class of  \emph{regular list functions} is the smallest class of functions which
	\begin{itemize}
		\item contains all functions described in Figure~\ref{fig:atomic-linear};
		\item is closed under the combinators described in Figure~\ref{fig:combinators-linear}.
	\end{itemize}
	The  \emph{first-order list functions} are the ones that can be constructed without using the  group products in Figure~\ref{fig:atomic-linear}.
\end{definition}

\newcommand{\simpleratomic}[2]{#1 & : & #2 \\}

\begin{figure}
\begin{eqnarray*}
\simpleratomic{\firstterm^{\tau \sigma}}{(\tau \times \sigma) \to \tau}
\simpleratomic{\secondterm^{\tau \sigma}}{(\tau \times \sigma) \to \tau}
\simpleratomic{\leftterm^{\tau \sigma}}{\tau \to (\tau + \sigma)}
\simpleratomic{\rightterm^{\tau \sigma}}{\tau \to (\tau + \sigma)}
\simpleratomic{\flatterm^{\tau }}{(\tau^*)^* \to \tau^*}
\simpleratomic{\headterm^{\tau}}{\tau^* \to (\tau + \set \bot)}
\simpleratomic{\tailterm^{\tau}}{\tau^* \to (\tau^* + \set \bot)}
\\
 \atomicfunction{\blue{\distrterm^{\tau \sigma \pi}}}{\tau \times (\sigma + \pi) \to (\tau \times \sigma) + (\tau \times \pi)}{distribute $+$ across $\times$}
 \atomicfunction{\blue{\reverseterm^{\tau}}}{\tau^* \to \tau^*}{reverse the input list}
\atomicfunction{\blue{\constterm^a}}{\tau \to \sigma}{return $a$ for every argument}
	\atomicfunction{\blue{\appendterm^\tau}}{(\tau \times \tau^*) \to \tau^*}{add first argument to the left of the list in the second argument }
	\atomicfunctionbis{\blue{\blockterm^{\tau \sigma }}}{(\tau+\sigma)^* \to (\tau^* + \tau^*)^*}{group into maximal blocks of type $\tau^*$ or $\sigma^*$}{[1,a,3,4,b,c] $\mapsto$ [[1],[a],[3,4],[b,c]] }
	\atomicfunctionbis{\blue{\groupterm^*_G}}{G^* \to G^*}{the $i$-th element of the output list is the product}{of the first $i$ elements in the input list}
\end{eqnarray*}
  \caption{\label{fig:atomic-linear}Atomic linear first-order list functions. The types $\tau, \sigma$ are required to be arrow-free, i.e.~are constructed from finite sets using $\tau + \sigma$, $\tau \times \sigma$ and $\tau^*$. The functions in black are already present in Figure~\ref{fig:atomic-combinators-1}, while the functions in \blue{blue} are not present, but can be derived.}
\end{figure}

\begin{figure}
\begin{align*}
  \frac{f: \tau \to \sigma \quad g :\sigma \to \pi}{ f \circ g : \tau \to \pi} \qquad \red{x \mapsto f(g(x))}
\end{align*}

\begin{align*}
  \frac{f_1: \tau_1 \to \sigma \quad f_2 :\tau_2 \to \sigma}{ 
  \langle f_1,f_2 \rangle : (\tau_1 + \tau_2) \to \sigma} \qquad \red{x \mapsto \begin{cases}
  	f_1(x) & \text{if $x \in \tau_1$} \\
  	  	f_2(x) & \text{if $x \in \tau_2$} 
  \end{cases}}
\end{align*}

\begin{align*}
  \frac{f_1: \tau \to \sigma_1 \quad f_2 :\tau \to \sigma_2}{ 
  (f_1,f_2) : \tau \to (\sigma_1 \times \sigma_2)} \qquad \red{x \mapsto (f_1(x), f_2(x))}
\end{align*}

\begin{align*}
  \frac{f: \tau \to \sigma}{ 
  f^* : \tau^* \to \sigma^*} \qquad \red{[x_1,\ldots,x_n] \mapsto [f(x_1),\ldots,f(x_n)]}
\end{align*}

  \caption{\label{fig:combinators-linear} Closure properties of linear list functions.}
\end{figure}

By definition, every regular list function (in particular, every first-order one), has a type of the form $\tau \to \sigma$ where $\tau$ and $\sigma$ are arrow-free, i.e.~are constructed from finite sets using only the type constructors $\tau + \sigma$, $\tau \times \sigma$ and~$\tau^*$.  In this sense, regular and first-order list functions cannot use higher-order types.  Another key design property is that the atomic functions from Figure~\ref{fig:atomic-linear} have linear growth and the combinators in~\ref{fig:combinators-linear}  preserve this property.  This is because regular list functions are designed to capture the regular string-to-string functions, which have linear growth. The linear growth also explains why $\lambda$-abstraction or $\splitterm$ are not allowed in regular list functions, since they could be used to generate functions of super-linear growth, see Example~\ref{ex:mapterm}.

\begin{lemma}\label{lem:linear-in-polynomial}
	Every  regular list function is a polynomial list function. Likewise for the first-order fragment.
\end{lemma}
\begin{proof}
	Clearly polynomial list functions and their first-order fragment have the closure properties described in Figure~\ref{fig:combinators-linear}.
	The atomic functions in Figure~\ref{fig:atomic-linear} are either already atomic polynomial list functions (the ones in black) or can be implemented using polynomial list functions (the ones in \blue{blue}). The implementations of $\reverseterm$ and $\blockterm$ were given in Lemmas~\ref{lem:not-iterated-reverse} and~\ref{lem:block}, while the implementations of the remaining functions
\begin{align*}
\blue{\distrterm_{\tau \sigma \pi} \qquad \constterm_a \qquad  \appendterm_\tau  \qquad \groupterm^*_G}
\end{align*}
are straightforward and  left to the reader. Note  the difference of the group operations \begin{align*}
\groupterm_G : G^* \to G \qquad   \blue{\groupterm^*_G : G^* \to G^*}
\end{align*} 
that are used in Figures~\ref{fig:atomic-combinators-1} and~\ref{fig:atomic-linear} respectively. These can be defined in terms of each other:  $\groupterm_G$  can be defined in terms of $\blue{\groupterm^*_G}$  using head and reverse, and a converse construction can be done using  $\splitterm$. However, $\splitterm$ is not available in Figure~\ref{fig:atomic-linear}, and hence the more powerful $\blue{\groupterm^*_G}$ is used as an atomic function in Figure~\ref{fig:atomic-linear}. One can, in fact, show that replacing $\blue{\groupterm^*_G}$ by $\groupterm_G$ in Figure~\ref{fig:combinators-linear} would lead to a weaker class of functions.
\end{proof}

The main result about regular list functions and their first-order fragment is that they correspond to the class of regular string-to-string functions.  Consider string-to-string functions, i.e.~functions of type\footnote{There is a notation clash here. Finite sets in automata theory are denoted using uppercase letters $\Sigma, \Gamma$,  while  $\lambda$-calculus typically uses  lowercase letters $\sigma,\tau$ for types, which includes the case of finite sets. We use lowercase when typing programs from the $\lambda$-calculus, and uppercase letters in other situations. } 
\begin{align*}
f : \sigma^* \to \tau^*	 \qquad \text{where $\sigma,\tau$ are finite sets.}
\end{align*}
When restricted to string-to-string functions, 
the regular  list functions are exactly the same as  as \mso string-to-string transductions~\cite[Theorem 6.1]{DBLP:conf/lics/BojanczykDK18}. Since  \mso string-to-string transductions are the same as functions recognised by  two-way transducers (i.e.~1-pebble transducers)~\cite[Theorem 13]{engelfriet2001mso}, it follows that for string-to-string transducers, the regular list functions are the same as 1-pebble transducers. 
This equivalence also works for the first-order case: the first-order list functions are exactly the same as 
first-order transductions~\cite[Theorem 4.3]{DBLP:conf/lics/BojanczykDK18}, and  first-order string-to-string  transductions are the same as functions recognised by first-order definable 1-pebble automata~\cite{Carton:2015bl}. Putting these results together, we get the following theorem.

\begin{theorem}\label{thm:folf} For  string-to-string  transducers, i.e.~functions of type
	\begin{align*}
f : \sigma^* \to \tau^*	 \qquad \text{where $\sigma,\tau$ are finite sets,}
\end{align*}
the regular list functions are exactly the same as 1-pebble transducers. The first-order list functions are exactly the same as  the first-order definable 1-pebble transducers.
 \end{theorem}
 
Putting together Lemma~\ref{lem:linear-in-polynomial} and Theorem~\ref{thm:folf}, we see that every 1-pebble transducer is a polynomial list function, and this construction preserves first-order definability. Since sequential functions are recognised by 1-pebble transducers, we obtain that every rational function is a polynomial list function (and the construction preserves first-order definability).  This argument would have also worked for iterated reverse, but the main work in Section~\ref{sec:reverse-to-haskell} consisted of coding non-iterated reverse and $\blockterm$, which were necessary to get Lemma~\ref{lem:linear-in-polynomial}.

This completes the proof that all polyregular functions are polynomial list functions (and the corresponding first-order result). 

\paragraph*{Combinators.} The regular list functions  (i.e.~the ones corresponding to 1-pebble automata) from Definition~\ref{def:regular-list-functions} are  \emph{combinatory}  in the sense that they do not have any variables and $\lambda$ construction. We can also identify a combinatory syntax for  polynomial list functions:  define a \emph{combinatory polynomial list function} to be any function generated by the rules from Definition~\ref{def:regular-list-functions} plus  $\splitterm$. It is not hard to see that squaring is a combinatory polynomial list function; and the same is true for iterated reverse and sequential functions (here $\splitterm$ is not needed) as we have seen above. It follows that
\begin{align*}
	\text{polyregular functions} \quad \subseteq \\ \text{combinatory polynomial list functions} \quad \subseteq \\ \text{polynomial list functions} \quad \ \ \ 
\end{align*}
Furthermore, since the first and third lines above are equal, it follows that the polynomial list functions collapse to their combinatorial fragment, i.e.~$\lambda$ abstraction can be eliminated without affecting the power of the programming language, at least as long as string-to-string functions are concerned.

\section{List functions to for-transducers}
\label{sec:program-to-for}
In this section, we show that every string-to-string function computed by a  polynomial list function can also be computed by a for-transducer, and the translation preserves first-order definability. 

In the first part of the proof, Section~\ref{sec:for-composition}, we show that string-to-string functions computed by for-transducers are closed under composition. The idea is similar to the composition closure for pebble transducers: if $f,g$ are for-transducers, then a for-transducer computing the composition $f \circ g$ is the same as the for-transducer $f$, except that instead of positions it uses  configurations of $g$. 

In the second and main part of the proof, Section~\ref{sec:beta-reduction}, we 
use a  semantics of polynomial list functions based on term rewriting.  The key idea is to do $\beta$-reduction in parallel, e.g.~if there is a list of terms, then one step of $\beta$-reduction can be applied in parallel to each term on the list. By doing $\beta$-reduction in parallel,  only a bounded number   number of rewriting steps is needed to compute the value of a term. Since a single parallel $\beta$-reduction step can be computed by a for-transducer, and for-transducers are closed under composition, it follows that a for-transducer can compute the value of a polynomial list program.

\subsection{Composition of for-transducers}
\label{sec:for-composition}
The goal of this section is to prove the following lemma. 
\begin{lemma}\label{lem:for-composition}
    String-to-string functions recognized by for-transducers are closed under composition. The construction preserves first-order definability.
\end{lemma}
\begin{proof}
  Consider two for-transducers
\begin{align*}
  \xymatrix{\Sigma^* \ar[r]^f & \Gamma^* \ar[r]^g & \Delta^*}
\end{align*}
Our goal is to show that the composition 
\begin{align*}
  \xymatrix{\Sigma^* \ar[r]^{g \circ f} & \Delta^*}
\end{align*}
is also a for-transducer, and if both $f,g$ were first-order, then the same is true for $g \circ f$.  
Recall the prenex normal form that was defined in Section~\ref{sec:for-to-pebble}. 
Using Lemma~\ref{lem:for-normal-form}, we can  assume that both $f$ and $g$  are in prenex normal form. An example of $f$ and $g$ is given in Figure~\ref{fig:composition-for-transducers}, the for-transducer for their composition is illustrated in Figure~\ref{fig:composite}. 

\begin{figure}
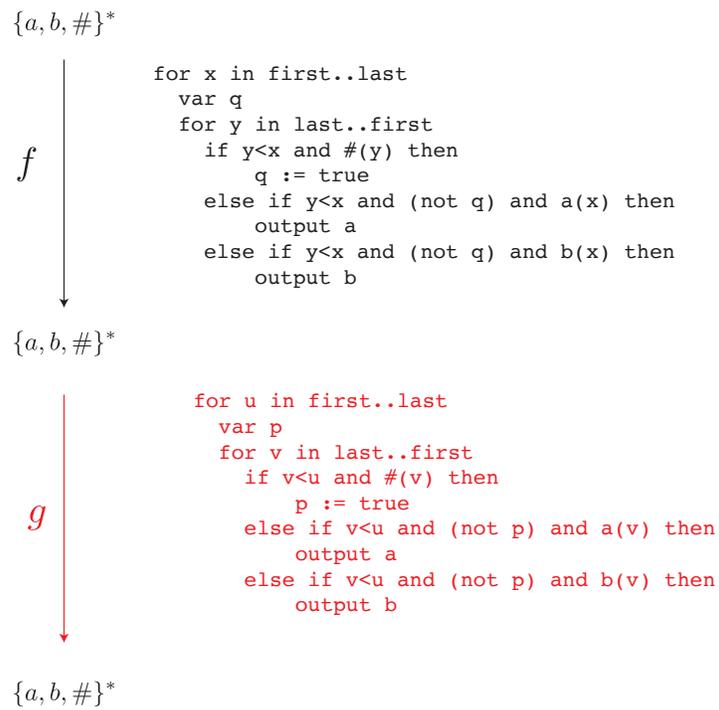

  \mypic{32}
  \caption{\label{fig:composition-for-transducers} Example instance of  composition of two for-transducers. In this example,  both for-transducers define  the same function, namely iterated reverse, only using renamed variables. }
\end{figure}

Suppose that the input to the composition   is some word $w \in \Sigma^*$. Our goal is to compute $g(f(w))$  using a single for-transducer that simulates the run of  $g$ over an input of the form $f(w)$. The idea is the same as in the proof of Theorem~\ref{thm:pebble-closed-under-composition}: use the same code as  $g$, except that positions in  the intermediate word  are represented by  configurations of $f$ that were  used to produce them. 

Consider the first for-transducer $f$, and let $n$ be the number for loops that it uses.  For an input word $w \in \Sigma^*$ and an $n$-tuple of positions $x_1,\ldots,x_n$, define
\begin{align*}
  f(w,x_1,\ldots,x_n) \in \Gamma^* 
\end{align*}
to be the letters that are output by the kernel of $f$ (recall that the kernel of a for-transducer in prenex normal form is the part that does not use for loops) in the iteration of the for loops where the position variables are set to $x_1,\ldots,x_n$. Since the kernel has no loops, the above word has fixed length. To simplify the proof of Lemma~\ref{lem:for-composition}, we  assume that this fixed length is at most one, i.e.~in each iteration of the innermost loop of $f$  at most one letter is produced. This assumption means that  each position in the intermediate  word $f(w)$ is represented uniquely by a tuple of position variables of $f$ in the input word $w$.  The proof without this additional assumption requires additional notation, but follows the same lines, and is left to the reader. 

As in the proof of Lemma~\ref{lem:pebble-to-for-formula}, we order $n$-tuples of positions in the input word $w$  using  a lexicographic order  $\preceq$, with the $i$-th coordinate being ordered first-to-last or last-to-first depending on the type of the $i$-th for loop. By definition 
\begin{align*}
f(w) = \prod_{x_1,\ldots,x_n}  f(w,x_1,\ldots,x_n) 
\end{align*}
where $\prod$ stands for concatenation, with $n$-tuples  $(x_1,\ldots,x_n)$ ordered by $\preceq$. 

The code of the   for-transducer $g \circ f$ is the same as the one for $g$, except that the for loops, the order tests {\tt x<y} and the label tests {\tt a(x)} are modified as follows, to account for the representation of positions in $f(w)$ via $n$-tuples of positions in $w$.

\begin{enumerate}
  \item Suppose that  $g$ does a loop of one of the two types below (which are the only types allowed in a for-transducer in prenex normal form):
  \begin{align}
    \underbrace{\tt for\ u\ in\ first..last}_{\text{first to last}} \qquad 
    \underbrace{\tt for\ u\ in\ last..first}_{\text{last to first}}
    \label{eq:for-loop-transducer}  
  \end{align}
  To simulate this loop, the  for-transducer computing $g \circ f$  enumerates through all $n$-tuples  of positions in the input word $w$, ordered according to the ordering $\preceq$ -- in  the first-to-last case --  or the opposite of $\preceq$ -- in the last-to-first case. This enumeration is done using $n$ nested for loops.   For each $n$-tuple of positions $(x_1,\ldots,x_n)$ from this enumeration, the for-transducer computing $g \circ f$  checks  if the output
  \begin{align*}
    f(w,x_1,\ldots,x_n) \in \Gamma^* 
  \end{align*}
  is  nonempty.  This check is done by running a fresh copy of $f$ from scratch, and waiting until it reaches configuration $(x_1,\ldots,x_n)$.  If the above output is empty, then the body of the for loop in~\eqref{eq:for-loop-transducer} is skipped, otherwise it is executed, with the label and order tests simulated as described below. 
  \item Suppose that $g$ tests the order $x \le y$ on two positions in the intermediate word $f(w)$. In the  for-transducer computing  $g \circ f$, these positions are represented as two $n$-tuples $(x_1,\ldots,x_n)$ and $(y_1,\ldots,y_n)$, and the corresponding comparison is $\preceq$, which is a Boolean combination of comparisons on the positions $x_i$ and $y_i$, and therefore can be implemented by a for-transducer.
  \item Suppose that $g$ tests the label $a(x)$ of a position $x$ in the intermediate word.  In the  for-transducer $g \circ f$, this position is represented as an $n$-tuple $(x_1,\ldots,x_n)$. To determine the label of this position, the  transducer $g \circ f$ does the same trick as in item 1: it runs a fresh copy of $f$ from scratch and waits until it reaches configuration $(x_1,\ldots,x_n)$. 
\end{enumerate}
\end{proof}

\begin{figure}
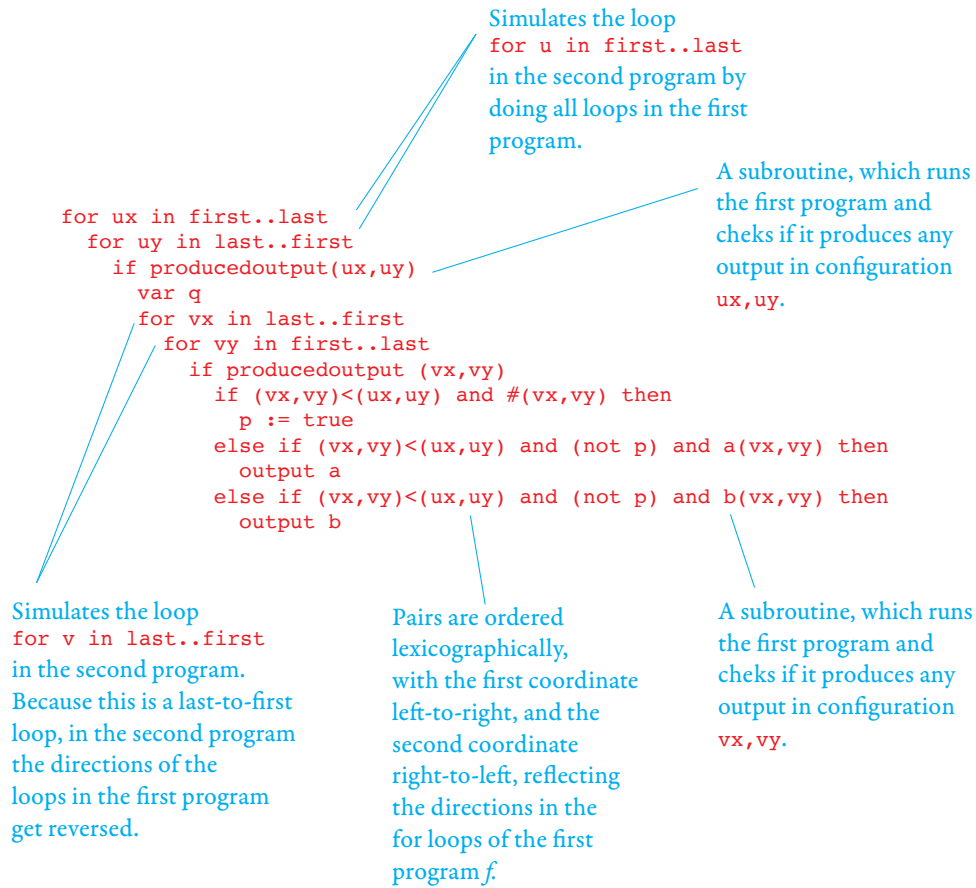

  \mypic{33}   
  \caption{\label{fig:composite} The  for-transducer computing the composition of the functions described in Figure~\ref{fig:composition-for-transducers}.}
\end{figure}
 
Before continuing with the proof that for-transducers can simulate polynomial list functions, we observe that the composition result above can be used to show equivalence of for-transducers with pebble transducers and polyregular functions.   
For-transducers are easily seen to recognize all of the atomic polyregular functions, i.e.~iterated reverse, squaring and sequential functions. Also,  first-order for-transducers are enough when only first-order sequential functions are used, see e.g.~Example~\ref{ex:for-transducers-do-formulas}. Therefore, the composition closure result from Lemma~\ref{lem:for-composition} implies that for-transducers recognize all polyregular functions, and  the conversion preserves first-order definability. Combining this with the results from the previous two sections, we get the equivalence of the following three models:

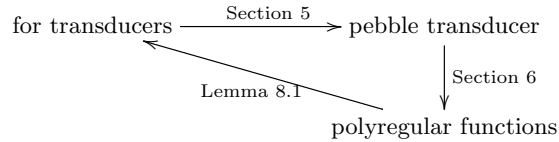
\begin{align*}
  \xymatrix{
     {\small \text{for transducers}} \ar[rr]^{\text{Section~\ref{sec:for-to-pebble}}} & & {\small \text{pebble transducer}} \ar[d]^{\text{Section~\ref{sec:pebble-to-polyregular}}}  \\
   & &   {\small \text{polyregular functions}} \ar[llu]^{\text{Lemma~\ref{lem:for-composition}}}
  } 	
  \end{align*}

\subsection{For-transducers implementing $\beta$-reduction}
\label{sec:beta-reduction}
In this section we complete the proof that  polynomial list programs can be simulated by for-transducers. The idea is to use a term rewriting approach to the  semantics of polynomial list functions, and to show that this approach can be simulated by a for-transducer.

\begin{definition}
  [$\beta$-reduction] \label{def:beta-reduction}Define $\beta$-reduction to be the binary relation on terms, denoted by 
  \begin{align*}
    M \to_\beta N,
  \end{align*}
  which holds if $N$ can be obtained from $M$ by applying one of the following reduction rules to a subterm:
  \begin{align*}
    \begin{array}{rcl}
      \reductionrule{\isterm}{type}{\blue{\isterm_a} b }{\begin{cases}
        \mathtt{true} & \text{if $a=b$}\\
        \mathtt{false} & \text{if $a\neq b$}.
      \end{cases}}
      \reductionrule{\lambda}{\tau \to \sigma}{\churchpar{\blue{\churchpar{(\lambda \church x \tau \church M \sigma )}{\tau \to \sigma}} \church N{\tau}} \sigma}{\churchpar{\church M \sigma [\church x \tau := \church N \tau]}\sigma \quad\text{\red{if $x$ is not bound in $M$}} }
  \reductionrule{\mapterm}{(\tau \to \sigma) \to \tau^* \to \sigma^*}{\churchpar{\church {\blue{\mapterm\ }} {(\tau \to \sigma) \to \tau^* \to \sigma^*} \church M {\tau \to \sigma} \church{[\church {N_1} \tau,\ldots,\church{N_k}\tau ]}{\sigma^*}} {\sigma^*}}{
    \church{[\churchpar{\church M {\tau \to \sigma} \church{N_1} \tau} \sigma ,\ldots,\churchpar{\church M {\tau \to \sigma} \church{N_k} \tau} \sigma]}{\sigma^*}
  }
  \reductionrule{\projterm \ }{(\tau_0 \times \tau_1) \to \tau_i}{\churchpar{\church {\blue{\projterm i}\ } {(\tau_0 \times \tau_1) \to \tau_i} \church{(\church{M_0} {\tau_0}, \church{M_1} {\tau_1}}{\tau_0 \times \tau_1})}{\tau_1}}{\church {M_i}{\tau_0} \quad\text{\red{for $i \in \set{0,1}$}}}
  \reductionrule{\caseterm}{(\tau_0 \to \sigma) \to (\tau_1 \to \sigma) \to  (\tau_0 + \tau_1) \to \sigma }{
  {\blue \caseterm} \ M_0 \ M_1\ (\coprojterm i N) 
  }
  {\church{\church {M_i}{\tau_i \to \sigma} \church N {\tau_i}}\sigma  \quad\text{\red{for $i \in \set{0,1}$}}}
  \reductionrule{\splitterm}{type}{{\blue \splitterm} \ [M_1,\ldots,M_n]} {[([],[M_1,\ldots,M_n]),\ldots,([M_1,\ldots,M_n],[]
  )]}
  \reductionrule{\headterm}{type}{{\blue \headterm}\ [M_1,\ldots,M_n]}{\rightterm\ M_1 \quad\text{\red{for $n \ge 1$}} }
  \reductionrule{\headterm}{type}{{\blue \headterm}\ []}{\leftterm\ \bot }
  \reductionrule{\tailterm}{type}{{\blue \tailterm}\ [M_1,\ldots,M_n]}{\rightterm\ [M_2,\ldots,M_n] \quad\text{\red{for $n \ge 1$}}  }
  \reductionrule{\tailterm}{type}{{\blue \tailterm}\ []}{\leftterm\ \bot}
    \end{array}
  \end{align*}
  In the rules above, the blue colour highlights the first term of the left hand side of a rule, which is called the \emph{leading term} of the rule (we use leading terms to define a reduction strategy).    
\end{definition} 
It is easy to see that    $\beta$-reduction  changes neither the type nor the  semantics of a term.
Define a \emph{redex} in a term $M$ to be a subterm to which a reduction rule can be applied, and define the \emph{type} of a redex to be the type of its main term. An example of a term and its redexes is shown in Figure~\ref{fig:redexes}. 

The \emph{normal form} of a term $M$ is a term that has no redexes, and which can be obtained from $M$ by finitely many steps of $\beta$-reduction and renaming bound variables (the latter process is  called $\alpha$-conversion)\footnote{One can show that  $\beta$-reduction for our terms is well-founded and has the Church-Rosser property; the reasons is that  our calculus can be embedded in System {\bf F}, see  Section 11.3 in~\cite{sorensen2006lectures}. This implies that each term has a unique normal form up to renaming of bound variables, and that this normal form is reached regardless of the order in which $\beta$-reduction is applied. However, the uniqueness of normal forms and strong normalisation are  not needed for this section, so we do not  prove them. Nevertheless, we acknowledge  the uniqueness of normal forms,  by talking about  ``the'' normal form of a term as opposed to ``a'' normal form. }.
For terms which represent strings -- i.e.~terms that have  type $\tau^*$ where $\tau$ is a finite set -- there is  no difference between a value -- i.e.~a string -- and a term in normal form, as given in the following lemma.  The lemma also implies the uniqueness of normal forms for terms representing strings. 
\begin{lemma}\label{lem:normal-form-is-value}
  Let  $M$ be a term in normal form without free variables, whose type is $\tau^*$ for some finite set $\tau$. Then $M = [a_1,\ldots,a_n]$ for some $a_1,\ldots,a_n \in \tau$. 
\end{lemma}
\begin{proof}
  A standard and folklore case analysis, see Appendix~\ref{sec:values}.     
\end{proof}


To compute the output of a polynomial list program $M : \tau^* \to \sigma^*$   on an input string $[a_1,\ldots,a_n] \in \tau^*$, we do the following:
\begin{enumerate}
  \item produce the term $M\ [a_1,\ldots,a_n]$ by prepending $M$ to the input word;
  \item convert the term from step 1 into normal form;
  
\end{enumerate}
Since $\beta$-reduction does not change the semantics of terms, and normal form terms of string type are the same as strings thanks to Lemma~\ref{lem:normal-form-is-value},  it follows that the normal form of $M\ [a_1,\ldots,a_n]$ is the same as the output of $M$ on input $[a_1,\ldots,a_n]$, as defined in Section~\ref{sec:list-programs}.  To prove that the above steps can be simulated by a for-transducer, we show that there is a reduction strategy (a choice of which redexes to reduce first, with possibly several redexes being reduced in parallel) such that (a) one step of the reduction strategy can be implemented by a for-transducer; and (b) there is a constant $k$ depending only on $M$  such that for every input string $[a_1,\ldots,a_n]$ at most $k$ steps of the reduction strategy are needed to convert $M \ [a_1,\ldots,a_n]$ into normal form. Since for-transducers are closed under composition by Lemma~\ref{lem:for-composition}, it follows that every string-to-string function recognised by a polynomial list program is also recognised by a for-transducer. We begin by explaining how terms are represented as strings so  that they can be handled by a for-transducer.

\begin{figure}
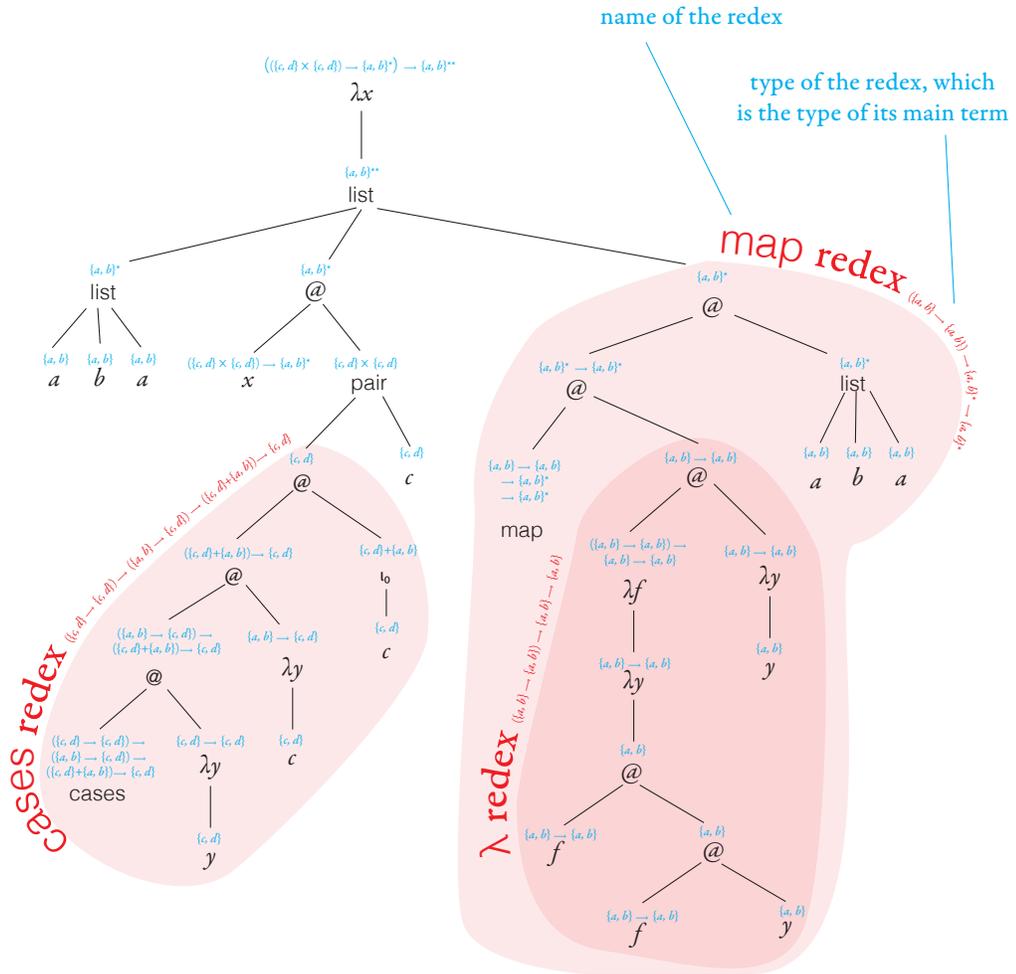

  \mypic{4}
  \caption{\label{fig:redexes}A term with three redexes. }
\end{figure}

\paragraph*{For-transducers manipulating terms.}  The difficulties with representing terms as strings come from: (a) renaming bound variables, (b) having complicated  types in subterms, and (c) dealing with the tree structure of a term. To get around these  difficulties, we assume  a bound on the depth of terms and on the set of types allowed in subterms, as given in the following definition.

\begin{definition}\label{def:bounded-terms} For a finite set of types 
    $\Gamma$  and   $k \in \set{1,2,\ldots}$, define 
  \begin{align*}
    \terms(\Gamma,k)
  \end{align*}  
  to be the  terms which satisfy the following conditions:
  \begin{enumerate}
    \item the depth (length of the longest path in its syntax tree) is at most $k$;
    \item  all    subterms   have types in  $\Gamma$; 
    \item for each type $\tau \in \Gamma$ there are at most $k$ variable names, call them  $x^\tau_1,\ldots,x^\tau_k$.
  \end{enumerate}
  
\end{definition}

Note that the size (number of nodes in the syntax tree) of a term can be unbounded even for terms of depth two, since  the depth of a list of terms is one plus the maximal depth of terms on the list.
The restriction on the variable names in item 3 of Definition~\ref{def:bounded-terms} is not important for the semantics, because  bound variables in a term of depth at most $k$ can be renamed so that there are at most $k$ variable names of each type (or even $k$ variables altogether, if we would allow the same variable name to be used for different types). Nevertheless, having a finite set of variable names  makes it possible to represent the terms as strings over a finite alphabet, which consists of the variable names, the atomic functions, $\lambda$, parentheses and commas. 
Because the depth of terms is bounded, a finite automaton can check if a string represents some term, in particular it can check if the parentheses are well matched. 
Using such a string  representation, we can talk about a function
\begin{align*}
  f : \terms(\Gamma,k) \to   \terms(\Gamma,m) \qquad \text{for $k,m \in \Nat$}
\end{align*}
being recognized by a string-to-string device, such as a for-transducer. This notion is used in the following lemma, which is the main result of this section.

\begin{lemma}\label{lem:red-for}
  For every  finite set of types $\Gamma$  and  $k \in \set{1,2,\ldots}$ there exists
  \begin{align*}
  m \in \set{1,2,\ldots} \qquad \text{and} \qquad   f : \terms(\Gamma,k) \to   \terms(\Gamma,m)
  \end{align*}
  such that $f$ is recognized by a for-transducer, and maps  every input term to its   normal form. For terms without the group product operation, a first-order for-transducer is enough.
\end{lemma}
In particular, the lemma says that the depth of normal forms is bounded by a constant $m$, once the $k$ and $\Gamma$ have been fixed. This constant is a tower of exponentials, whose height is linear in the size of $\Gamma$, and this cannot be avoided, see   Section 3.7 in~\cite{sorensen2006lectures}.

The above lemma,  together with Lemma~\ref{lem:normal-form-is-value} about normal forms of terms of string type, completes the proof that for-transducers can simulate polynomial list programs. The rest of Section~\ref{sec:program-to-for} is devoted to proving the above lemma.  We define a specific reduction strategy and show that it  terminates in a constant number of steps. The key point is that the reduction strategy is parallel, i.e., multiple redexes are reduced at the same time.

\paragraph*{The reduction strategy.} The reduction strategy  is adapted from  Theorem 3.5.1 in~\cite{sorensen2006lectures}, with the modification that it uses parallel reduction when possible. Likely  other reduction strategies would work as well, e.g.~reducing all innermost or all outermost redexes.  Fix $\Gamma$ and $k$ in the assumptions of Lemma~\ref{lem:red-for}. Define the \emph{degree} of a type to be the depth of its syntax tree, as illustrated below:
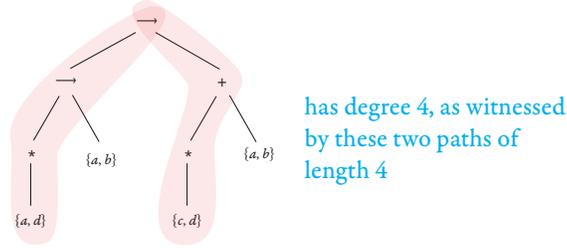
\mypic{42} 
Define the \emph{type} of a redex in a term to be the type of its leading term, see the remarks at the end of Definition~\ref{def:bounded-terms} and the example in Figure~\ref{fig:redexes}.  Define a \emph{maximal redex} in a term $M$ to be any redex such that (a) the type $\tau$ of the redex has maximal degree among the types of other redexes in the same term; (b) the redex is maximally deep in the sense that  there are no redexes of the same degree in its subtree. A term might have multiple maximal redexes, but they are incomparable in the tree ordering thanks to condition (b), so it makes sense to reduce them in parallel. Our reduction strategy is to reduce all maximal redexes in parallel: for a term $M$, define $\reduce M$ to be the result of reducing in parallel all maximal redexes (and doing nothing if there are no redexes).  It is easy to see that the depth of $\reduce M$ is at most twice the depth of $M$. The doubling of depth occurs when reducing application with $\lambda$, other redexes increase the depth by at most 2. 

The following lemma, together with the closure of for-transducers under composition from Lemma~\ref{fig:composition-for-transducers}, completes the proof of Lemma~\ref{lem:red-for}, and therefore also the proof that polynomial list-programs can be simulated by for-transducers.

\begin{lemma}\label{lem:reduction-lemma}
  Let $\Gamma$ be a finite set of types and let $k\in \set{1,2,\ldots}$. 
  \begin{enumerate}
    \item One step of the  reduction strategy, i.e.~the function
    \begin{align*}
      \reduce : \terms(\Gamma,k) \to   \terms(\Gamma,2k) 
    \end{align*}
    is recognised by a for-transducer. For terms without group products, a first-order for-transducer is enough.
    \item There exists some $m \in \set{1,2,\ldots}$ such that 
    \begin{align*}
      \overbrace{\reduce \circ \cdots \circ \reduce}^{\text{$m$ times}}(M)
    \end{align*}
    is in normal form for every $M \in \terms(\Gamma,k)$. 
  \end{enumerate}
\end{lemma}
\begin{proof}
  The first item is proved by formalizing the $\beta$-reduction rules  in a straightforward manner. We concentrate on the second item, about normalization in a bounded number of steps.  
  \begin{sublemma}\label{lem:norm-falls}
      Let $M$ be a term where the maximal degree of redexes is $d$. Then all redexes in $\reduce M$ have degree $\le d$ and 
      \begin{align*}
        |\reduce M|_d  < |M|_d
      \end{align*}
      where  $|M|_d$ is  the maximal  number of redexes of degree $d$ that can be found on root-to-leaf paths in the syntax tree of $M$. 
  \end{sublemma}
  \begin{proof}
    By inspecting the reduction rules in the definition of $\beta$-reduction. This inspection is easier by looking at pictures of the reduction rules as tree transformations, see  Appendix~\ref{sec:reduction-pictures}. 
  \end{proof}

  A corollary of the above sublemma is that if a term $M$ has depth $k$ and its maximal redexes have degree $d$, then $k$ applications of  $\reduce$ will eliminate all redexes of degree $d$, yielding a term  where all redexes have degree $< d$. Furthermore, each application of $\reduce$ at most doubles the depth of a term, and thus eliminating all redexes of degree $d$ in a term of depth $k$ leads to a term of depth at most $k \cdot 2^k$. 
  Putting these observations together, we see that a term normalises after a number of steps which depends only on the depth of the initial term and the maximal degree of  types appearing in it. This dependence is a tower of exponentials, with the height of the tower being the maximal degree, see  Section 3.7 in~\cite{sorensen2006lectures} for why this bound is tight.
\end{proof}

\pagebreak
\part{Algorithms}
This part is about algorithms for polyregular functions. It has only one section, Section~\ref{sec:evaluation}, which shows  that polyregular functions can be efficiently evaluated.  A future version might include an algorithm (or an undecidability proof) for the equivalence problem, which is left open for now:
\begin{description}
\item	{\bf Open Problem.} Is equivalence decidable for polyregular functions, i.e., can one decide if two given  polyregular functions produce the same outputs on every input?
\end{description}
\section{Evaluation algorithms}
\label{sec:evaluation}
In Section~\ref{sec:linear-time}, we show that for every polyregular function there is a linear time algorithm for computing its values, in the sense that the running time is linear in the combined input and output size. In Section~\ref{sec:aczero}, we focus on another aspect of efficient evaluation, namely parallelization. We show that every first-order polyregular function can be computed in \aczero, i.e., by circuits of polynomial size and constant depth (assuming that the output contains a special empty letter $\varepsilon$ which can be ignored).  

\subsection{Linear time}\label{sec:linear-time}
Polyregular functions can be evaluated in time linear in the combined input and output size. 
\begin{theorem}\label{thm:linear-time-evaluation}
	Every polyregular function can be evaluated in time 
	\begin{align*}
 \Oo(\text{length of input string} + \text{length of output string}).
\end{align*}
\end{theorem}

A natural idea for evaluating a polyregular function would be to simply implement algorithmically the semantics of any one of the devices (compositions of atomic functions, polynomial list functions, for-transducers, or pebble automata). Unfortunately, in each case,  the naive algorithm would have running time that is super-linear in the output size; the reason being that all of the  devices mentioned above  can do many steps before producing any output. Therefore, we need to run an optimisation in the algorithm such that subcomputations which do not produce any output are simulated in constant time. 

The rest of Section~\ref{sec:linear-time} is devoted to finding such an optimisation. 
Our proof shows a slightly stronger result, namely that after a precomputation that is linear in the input string, one can start producing the output string  with constant delay between positions. The ideas in the proof are closely based  constant delay enumeration algorithms for \mso on strings, see~\cite{bagan2006mso} and especially~\cite{Kazana:2013jq}. The only difference between Theorem~\ref{thm:linear-time-evaluation} and~\cite{bagan2006mso,Kazana:2013jq} is that in Theorem~\ref{thm:linear-time-evaluation}, we have to list tuples according to a given order, while the tuples in~\cite{bagan2006mso,Kazana:2013jq} can be listed in an order chosen by the algorithm (e.g.~lexicographic). This difference is not very important for the solution.

Let $f$ be a polyregular function, which is recognised by a $k$-pebble transducer $\Aa$. 
A configuration of the $k$-pebble transducer   is called \emph{productive} if its state outputs at least one letter. The general idea in the linear time algorithm from  Theorem~\ref{thm:linear-time-evaluation}
is to enumerate the productive configurations, with constant delay between two consecutive ones.  By Lemma~\ref{lem:all-pebble-mso}, for every states $p$ and $q$ and numbers $i,j \in \set{1,\ldots,k}$  there is an \mso formula 
\begin{align*}
\varphi^{ij}_{pq}(\underbrace{x_1,\ldots,x_i}_{\bar x}, \underbrace{y_1,\ldots,y_j}_{\bar y})	
\end{align*}
such that for every nonempty word $w$ over the input alphabet,
\begin{align*}
w \models \varphi^{ij}_{pq}(\bar x, \bar y)	
\end{align*}
holds if and only the automaton $\Aa$ has a run from configuration $p(\bar x)$ to configuration $q(\bar y)$. Using the above formulas, we can express in \mso that the automaton has a run from $p(\bar x)$ to $q(\bar y)$ which uses only non-productive configurations except for the source and target. In other words, the ``next productive configuration'' function can be defined in \mso.  The following lemma shows that, using a data structure that can be computed in linear time with respect to the input, one can evaluate in constant time every \mso definable partial function from tuples of positions to tuples of positions, in particular the ``next productive configuration'' function. Therefore, to prove Theorem~\ref{thm:linear-time-evaluation} it remains to prove the lemma.

\begin{lemma}\label{lem:compute-function-constant-time}
	Consider an \mso formula
\begin{align*}  
	\varphi(\overbrace{x_1,\ldots,x_i,y_1,\ldots,y_j}^{\text{first-order variables}})
\end{align*} 
which uses order $<$ and label predicates $a(\_)$ for $a \in \Sigma$. 
There is an algorithm which does the following for every input string $w \in \Sigma^*$.
\begin{enumerate}
	\item Computes a data structure in time linear in $|w|$;
	\item Using the data structure, answers the following queries in constant time: 
	\begin{itemize}
		\item {\bf Input.} A tuple of positions $\bar x = x_1,\ldots,x_i$ in $w$;
		\item {\bf Output.} The lexicographically least  tuple of positions $\bar y = y_1,\ldots,y_j$ such that $w \models \varphi(\bar x,\bar y)$, or an answer that no such tuple exists.
	\end{itemize}
\end{enumerate}
\end{lemma}
\begin{proof} The idea is to use factorisation forests and compositionality of \mso, in the same way as  in~\cite{colcombet2007combinatorial},~\cite{Kazana:2013jq}  or~\cite[Section 2.1]{bojanczyk2009factorization}. Since the proof is a minor adaptation of~\cite{colcombet2007combinatorial,bojanczyk2009factorization,Kazana:2013jq}, we only give a rough sketch.

	We assume without loss of generality that $j=1$, i.e.~the output tuple $\bar y$ consists of only one position. The case of $j>1$ can be easily reduced to the case of $j=1$, by computing successively the coordinates of the output tuple $\bar y$. 

	For $r \in \set{1,2,\ldots}$ define  $r$-equivalence  to be the equivalence relation on $\Sigma^*$ which identifies two words when they satisfy the same \mso sentences of quantifier rank at most $r$.  Compositionality for \mso says that  $r$-equivalence classes can be equipped with a monoid structure  so that mapping a word to its $r$-equivalence class becomes a  monoid homomorphism (see e.g.~\cite[proof of Lemma 4.1]{thomas1997languages}). Furthermore, for every position $y$ and every \mso formula $\psi(y)$  of quantifier rank at most $r$, whether or not $w \models \psi(y)$  depends only on the following information (an interval is a connected set of positions in a word):
		\mypic{44}
	A similar result is true for queries with more free variables.

	The data structure from item 1 in the lemma will be a factorisation as in the Factorisation Forest Theorem; we have already used a similar data structure in Section~\ref{sec:pebble-to-polyregular}. Let us begin with some terminology.  A \emph{factorisation} of an interval $I$ is a sequence of intervals $I_1,\ldots,I_m$ which partitions $I$, listed in left-to-right order. When talking about the $r$-equivalence class of an interval $I$ in a word $w$, we mean the $r$-equivalence class of the word $w[I]$ which is obtained from $w$ by keeping only positions from $I$.

	\paragraph*{The data structure.} We now define the data structure from item 1 of the lemma.  Define an \emph{$r$-factorisation} of an input word $w \in \Sigma^*$ to be a tree (an unranked tree with ordered siblings) with nodes labelled  by intervals in $w$ such that:
	\begin{enumerate}
		\item the root is labelled by the full interval containing all positions;
		\item leaves are labelled by intervals with one position only;
		\item if  a node has label $I$ and its children have labels $I_1,\ldots,I_m$  then
		\begin{enumerate}
			\item $m \ge 2$ and the intervals $I_1,\ldots,I_m$ are a factorisation of $I$;
			\item if $m > 2$, then $I_1,\ldots,I_m$  have the same  $r$-equivalence class, which is idempotent in the semigroup of $r$-equivalence classes\footnote{In Section~\ref{sec:pebble-to-polyregular} we did not assume idempotence, because we wanted the factorisation to be defined in first-order logic. Here we do not need first-order definability, which allows us to use idempotence. }.
		\end{enumerate}
	\end{enumerate}
	
	Let $r$ be the quantifier rank of the formula $\varphi$ in the assumption of the lemma. 
One can compute in linear time an  $(r+1)$-factorisation $t$ of $w$  whose  depth  is constant, i.e.~depends only on $r$ and not on $w$. The  algorithm that computes $t$  is implicit in the proof of the Factorisation Forest Theorem~\cite{simon1990factorization}, see also~\cite{bojanczyk2010efficient} for an explicit description of this algorithm and its more efficient versions. The $(r+1)$-factorisation $t$ is the data structure from item 1 in the statement of the lemma.   We also assume that each interval $I$ in the factorisation  is labelled by its $(r+1)$-equivalence class of $I$; these labels can be computed in linear time.

\newcommand{\firstposi}{\mathit{first}}
\newcommand{\lastposi}{\mathit{last}}
	
\paragraph*{Using the data structure.}
	It remains to show that the data structure $t$ can be used to answer in constant time the queries from item 2 in the statement of the lemma. The key step is given in the following sublemma. To get the result from the statement of the lemma, one applies a routine compositionality argument, see~\cite[Theorem 3.1]{Kazana:2013jq}. Therefore, we only prove the sublemma. 

	\begin{sublemma}
		Using $t$, the following  can be answered in constant time:
		\begin{itemize}
			\item {\bf Input.} An interval $I$ and  an \mso query $\psi(y)$ of quantifier rank $r$;
			\item {\bf Output.} The first position $y \in I$ such that $y$ satisfies $\psi$ in $w[I]$, or an answer that there is no such position.
		\end{itemize}	
	\end{sublemma}
	\begin{proof}
		Distinct nodes in $t$ are labelled by distinct intervals, and therefore  we can identify nodes in $t$ with the  intervals that label them. 
		Choose the node in the tree which represents an interval  $J  \subseteq I$  that contains $I$ and is smallest inclusion-wise  for this property. This node can be computed in constant time,  because the depth of the tree is constant. The algorithm works by induction on the  height of the  subtree of $J$.   If  $J$ is a leaf, then $w[I]$ has one letter only, and therefore one only needs to check if the unique position $y \in I$ satisfies $\psi(y)$ in $w[I]$, which can be done in constant time. 
		
		Otherwise $J$ has at least two children whose represent a factorisation
		\begin{align*}
			J = J_1 \cup \cdots \cup J_m.
		\end{align*}
		We only deal with the more  interesting case  of $m > 2$, in which case all of the intervals $J_1,\ldots,J_m$ have the same $(r+1)$-equivalence class $e$, which is idempotent.  Choose  $\firstposi, \lastposi \in \set{1,\ldots,m}$ so that $J_{\firstposi}$ contains the first position of $I$ and  $J_\lastposi$ contains the last position of $I$. By assumption on minimality of $J$, we have $\firstposi < \lastposi$. 
		The key observation is that if there is some position $y \in I$  which satisfies $\psi(y)$ in the word $w[I]$, then the leftmost position with this property  belongs to one of the intervals 
		\begin{align}\label{eq:four-intervals}
			J_{\firstposi},J_{\firstposi+1}, J_{\lastposi-1}, J_{\lastposi}.
		\end{align}
		This is because of the  assumption on the idempotence of the $(r+1)$-type $e$, which tells us when an interval contains at least one position satisfying a given rank $r$ query $\psi(y)$. (This is the same kind of observation as was used in Sublemma~\ref{claim:head-moves-little}.) Therefore, we can use the induction assumption to search for the  position $y$ in one of  the four intervals in~\eqref{eq:four-intervals}.
	\end{proof}
\end{proof}

We finish Section~\ref{sec:linear-time} with  some questions for future work. 
As mentioned at the beginning of the proof of Theorem~\ref{thm:linear-time-evaluation}, our proof shows that  after a precomputation that is linear in the input size, one can start producing the output word with constant delay between positions. This construction leaves some space for improvement. One improvement would be to lower the constants in the $\Oo$ notation. In the current algorithm, the constants in the linear time  are towers of exponentials in the size of the query, and therefore  the algorithm is not likely to be  practical. Maybe the constants can be made  polynomial in the size of a pebble automaton recognising the function? Another improvement would be to have random access to the output in the following sense:  after a precomputation (linear time, or maybe $n \log n$) one can  answer in constant time queries of the form ``give the $i$-th letter of the output  word''. Yet another direction is other kinds of access to the output word, e.g., pattern matching. These questions are  left for future work.

\subsection{ \aczero}
\label{sec:aczero}
This  section is about evaluating polyregular functions using \aczero  circuits (i.e., constant depth and polynomial size). There are two caveats: this can only be done for first-order polyregular functions, and the circuits are allowed to produced blank letters which do not count in the output. 
 The takeaway is that evaluation of first-order polyregular functions can be done efficiently in parallel. 
\paragraph*{Circuits defining string-to-string functions.} For definitions of circuits and the class \aczero, see~\cite[Section VIII]{straubing2012finite}. 
A circuit with $n$ input gates and $k$ output gates defines a Boolean function $2^n \to 2^k$. 

\begin{definition}[Functions in \aczero]\label{def:aczero}
	A function $f : 2^* \to 2^*$ is said to be in \aczero if there is  a family of constant depth and polynomial size circuits $\set{C_n}_{n \in \Nat}$ such that for input words $w$ of length $n$, the output  $f(w)$ is obtained by applying the circuit $C_n$ and reading the output gates.  We extend the definition of \aczero functions to functions of type $\Sigma^* \to \Gamma^*$, for finite alphabets $\Sigma,\Gamma$, by coding letters as bit strings of fixed length.
\end{definition}
Under the above  definition,  the length of the output is uniquely determined by  the length of the input (call such functions \emph{fixed output length}), because it is determined by the number of output gates in the  circuit appropriate to the input length. We want to use circuits to compute  functions where the output length is not determined by the input length (call such functions \emph{variable output length}), e.g.~the homomorphism
\begin{align*}
  h : \set{a,b}^* \to \set{a}^*
\end{align*}
which erases all $b$'s does not have fixed output length but is  polyregular.   To extend the definition of \aczero to functions of variable length, we use the class of functions which can be decomposed as 
\begin{align*}
  \xymatrix{\Sigma^* \ar[r]^f & (\Gamma+\varepsilon)^* \ar[r]^h & \Gamma^*}
\end{align*}
where $f$ is in \aczero as in Definition~\ref{def:aczero} (and therefore has fixed output length) and $h$ is the function that erases symbol $\varepsilon$. We use \homaczero to denote the resulting class of functions.
An equivalent definition of the class, which motivates its name, would be to consider compositions of an \aczero function followed by an arbitrary string-to-string homomorphism (i.e.~a homomorphism of free monoids), and not just the special homomorphism that erases $\varepsilon$.

Here is the main result of Section~\ref{sec:aczero}.
\begin{theorem}\label{thm:homacz}
Every first-order polyregular function is in \homaczero.
\end{theorem}

The assumption on first-order is crucial, because the product operation in the two-element group is polyregular (even sequential), but  not in \aczero, see~\cite[Theorem 4.3]{furst1984parity}. 

\begin{myexample}\label{ex:homaczero}
	An alternative idea for dealing with variable output lengths is to pad the output with a symbol $\#$. The resulting class, call it \emph{padding \aczero}, would not contain some polyregular functions.  Consider the homomorphism
	\begin{align*}
		h : \set{a,b}^* \to \set{a}^*
	\end{align*}
	which erases all letters $b$. This function is clearly polyregular. If $h$ were in padding \aczero,  then there would be a family of \aczero circuits, in the sense of Definition~\ref{def:aczero},   computing the function
\begin{align*}
	w \in \set{a,b}^* \quad \mapsto \quad  a^n\#^m 
\end{align*}
where $n$ is the number of $a$'s in $w$ and $m$ is the number of $b$'s.  By composing a circuit for the above function with a circuit that checks if the first $\#$ is on an even numbered position, we would get an \aczero  family of circuits for parity, contradicting~\cite[Theorem 4.3]{furst1984parity}. 
\end{myexample}

A natural idea for proving Theorem~\ref{thm:homacz}  would be to show that the class of functions \homaczero is closed under composition, and then show that the atomic first-order polyregular functions (first-order sequential functions, iterated reverse and squaring) are in \homaczero. Unfortunately,  \homaczero is not closed under composition, which can be shown  using the same reasoning as in Example~\ref{ex:homaczero}. Therefore, we need a different approach.

	We say that a sequential function has \emph{fixed block size} of there is some $c \in \set{0,1,2,\ldots}$ such that for every transition in the underlying automaton, the associated output word (see item~\ref{it:rational-output-words} in Definition~\ref{def:rational-function}) has length exactly $c$. In particular, for nonempty inputs, the output is exactly $c$ times longer than the input.
\begin{lemma}\label{lem:sweeper-decompose-block-size}
	Every first-order polyregular function can be decomposed as 
	\begin{align*}
  h \circ f_1 \circ \cdots \circ f_n
\end{align*}
where 
\begin{enumerate}
	\item $h$ is a homomorphism where letters are mapped to strings of length $\le 1$;
	\item each of $f_1,\ldots,f_n$ is squaring, iterated reverse,  or a first-order sequential function with fixed block size.
\end{enumerate}
\end{lemma}
\begin{proof}
Let us write $H$ for the functions as in the first item of the lemma, and $F$ for the functions as in the second item of the lemma. If $*$ denotes closure under composition, then statement of the lemma is that
\begin{align*}
	\text{polyregular} \subseteq H \circ F^*.
\end{align*}
(The converse inclusion is clearly seen to be true.) It is easy to see that  the atomic first-order polyregular functions are all in $H \circ F^*$.  Therefore, the prove the lemma, it remains to show that $H \circ F^*$  is closed under composition. Since both $F$ and $H^*$ are closed under composition,  it is enough to show
\begin{align*}
  F \circ H \subseteq H \circ F^*,
\end{align*}
which is left as an easy exercise for the reader.
\end{proof}

\begin{proof}[Proof of Theorem~\ref{thm:homacz}]
	Apply Lemma~\ref{lem:sweeper-decompose-block-size}, yielding a decomposition  
	\begin{align*}
  h \circ f_1 \circ \cdots \circ f_n.
\end{align*}
Since the class \aczero from Definition~\ref{def:aczero} is closed under composition, to prove the lemma, it suffices to show that all of the functions $f_1,\ldots,f_n$ are in the class \aczero.  In other words, we need to show that  the class \aczero contains  the squaring function, iterated reverse and also every first-order sequential function with fixed block size. For squaring, the circuit is easy to construct, and it has depth one since it essentially amounts to copying the input without really reading it. For first-order sequential functions with fixed block size $c$, we observe that the $i$-th bit of the output depends a first-order property of the first $i/c$ input letters, and such first-order properties can be computed in \aczero, see~\cite[Theorem IX.2.1]{straubing2012finite}. 

We are left with showing that iterated reverse is in \aczero. Without loss of generality we consider the case of two letters $\set{0,1}$ and a separator $\#$,  as in this example
\begin{align*}
	 011 \# 1000 \# 00111  \qquad \mapsto \qquad  110 \# 0001 \# 11100
\end{align*}
Suppose that the input has length $n$, and therefore also the output will have length $n$. For an output position $i \in \set{1,\ldots,n}$, the $i$-th output letter is:
\begin{enumerate}
	\item $\#$ if the  $i$-th  input letter is $\#$;
	\item $a \in \set{0,1}$ if there exist positions $j < i < k$ such that:
	\begin{enumerate}
		\item the input word has no separator strictly between positions $j$ and $k$;
		\item $j=0$ or the input word has a separator on position $j$;
		\item $k=n+1$ or the input word has a separator on position $k$;
		\item the $(k-i+j)$-th letter of the input is $a$.
	\end{enumerate}
	The above conditions can be formalised in an \aczero circuit.
\end{enumerate}

\end{proof}

\bibliographystyle{alpha}
\bibliography{bib}

\begin{thebibliography}{BGMP15}

\bibitem[A{\v{C}}11]{alur2011streaming}
Rajeev Alur and Pavol {\v{C}}ern{\`y}.
\newblock Streaming transducers for algorithmic verification of single-pass
  list-processing programs.
\newblock In {\em ACM SIGPLAN Notices}, volume~46, pages 599--610. ACM, 2011.

\bibitem[AFR14]{alur2014regular}
Rajeev Alur, Adam Freilich, and Mukund Raghothaman.
\newblock Regular combinators for string transformations.
\newblock In {\em Proceedings of the Joint Meeting of the Twenty-Third EACSL
  Annual Conference on Computer Science Logic (CSL) and the Twenty-Ninth Annual
  ACM/IEEE Symposium on Logic in Computer Science (LICS)}, page~9. ACM, 2014.

\bibitem[Alu10]{alur2010expressiveness}
Rajeev Alur.
\newblock Expressiveness of streaming string transducers.
\newblock 2010.

\bibitem[AU70]{aho1970characterization}
Alfred~V Aho and Jeffrey~D Ullman.
\newblock A characterization of two-way deterministic classes of languages.
\newblock {\em Journal of Computer and System Sciences}, 4(6):523--538, 1970.

\bibitem[Bag06]{bagan2006mso}
Guillaume Bagan.
\newblock Mso queries on tree decomposable structures are computable with
  linear delay.
\newblock In {\em International Workshop on Computer Science Logic}, pages
  167--181. Springer, 2006.

\bibitem[BC]{toolbox}
Miko{\l}aj Boja\'nczyk and Wojciech Czerwi\'nski.
\newblock Automata toolbox,
  https://www.mimuw.edu.pl/~bojan/upload/reduced-may-25.pdf.

\bibitem[BDK18]{DBLP:conf/lics/BojanczykDK18}
Miko{\l}aj Boja{\'{n}}czyk, Laure Daviaud, and Shankara~Narayanan Krishna.
\newblock Regular and first-order list functions.
\newblock In {\em Proceedings of the 33rd Annual {ACM/IEEE} Symposium on Logic
  in Computer Science, {LICS} 2018, Oxford, UK, July 09-12, 2018}, pages
  125--134, 2018.

\bibitem[Ber13]{berstel2013transductions}
Jean Berstel.
\newblock {\em Transductions and context-free languages}.
\newblock Springer-Verlag, 2013.

\bibitem[BGMP15]{baschenis2015one}
F{\'e}lix Baschenis, Olivier Gauwin, Anca Muscholl, and Gabriele Puppis.
\newblock One-way definability of sweeping transducers.
\newblock In {\em 35th IARCS Annual Conference on Foundations of Software
  Technology and Theoretical Computer Science (FSTTCS'15)}, 2015.

\bibitem[Boj09]{bojanczyk2009factorization}
Miko{\l}aj Boja{\'n}czyk.
\newblock Factorization forests.
\newblock In {\em International Conference on Developments in Language Theory},
  pages 1--17. Springer, 2009.

\bibitem[BP10]{bojanczyk2010efficient}
Miko{\l}aj Boja{\'n}czyk and Pawe{\l} Parys.
\newblock Efficient evaluation of nondeterministic automata using factorization
  forests.
\newblock In {\em International Colloquium on Automata, Languages, and
  Programming}, pages 515--526. Springer, 2010.

\bibitem[CD15]{Carton:2015bl}
Olivier Carton and Luc Dartois.
\newblock {Aperiodic Two-way Transducers and FO-Transductions.}
\newblock {\em CSL}, 2015.

\bibitem[CE12]{courcelle2012graph}
Bruno Courcelle and Joost Engelfriet.
\newblock {\em Graph structure and monadic second-order logic: a
  language-theoretic approach}, volume 138.
\newblock Cambridge University Press, 2012.

\bibitem[Cho77]{choffrut1977caracterisation}
Christian Choffrut.
\newblock Une caract{\'e}risation des fonctions s{\'e}quentielles et des
  fonctions sous-s{\'e}quentielles en tant que relations rationnelles.
\newblock {\em Theoretical Computer Science}, 5(3):325--337, 1977.

\bibitem[Cho79]{choffrut1979generalization}
Christian Choffrut.
\newblock A generalization of ginsburg and rose's characterization of gsm
  mappings.
\newblock In {\em International Colloquium on Automata, Languages, and
  Programming}, pages 88--103. Springer, 1979.

\bibitem[Cho03]{choffrut2003minimizing}
Christian Choffrut.
\newblock Minimizing subsequential transducers: a survey.
\newblock {\em Theoretical Computer Science}, 292(1):131--144, 2003.

\bibitem[CJ77]{chytil1977serial}
Michal~P Chytil and Vojt{\v{e}}ch J{\'a}kl.
\newblock Serial composition of 2-way finite-state transducers and simple
  programs on strings.
\newblock In {\em International Colloquium on Automata, Languages, and
  Programming}, pages 135--147. Springer, 1977.

\bibitem[Col07]{colcombet2007combinatorial}
Thomas Colcombet.
\newblock A combinatorial theorem for trees.
\newblock In {\em International Colloquium on Automata, Languages, and
  Programming}, pages 901--912. Springer, 2007.

\bibitem[DGK18]{DBLP:conf/lics/DaveGK18}
Vrunda Dave, Paul Gastin, and Shankara~Narayanan Krishna.
\newblock Regular transducer expressions for regular transformations.
\newblock In {\em Proceedings of the 33rd Annual {ACM/IEEE} Symposium on Logic
  in Computer Science, {LICS} 2018, Oxford, UK, July 09-12, 2018}, pages
  315--324, 2018.

\bibitem[EH01]{engelfriet2001mso}
Joost Engelfriet and Hendrik~Jan Hoogeboom.
\newblock Mso definable string transductions and two-way finite-state
  transducers.
\newblock {\em ACM Transactions on Computational Logic (TOCL)}, 2(2):216--254,
  2001.

\bibitem[Eil74]{eilenberg1974automata}
Samuel Eilenberg.
\newblock {\em Automata, languages, and machines}.
\newblock Academic press, 1974.

\bibitem[EM65]{elgot1965relations}
Calvin~C Elgot and Jorge~E Mezei.
\newblock On relations defined by generalized finite automata.
\newblock {\em IBM Journal of Research and Development}, 9(1):47--68, 1965.

\bibitem[EM02]{engelfriet2002two}
Joost Engelfriet and Sebastian Maneth.
\newblock Two-way finite state transducers with nested pebbles.
\newblock In {\em International Symposium on Mathematical Foundations of
  Computer Science}, pages 234--244. Springer, 2002.

\bibitem[Eng15]{engelfriet2015two}
Joost Engelfriet.
\newblock Two-way pebble transducers for partial functions and their
  composition.
\newblock {\em Acta Informatica}, 52(7-8):559--571, 2015.

\bibitem[FGL16]{filiot2016first}
Emmanuel Filiot, Olivier Gauwin, and Nathan Lhote.
\newblock First-order definability of rational transductions: An algebraic
  approach.
\newblock In {\em Proceedings of the 31st Annual ACM/IEEE Symposium on Logic in
  Computer Science}, pages 387--396. ACM, 2016.

\bibitem[FGRS13]{DBLP:conf/lics/FiliotGRS13}
Emmanuel Filiot, Olivier Gauwin, Pierre{-}Alain Reynier, and
  Fr{\'{e}}d{\'{e}}ric Servais.
\newblock From two-way to one-way finite state transducers.
\newblock In {\em 28th Annual {ACM/IEEE} Symposium on Logic in Computer
  Science, {LICS} 2013, New Orleans, LA, USA, June 25-28, 2013}, pages
  468--477, 2013.

\bibitem[FR16]{filiot2016transducers}
Emmanuel Filiot and Pierre-Alain Reynier.
\newblock Transducers, logic and algebra for functions of finite words.
\newblock {\em ACM SIGLOG News}, 3(3):4--19, 2016.

\bibitem[FSS84]{furst1984parity}
Merrick Furst, James~B Saxe, and Michael Sipser.
\newblock Parity, circuits, and the polynomial-time hierarchy.
\newblock {\em Mathematical systems theory}, 17(1):13--27, 1984.

\bibitem[GH96]{Globerman:1996he}
Noa Globerman and David Harel.
\newblock {Complexity Results for Two-Way and Multi-Pebble Automata and their
  Logics.}
\newblock {\em Theor. Comput. Sci.}, 169(2):161--184, 1996.

\bibitem[Gur82]{gurari1982equivalence}
Eitan~M Gurari.
\newblock The equivalence problem for deterministic two-way sequential
  transducers is decidable.
\newblock {\em SIAM Journal on Computing}, 11(3):448--452, 1982.

\bibitem[Hut99]{hutton1999tutorial}
Graham Hutton.
\newblock A tutorial on the universality and expressiveness of fold.
\newblock {\em Journal of Functional Programming}, 9(4):355--372, 1999.

\bibitem[Iba71]{Ibarra:1971ij}
Oscar~H Ibarra.
\newblock {Characterizations of some tape and time complexity classes of turing
  machines in terms of multihead and auxiliary stack automata}.
\newblock {\em Journal of Computer and System Sciences}, 5(2):88--117, April
  1971.

\bibitem[KR65]{krohn1965algebraic}
Kenneth Krohn and John Rhodes.
\newblock Algebraic theory of machines. i. prime decomposition theorem for
  finite semigroups and machines.
\newblock {\em Transactions of the American Mathematical Society},
  116:450--464, 1965.

\bibitem[KS13]{Kazana:2013jq}
Wojciech Kazana and Luc Segoufin.
\newblock {Enumeration of monadic second-order queries on trees.}
\newblock {\em ACM Trans. Comput. Log.}, 14(4):1--12, 2013.

\bibitem[Kuf08]{kufleitner2008height}
Manfred Kufleitner.
\newblock The height of factorization forests.
\newblock In {\em International Symposium on Mathematical Foundations of
  Computer Science}, pages 443--454. Springer, 2008.

\bibitem[LMSV01]{lautemann2001descriptive}
Clemens Lautemann, Pierre McKenzie, Thomas Schwentick, and Heribert Vollmer.
\newblock The descriptive complexity approach to logcfl.
\newblock {\em Journal of Computer and System Sciences}, 62(4):629--652, 2001.

\bibitem[MSV03]{milo2003typechecking}
Tova Milo, Dan Suciu, and Victor Vianu.
\newblock Typechecking for xml transformers.
\newblock {\em Journal of Computer and System Sciences}, 66(1):66--97, 2003.

\bibitem[RS91]{reutenauer1991minimization}
Christophe Reutenauer and Marcel-Paul Schutzenberger.
\newblock Minimization of rational word functions.
\newblock {\em SIAM Journal on Computing}, 20(4):669--685, 1991.

\bibitem[Sco67]{scott1967some}
Dana Scott.
\newblock Some definitional suggestions for automata theory.
\newblock {\em Journal of Computer and System Sciences}, 1(2):187--212, 1967.

\bibitem[Sim90]{simon1990factorization}
Imre Simon.
\newblock Factorization forests of finite height.
\newblock {\em Theoretical Computer Science}, 72(1):65--94, 1990.

\bibitem[Str12]{straubing2012finite}
Howard Straubing.
\newblock {\em Finite automata, formal logic, and circuit complexity}.
\newblock Springer Science \& Business Media, 2012.

\bibitem[Str18]{straubing-siglog}
Howard Straubing.
\newblock First-order logic and aperiodic languages: A revisionist history.
\newblock {\em {SIGLOG} News}, 6(1), 2018.

\bibitem[STT09]{sakarovitch2009elements}
Jacques Sakarovitch, Reuben Thomas, and Reuben Thomas.
\newblock {\em Elements of automata theory}, volume~6.
\newblock Cambridge University Press Cambridge, 2009.

\bibitem[SU06]{sorensen2006lectures}
Morten~Heine S{\o}rensen and Pawel Urzyczyn.
\newblock {\em Lectures on the Curry-Howard isomorphism}, volume 149.
\newblock Elsevier, 2006.

\bibitem[Tho97]{thomas1997languages}
Wolfgang Thomas.
\newblock Languages, automata, and logic.
\newblock In {\em Handbook of formal languages}, pages 389--455. Springer,
  1997.

\bibitem[Tra08]{trakhtenbrot2008logic}
Boris~A Trakhtenbrot.
\newblock From logic to theoretical computer science--an update.
\newblock In {\em Pillars of computer science}, pages 1--38. Springer, 2008.

\bibitem[UH67]{ullman1967approach}
JD~Ullman and JE~Hopcroft.
\newblock An approach to a unified theory of automata.
\newblock {\em Bell System Technical Journal}, 46(8):1793--1829, 1967.

\end{thebibliography}

\pagebreak
\part{Appendix}

\appendix
\section{Haskell Code}
\label{sec:haskell}
This section contains Haskell code for the polynomial list programs used in the paper. 

\subsection{Atomic polynomial list programs}
\label{sec:haskell-atomic}
{\small 
  \begin{verbatim}
data Errortype = Error

-- projections
first :: (a, b) -> a
first (x,y) = x
second :: (a, b) -> b
second (x,y) = y

-- coprojections are built into Haskell
-- Left :: a -> Either a b
-- Right :: b -> Either a b

-- case
cases :: (a1 -> b) -> (a2 -> b) -> Either a1 a2 -> b
cases f g (Left x) = f x 
cases f g (Right x) = g x

-- map is built into Haskell
-- map :: (a -> b) -> [a] -> [b]

-- concat is built into Haskell
-- concat :: [[a]] -> [a]

-- similar to Haskell head, but with an explicit error in the type
hd :: [a] -> Either  Errortype a
hd (h:_) = Right(h)
hd [] = Left(Error)

-- similar to Haskell tail, but with an explicit error in the type
tl :: [a] -> Either  Errortype [a]
tl (h:t) = Right t
tl [] = Left (Error)

-- input a list and return all ways of splitting it into two parts
split :: [a] -> [([a], [a])]
split [] = [([],[])]
split (h:t)  = (map (\(x,y) -> (h:x,y)) (split t)) ++ [([],h:t)]
\end{verbatim}
}

\subsection{Iterated reverse}
\label{sec:haskell-iterated-reverse}

{\small
\begin{verbatim}
-- like map, but only keeps the non-error values
errmap :: (a -> ( Either Errortype b)) -> [a] -> [b]
errmap f l = concat(map 
   (\e -> cases (\e -> []) (\e -> [e]) ( f e))  l)

-- we can use reverse then else because it could be coded as such:
-- reverse l =  errmap (\x -> hd (second x)) (split l)

-- empty test for lists
empty :: [a] -> Bool
empty l = cases (\x -> True) (\x -> False) (hd l)

-- errcatch: if there is an error, then replace it by some default value
errcatch :: (Either Errortype a) -> a -> a
errcatch a b = cases (\x -> b) id  a

-- is there some element on the list that satisfies f?
exists :: (a -> Bool) -> [a] -> Bool
exists f l = not (empty (filter f l))

-- do all elements of the list satisfy f?
forall :: (a -> Bool) -> [a] -> Bool
forall f l =  (empty (filter (not.f) l))

-- return the longest prefix of the list where all elements satisfy f
fprefix :: (a -> Bool) -> [a] -> [a]
fprefix f l = 
  let
  onlyfs = filter (forall f) (map first (split l))
  in
  errcatch (hd (map (filter f) onlyfs)) []

-- which part of the coproduct is used?
isleft :: Either a b -> Bool
isleft (Left a) = True
isleft _ = False

-- we can use if then else because it could be coded as such:
-- ifthenelse :: Bool -> t -> t -> t
-- ifthenelse c t e = if (c == True) then t else e

-- we can use filter, because it could be coded as such:
-- filter :: (a -> Bool) -> [a] -> [a]
-- filter f l = concat (map (\x -> ifthenelse (f x) [x] []) l)

-- blocks a list of coproducts into maximal blocks of same type
-- for example
-- [Left 1, Left 2, Right 'a', Right 'b', Left 3, Right 'c', Right 'd'] 
-- will be mapped to
-- [Left [1,2], Right ['a','b'], Left[3], Right['c','d']]
  
block l = 
  let
--f inputs a prefix/suffix pair, and outputs the sametype prefix of b (see below)
--if the last and first elements of a,b disagree on Left/Right type
  f (a,b) = 
    let
--sametype is the list prefix of b with same Left/Right type as first element
    sametype =  
     let
     leftprefix = 
       Left(
         concat (map (cases (\x -> [x]) (\x -> [])) (fprefix isleft b))
         )
     rightprefix =
       Right(
         concat (map (cases (\x -> []) (\x -> [x])) (fprefix (not.isleft) b))
         )
     in do
      {
      x <- hd b;
      return (cases (\_ -> leftprefix) (\_ -> rightprefix) x)
      }
    in
    if (empty a) then 
     return sametype
    else 
     do {
      lasta <- hd (reverse a);
      firstb <- hd b;
      if ((isleft firstb) && ((not.isleft) lasta)) then return sametype
      else if (((not.isleft) firstb) && (isleft lasta)) then return sametype 
      else Left Error
      } 
  in 
  reverse (errmap f (split l))
\end{verbatim}
}

\subsection{Squaring}
\label{sec:haskell-squaring}

{
\small
\begin{verbatim}
-- maps [1,2,3] to [([1,2,3],[]),([1,2],[3]),([1],[2,3]),([],[1,2,3])]
revsplit :: [a] -> [([a], [a])]
revsplit l = map (\x -> (reverse (second x), reverse (first x))) (split (reverse l))

-- for each element of the input list, 
-- output a triple (before the element, the element, after the element)
triples :: [t] -> [([t], (t, [t]))]
triples l = errmap 
 (\p -> do
  { h <- hd (second p);
  t <- tl (second p);
  return (first p, (h,t))
  }
 ) 
 (revsplit l)

-- squaring 
square :: [a] -> [Either a a]
square l = 
    let 
    reformat x =
     concat [map Left (first x), 
     [Right (first (second x))], 
     map Left (second (second x))]
    in
    concat (map reformat (triples l))\end{verbatim}
}

\section{Appendix on $\beta$-reduction}

\subsection{Values}
\label{sec:values}
In this part of the appendix, we analyse the shape of terms in normal form.
Define a \emph{function type} to be a type  where the outermost type constructor is $\to$, i.e.~the type has form $\tau \to \sigma$. Similarly, we talk about \emph{list types} $\tau^*$, \emph{product types} $\tau \times \sigma$, \emph{co-product} types $\tau + \sigma$ and \emph{finite set types}. The following lemma says that normal form terms  with non-function types and no free variables can only be built using their appropriate term constructors. A corollary of the lemma is that a normal form term with an arrow-free type and no free variables can only be built using elements of finite sets, pairing, co-pairing and lists. In particular, the lemma below implies Lemma~\ref{lem:normal-form-is-value}. 
\begin{lemma}\label{lem:values-non-functional}
    Let $M : \tau$ be a term in normal form without free variables. Then:
    \begin{enumerate}
        \item If $\tau$ is a finite set type, then $M$ is an element of $\tau$.
        \item If $\tau$ is a coproduct type, then $M$ is a coprojection of  a normal form term;
        \item If $\tau$ is a product type, then $M$ is a pair of normal form terms.
        \item If $\tau$ is a list type,  then $M$ is a list of normal form terms;
            \end{enumerate}
\end{lemma}

\begin{proof}
    Induction on the size of the term. Suppose that the  type $\tau$  is not a function type. 
Consider the leftmost branch of the syntax tree of $M$, and take the prefix of that path that uses only applications, leading to a decomposition of $M$ as  in the following picture:
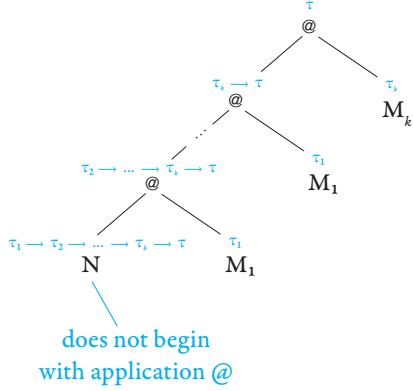
\mypic{45}
The number $k$ could be $0$, in the case when the outermost operation in $M$ is not an application. Consider the possibilities for the term $N$, according to the items in Definition~\ref{def:set-of-terms}.
\begin{enumerate}
	\item \emph{Variable.} This case cannot hold,  because $M$ has no free variables.
	\item \emph{Application.} This case cannot hold, by choice of $N$.
    \item \emph{Abstraction.} For $k=0$ this case cannot hold since otherwise  $\tau$ would be a function type. For $k>0$ this case cannot hold since otherwise there would be a redex.
	\item \emph{Element of a finite set.} In this case $k=0$ and item 1 in the  conclusion of the lemma is satisfied. 
	\item \emph{Pair of terms.} In this case $k=0$ and item 2 in the  conclusion of the lemma is satisfied. 
    \item \emph{Co-projection. } In this case $k=0$ and item 3 in the  conclusion of the lemma is satisfied. 
	\item \emph{List of terms.} In this case $k=0$ and item 4 in the  conclusion of the lemma is satisfied. 
	\item \label{it:normal-form-atomic}\emph{Atomic program.} This case cannot happen, which follows from  a case analysis on the atomic programs:
    \begin{itemize}
        \item[$\isterm$]  Since $\isterm$ takes an argument of finite set type, we must have $k \ge 1$ since otherwise $\tau$ would be a function type.  By induction assumption, $N_1$ is an element of the finite set in the input  type for $\isterm$. Therefore there is a redex,   contradicting the assumption on normal form.
        \item[$\projterm i$]  Since $\projterm i$ takes an argument of pair type, we must have $k \ge 1$ since otherwise $\tau$ would be a function type. By induction assumption, $N_1$ is a pair of terms. Therefore there is a redex,   contradicting the assumption on normal form.   
        \item[$\caseterm$] Since $\caseterm$ takes three arguments, we must have $k \ge 3$  because otherwise $\tau$ would be a function type. 
        In particular, the term $N_3$ is defined, and  by induction assumption it is a coprojection. Therefore there is a redex,  contradicting the assumption on normal form.
        \item  The remaining cases of $\mapterm$, $\headterm$, $\tailterm$, $\flatterm$ and $\splitterm$ are dealt with in the same way. 
    \end{itemize}

\end{enumerate}
    
\end{proof}

\subsection{Pictures of the reduction rules}
\label{sec:reduction-pictures}
This appendix contains pictures of the reduction rules from Definition~\ref{def:beta-reduction}.

\mypic{43}
\mypic{41}

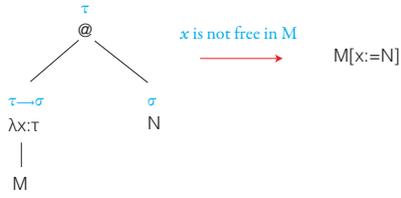
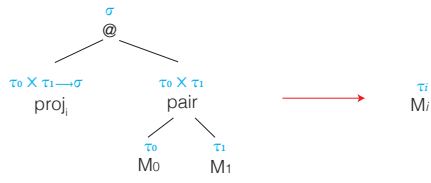
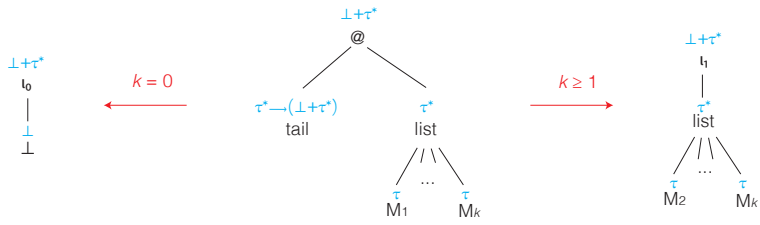
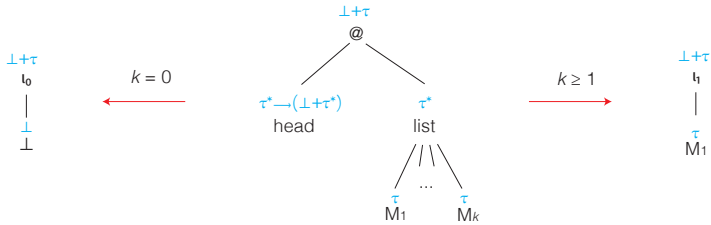
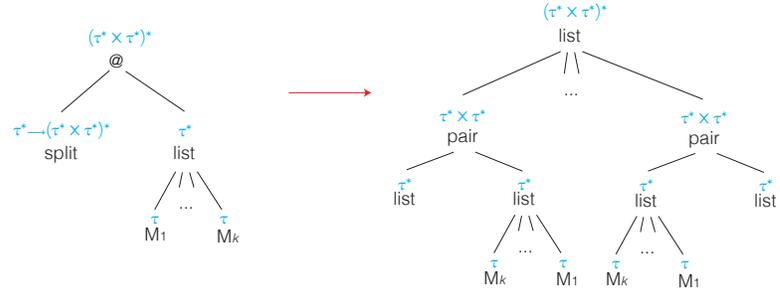
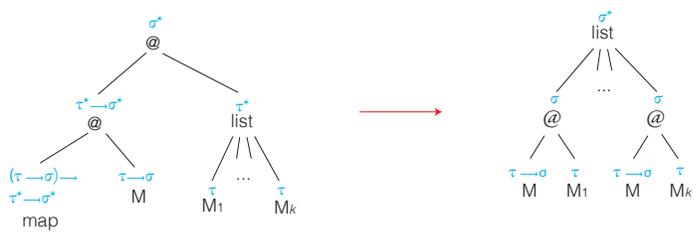
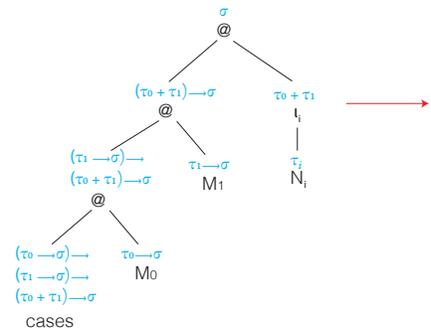
\end{document}